\documentclass[10pt,dvips,twoside,letterpaper]{article}
\usepackage{pslatex}
\usepackage{fancyhdr}
\usepackage{graphicx}
\usepackage{geometry}
\usepackage{amsmath}
\usepackage{amssymb}
\usepackage{amsfonts}
\usepackage{amsthm,amscd}

\def\figurename{Figure}
\makeatletter
\renewcommand{\fnum@figure}[1]{\figurename~\thefigure.}
\makeatother
\def\tablename{Table}
\makeatletter
\renewcommand{\fnum@table}[1]{\tablename~\thetable.}
\makeatother
\newtheorem{theorem}{Theorem}[section]

\newtheorem{proposition}[theorem]{Proposition}
\theoremstyle{definition}
\newtheorem{definition}[theorem]{Definition}

\theoremstyle{remark}
\newtheorem{remark}[theorem]{Remark}
\numberwithin{equation}{section}

\input{tcilatex}

\begin{document}

\title{\bfseries%
\scshape{Construction of initial data associated to the
characteristic initial value problem for the
Einstein-Yang-Mills-Higgs system}}
\author{\bfseries\scshape Calvin Tadmon$^{1,2}$\thanks{%
E-mail address: tadmonc@yahoo.fr; calvin.tadmon@up.ac.za} \\
%EndAName
$^{1}$ Department of Mathematics and Computer Science, \\
University of Dschang, P.\ O.\ Box 67, Dschang, Cameroon\\
$^{2}$ Department of Mathematics and Applied Mathematics,\\
University of Pretoria, Pretoria 0002, South Africa}
\date{}
\maketitle

\begin{abstract}
We show how to assign initial data for the characteristic
Einstein-Yang-Mills-Higgs system on two intersecting smooth null
hypersurfaces. We successfully adapt the hierarchical method set
up by A. D. Rendall to solve the same problem for the Einstein
equations in vacuum and with perfect fluid source. Unlike the work
of Rendall, many delicate calculations and expressions are given
in details so as to address, in a forthcoming work, the issue of
global resolution of the characteristic initial value problem for
the Einstein-Yang-Mills-Higgs system. The\ method\ obviously\
applies\ to\ the\ Einstein-Maxwell\ and\ the\ Einstein-scalar\
field\ models.
\end{abstract}

%\vskip 0.4in

% ------- [First Page Running Head] - place it immediately after title! ------
\thispagestyle{empty} \fancyhead{} \fancyfoot{} \renewcommand{%
\headrulewidth}{0pt}


\vspace{0.08in} \noindent \textbf{Keywords}:
Einstein-Yang-Mills-Higgs system, Characteristic initial value
problem, Null or characteristic hypersurfaces, constraints
problem, Harmonic gauge, Lorentz gauge.

\section{Introduction}

Before going straight to the point it is worth noting from the
outset that, following the publication of two recent joint works
\cite{14,15} with Marcel Dossa, some readers contacted us asking
for some supplementary explanations that may enable them
understand the resolution of the constraints problem associated to
the characteristic initial value problem for the
Einstein-Yang-Mills-Higgs system (EYMH). Instead of responding to
those readers alone, we have preferred to satisfy the whole
mathematical community by providing the complete details
concerning those aspects of the resolution of the constraints
problem associated to the characteristic initial value problem for
the EYMH system which were missing in \cite{14,15}. The
presentation is made in such a way that the paper, devoted to the
resolution of the constraints problem associated to the
characteristic initial value problem for the EYMH system on two
intersecting smooth null hypersurfaces, is almost self-contained.
The interests and physical motivations for studying characteristic
initial value problems have been widely mentioned in \cite{13, 16,
17, 20, 22, 23, 24, 25, 26, 27}. Among those interests and
physical motivations we can mention the\ following:

\begin{itemize}
\item In a certain way data on a null cone represent our present
knowledge of the cosmos better than data on spacelike
hypersurfaces and are closely related to observable quantities
(see \cite{16, 22, 24}).

\item In General Relativity (GR) where existence, non existence,
uniqueness play an essential role, specified hypersurfaces on
which data are posed are rather characteristic than spacelike
(horizons, null cones, past\ or\ future null infinity, etc\ldots ,
see \cite{13, 20, 23}).

\item The characteristic problem is quite natural in the case of
shock waves where discontinuities necessarily appear on
characteristics (see \cite{17, 22}).

\item Characteristic data have the important advantage that there
exists a set of quantities that can be given independently of each
other and that determine the solution uniquely. (see \cite{22,
24})

\item A lot of quantities have to be obtained by solving transport
equations along the bicharacteristics: For second order systems,
the first derivatives are not given on the initial null
hypersurfaces, they are obtained through transport equations. (see
\cite{23})

\item In the case of the Einstein equations, the constraints for
the gravitational field as well as for the harmonic coordinates
reduce to explicit ordinary differential equations (ODEs) instead
of the elliptic partial differential equations (PDEs) which arise
when the classical Cauchy problem (where initial data are assigned
on a spacelike hypersurface) is considered (see \cite{24, 25}). So
one interesting application of the characteristic initial value
problem is in numerical GR (see \cite{26, 27}).
\end{itemize}

As for the EYMH system, it is a mathematically and physically
interesting model in GR and Gauge field Theory. The EYMH equations
form a basis for the well-known self-consistent model of
interaction of gravitational, gauge and scalar fields (see
\cite{1, 29} for review and basic references). The
Einstein-Yang-Mills-Higgs theory unifies two important trends in
the theory of gravity. The first one is the Einstein-Yang-Mills
model, which is a non-abelian generalization of the well known
Einstein-Maxwell model of relativistic electromagnetism theory.
The second trend is connected with the investigations of
interaction of gravitational and scalar fields. The importance of
the EYMH model in the theory of dark matter an in the theories of
dark energy has been underlined by A.\ B.\ Balakin et al. \cite{1}
and M.S.\ Volkov and D. V.\ Gal'tsov \cite{29}.

It is well known that, under a suitable choice of gauge
conditions, the initial value problem for the EYMH system splits
into two parts called the evolution problem and the constraints
problem. Throughout all the work we will use harmonic gauge for
the gravitational field and Lorentz gauge for the Yang-Mills
potential.

Some important works on characteristic initial value problems with
initial
data prescribed on two intersecting null hypersurfaces can be found in \cite%
{2, 3, 4, 5, 6, 10, 11, 12, 13, 14, 15, 19, 20, 22, 23, 24, 25}.

When attempting to solve the constraints problem for the
characteristic EYMH system by the hierarchical method of Rendall
(see \cite{24, 25}), some obstacles occur due to the form of\ the
stress-energy tensor which is more complex compared to the vacuum
Einstein equations or the Einstein-perfect fluid cases. The
feature of the present work resides in the fact that we have
patiently overcome all these difficulties by judiciously
performing very tedious calculations. This seems to be, to our
knowledge, a true advance and in all points of view, new.
Moreover, unlike the work of Rendall \cite{24, 25}, many delicate
calculations and expressions are given in details in such a way
that one can foresee promising resolution of the global
characteristic EYMH system by adapting recent methods developed by
H. Lindblad and I. Rodnianski \cite{21} and C.\ Svedberg
\cite{28}\ to study ordinary spatial Cauchy problems for Einstein
equations in vacuum and Einstein-Maxwell-Scalar field system
respectively.

Since the reduced EYMH equations stand as a second order
quasilinear hyperbolic system, the resolution of the corresponding
evolution problem can obviously be achieved directly by just
referring to Theorem 1 of the paper of A.\ D.\ Rendall \cite{24}.
So, in the present work, we mainly focus our attention on the
resolution of the constraints problem associated to the EYMH
system.

The paper is organized as follows. In section 2, we give some
preliminaries about the EYMH system. The complete form as well as
the reduced form of the EYMH system, suitable\ for the resolution
of the constraints problem, are given explicitly under harmonic
gauge and Lorentz gauge conditions. The concern of section 3 is
the resolution of the constraints problem for the characteristic
EYMH system i.e., the construction, from arbitrary choice of some
components of the unknown function (called free data) on the
initial null hypersurfaces, of the complete initial data for the
reduced EYMH system such that the harmonic gauge and the Lorentz
gauge conditions are satisfied
on the initial null hypersurfaces. For sake of simplicity, only the case of $%
C^{\infty }$ data will be discussed. Data of finite
differentiability order may be constructed in Sobolev type spaces
using energy inequalities and other classical tools as described
in \cite{13, 14, 15, 23, 24}. The conclusion of the work is given
in section 4 where compatibility conditions are discussed and the
final result stated. In appendices A to D we provide the proofs
that were missing in \cite{15} (according to some readers as noted
at the beginning the lack of these important ingredients, among
others, rendered the work in \cite{15} difficult to understand).

\section{The Einstein-Yang-Mills-Higgs system}

% ------------ [Running Heads - for odd and even pages] - please insert it only on page 2!
\pagestyle{fancy} \fancyhead{} \fancyhead[EC]{C. Tadmon} \fancyhead[EL,OR]{%
\thepage} \fancyhead[OC]{Initial data for the characteristic EYMH
system} \fancyfoot{}
\renewcommand\headrulewidth{0.5pt}

We first introduce some more comprehensible geometric tools that
are necessary for the good understanding of the deep structure of
the EYMH equations. The framework where the constraints problem
for the EYMH equations will be solved is presented. Throughout all
the section, the notations of \cite{23} are used.

\subsection{Geometric tools and notations}

$L$ denotes a compact domain of $\mathbb{R}^{4}$ with a piecewise
smooth boundary $\partial L,$ $G^{1}$ and $G^{2}$ are two
$3$-dimensional surfaces
such that $G^{\omega }\subset L$ for $\omega =1,$ $2$. We assume that $%
G^{\omega }$ are defined by
\begin{equation}
G^{\omega }=\left\{ x\in L:x^{\omega }=0\right\} ,\qquad \omega =1,\text{ }2,%
\text{ }  \tag{2.1}
\end{equation}%
where $x=\left( x^{a}\right) =\left( x^{1},...,x^{4}\right) $ is
the global canonical coordinates system of $\mathbb{R}^{4}$. In
addition we assume that $G^{1}\cup G^{2}\subset \partial L$. Set
\begin{equation}
\tau \left( x\right) =x^{1}+x^{2}\text{ and }T_{0}=\underset{x\in L}{\sup }%
\tau \left( x\right) .  \tag{2.2}
\end{equation}%
For $t\in \left[ 0,T_{0}\right] $, define the following point sets
\begin{equation}
L_{t}=\left\{ x\in L:0\leq \tau \left( x\right) \leq t\right\} ,\text{\quad }%
G_{t}^{\omega }=\left\{ x\in G^{\omega }:0\leq \tau \left(
x\right) \leq t\right\} .  \tag{2.3}
\end{equation}

\begin{remark}
\textit{The initial data will be constructed on} $G_{T}^{1}\cup
G_{T}^{2},$ \textit{for} $T\in (0,T_{0}].$
\end{remark}

The basic geometric framework is a space-time $\left(
\mathcal{M},g\right) $ i.e., a $4-$dimensional manifold
$\mathcal{M}$ equipped with a lorentzian metric $g$ of signature
$-+++$.

A Yang-Mills potential is usually represented by a $1-$form $A$ defined on $%
\mathcal{M}$ with values in the Lie algebra $\mathcal{G}$ of a Lie
group $G.$ We assume\ that the\ Lie\ group\ $G$\ admits\ a\
non-degenerate\ bi-invariant\ metric (it is the case if $G$ is the
product of abelian and semi-simple groups (see \cite{8, 9})).\ The
Lie algebra $\mathcal{G}$ admits then an $Ad$-invariant non
degenerate scalar product, denoted by a dot \textquotedblleft
$.$\textquotedblright , which enjoys the following
property:%
\begin{equation}
f.\left[ k,l\right] =\left[ f,k\right] .l,\text{ }\forall f,k,l\in \mathcal{G%
}.  \tag{2.4}  \label{66}
\end{equation}%
Here $\left[ ,\right] $ denote the Lie brackets of the Lie algebra $\mathcal{%
G}$. It is assumed that $\mathcal{G}$ is an $N$-dimensional $\mathbb{R}$%
-based Lie algebra. For simplicity $\left( x^{i}\right)
_{i=1,...,4}$ also denote the local coordinates in $\mathcal{M}$
and $\left( \varepsilon _{I}\right) _{I=1,...,N}$ denotes an
orthogonal basis of $\mathcal{G}$. Then the Yang-Mills potential
is locally defined by the following equality:

\begin{equation}
A=A_{i}^{I}dx^{i}\otimes \varepsilon _{I},\text{ with }A_{i}^{I}:\mathcal{M}%
\rightarrow \mathbb{R}\text{.}  \tag{2.5}  \label{67}
\end{equation}%
The Yang-Mills field is the curvature of the Yang-Mills potential.
It is
represented by a $\mathcal{G}$-valued antisymmetric $2-$form $F$ defined on $%
\mathcal{M}$ by the following equality:%
\begin{equation}
F=dA+\frac{1}{2}\left[ A,A\right] .  \tag{2.6a}  \label{68a}
\end{equation}%
In the local coordinates $\left( x^{i}\right) $ and basis $\left(
\varepsilon _{I}\right) $ the above equality $\left(
\ref{68a}\right) $
reads:%
\begin{equation}
F_{ij}^{I}=\nabla _{i}A_{j}^{I}-\nabla _{j}A_{i}^{I}+\left[ A_{i},A_{j}%
\right] ^{I}=\nabla _{i}A_{j}^{I}-\nabla
_{j}A_{i}^{I}+C_{JK}^{I}A_{i}^{J}A_{j}^{K},  \tag{2.6b}
\label{68b}
\end{equation}%
or in the summary form%
\begin{equation}
F_{ij}=\nabla _{i}A_{j}-\nabla _{j}A_{i}+\left[ A_{i},A_{j}\right]
, \tag{2.6c}  \label{68c}
\end{equation}%
where $\nabla $ denotes the covariant derivative w.r.t. the
space-time metric, and $C_{JK}^{I}$ are the structure constants of
the Lie group $G$.

In addition to the Yang-Mills field, many physical theories
consider a Higgs field or a scalar multiplet which is represented
by a $\mathcal{G}$-valued function $\Phi $ defined on
$\mathcal{M}.$ In the local basis $\left(
\varepsilon _{I}\right) ,$ $\Phi $ is defined as follows:%
\begin{equation}
\Phi =\Phi ^{I}\varepsilon _{I},\text{ with }\Phi ^{I}:\mathcal{M}%
\rightarrow \mathbb{R}.  \tag{2.7}  \label{69}
\end{equation}

\subsection{The EYMH equations}

Throughout the paper Roman indices vary from $1$ to $4$ and the
standard
convention of summing over repeated indices is used i.e., $u_{i}v^{i}=%
\underset{i=1}{\overset{4}{\sum }}u_{i}v^{i}$. Comma denotes
partial
derivative i.e., $u_{,i}=\frac{\partial u}{\partial x^{i}}$. We also denote $%
Du=\left( \frac{\partial u}{\partial x^{i}}\right) _{i=1,...,4}$.
For sake of clarity, we omit the multiplicative physical constant
that usually
appears beside the stress-energy tensor and take the cosmological constant $%
\Lambda =0$. In the local coordinates $\left( x^{i}\right) $ on $\mathcal{M}$%
, the EYMH system reads as follows (see \cite{1, 8, 14, 15, 29})%
\begin{equation}
\begin{array}{l}
R_{ij}-\frac{1}{2}Rg_{ij}=\rho _{ij}\text{, } \\
\widehat{\nabla }_{i}F^{ij}=J^{j}\text{, } \\
\widehat{\nabla }_{i}\widehat{\nabla }^{i}\Phi =H\text{. }%
\end{array}
\tag{2.8}  \label{215}
\end{equation}%
Recall that $\left( R_{ij}\right) $ and $R$ are respectively the
Ricci tensor and the scalar curvature of the metric $g$ i.e.,

\begin{equation}
R_{ij}=\Gamma _{ij,k}^{k}-\Gamma _{ik,j}^{k}+\Gamma
_{kl}^{k}\Gamma _{ij}^{l}-\Gamma _{jl}^{k}\Gamma _{ik}^{l},\quad
R=g^{ij}R_{ij},  \tag{2.9} \label{217}
\end{equation}%
where $\Gamma _{ij}^{k}\ $are the Christofell symbols relative to
the space-time metric $g$. $F^{ij}$ are the contravariant
components of the Yang-Mills field (here and throughout the
section, indices are raised or lowered w.r.t. the space-time
metric $g$ i.e., $F^{ij}=g^{ia}g^{jb}F_{ab}$). $\Phi $ is the
Higgs field. $\rho _{ij}$ is the energy-momentum or the
stress-energy tensor which is defined\ by

\begin{equation}
\rho _{ij}=F_{ik}.F_{j}^{\quad
k}-\frac{1}{4}g_{ij}F_{kl}.F^{kl}+\Phi _{ij}, \tag{2.10a}
\label{216}
\end{equation}%
where

\begin{equation}
\Phi _{ij}=\widehat{\nabla }_{i}\Phi .\widehat{\nabla }_{j}\Phi -\frac{1}{2}%
g_{ij}\left( \widehat{\nabla }_{k}\Phi .\widehat{\nabla }^{k}\Phi
+V\left( \Phi ^{2}\right) \right) ,  \tag{2.10b}  \label{216a}
\end{equation}%
with
\begin{equation}
\Phi ^{2}=\Phi .\Phi ,  \tag{2.10c}  \label{216b}
\end{equation}%
$V\ $is a $C^{\infty }$ real valued function defined on
$\mathbb{R}$ (often called the self interaction potential).
$J^{k}$\ is the Yang-Mills current
defined\ by%
\begin{equation}
J^{k}\left( A,\Phi ,D\Phi \right) =\left[ \Phi ,\widehat{\nabla }^{k}\Phi %
\right] .  \tag{2.10d}  \label{216d}
\end{equation}%
$\widehat{\nabla }$ is the gauge covariant derivative or the
Yang-Mills operator; it acts on $\Phi $ and $F^{ij}$ as follows

\begin{equation}
\widehat{\nabla }_{i}\Phi =\nabla _{i}\Phi +\left[ A_{i},\Phi
\right] ,\quad \widehat{\nabla }_{i}F^{ij}=\nabla
_{i}F^{ij}+\left[ A_{i},F^{ij}\right] . \tag{2.11}  \label{219}
\end{equation}%
$H\left( \Phi \right) $ is the Higgs potential; it is a $C^{\infty }$\ $%
\mathcal{G}$-valued function given\ by (see \cite{8}):%
\begin{equation}
H^{I}\left( \Phi \right) =V^{\prime }\left( \Phi ^{2}\right) \Phi
^{I}, \tag{2.12}  \label{220}
\end{equation}%
where $V^{\prime }$ is the derivative of $V.$

\begin{remark}
\textbf{\ }$\left( i\right) $ \textit{Due to the following
consequences of
Bianchi identities }%
\begin{equation*}
\nabla ^{i}\left( R_{ij}-\frac{1}{2}Rg_{ij}\right) =0,\quad \widehat{\nabla }%
_{i}\widehat{\nabla }_{j}F^{ij}=0,
\end{equation*}%
\textit{it is easy to see that: if the EYMH system }$\left(
\ref{215}\right)
$\textit{\ is satisfied, then the stress-energy tensor }$\rho _{ij}$\textit{%
\ and the current }$J^{a}$\textit{\ satisfy the following conservation laws}%
\begin{equation}
\nabla _{i}\rho ^{ij}=0,\quad \widehat{\nabla }_{i}J^{i}=0.
\tag{2.13} \label{221}
\end{equation}%
$\left( ii\right) $\textit{\ Due to the expression }$\left( \ref{216d}%
\right) $\textit{\ of} $J^{a},$\textit{\ the Higgs potential
}$H$\textit{\
must satisfy the following algebraic structural condition}%
\begin{equation}
\left[ H\left( \Phi \right) ,\Phi \right] =0.  \tag{2.14}
\end{equation}%
\textit{The relation }$\left( 2.14\right) $\textit{\ is fulfilled by }$%
H\left( \Phi \right) $\textit{\ given by }$\left( \ref{220}\right)
.$

$\left( iii\right) $ \textit{If the YMH system\ is satisfied then
a direct
calculation shows that the conservation laws }$\left( \ref{221}\right) $%
\textit{\ are fulfilled by the stress-energy tensor given by }$\left( \ref%
{216}\right) $\textit{\ and the current given by }$\left(
\ref{216d}\right)
, $ \textit{for }$H\left( \Phi \right) $\textit{\ given by }$\left( \ref{220}%
\right) $. It results that\textit{\ the EYMH system is
coherent}$.$

$\left( iv\right) $ \textit{The proof of} $\left( iii\right) $
\textit{is provided in Appendix A}.
\end{remark}

\subsection{\textbf{The reduced EYMH system}}

It is a well known fact that system $\left( \ref{215}\right) $ is
not an evolution system as it stands. In order to reduce it to an
evolution system, one needs to impose to the unknown functions
(the components of the unknown metric and those of the unknown
Yang-Mills potential) some supplementary conditions called gauge
conditions or to choose a special or preferred system of
coordinates. In the present paper we will use the Lorentz gauge
condition and the harmonic coordinates which were historically the
first special coordinates (e.g., in 1952, Y.\ Choquet-Bruhat
\cite{7} used these special coordinates to prove the local
well-posedness of the vacuum Einstein equations).

\begin{definition}
Let $\left( x^{i}\right) _{i=1,...,4}$ be local coordinates on a
$4-d$ manifold $\mathcal{M}$ endowed with a Lorentzian metric $g.$
$\left( x^{i}\right) _{i=1,...,4}$ are called harmonic coordinates
if they satisfy the following equation
\begin{equation}
\square _{g}x^{i}=0,\text{ }i=1,2,3,4,  \tag{2.15}  \label{2}
\end{equation}%
where $\square _{g}=\nabla _{k}\nabla ^{k}$ is the geometric wave operator, $%
\nabla $ representing the covariant derivative relative to the
metric $g.$
\end{definition}

Let $\left( x^{i}\right) _{i=1,...,4}$ be local coordinates on a
$4-d$ manifold $\mathcal{M}$ equipped with a Lorentzian metric
$g$. Recall the definition of the Christoffel symbols
\begin{equation}
\Gamma _{ij}^{k}=\frac{1}{2}g^{km}\left(
g_{mj,i}+g_{mi,j}-g_{ij,m}\right) , \tag{2.16}  \label{3}
\end{equation}%
and set
\begin{equation}
\Gamma ^{k}=g^{ij}\Gamma _{ij}^{k},  \tag{2.17}  \label{4}
\end{equation}%
where $g^{ij}$ denotes the inverse of $g_{ij}$ i.e.,
\begin{equation*}
g^{ij}g_{ik}=\delta _{k}^{j}=\left\{
\begin{array}{c}
1\text{ if }j=k, \\
0\text{ if }j\neq k.%
\end{array}%
\right.
\end{equation*}%
It is easy to see by a simple calculation that the local
coordinates $\left( x^{i}\right) _{i=1,...,4}$ are harmonic if and
only if the metric $g$
satisfy the following so-called harmonic gauge condition%
\begin{equation}
\Gamma ^{k}=0,\text{ }\forall k=1,...,4.  \tag{2.18}  \label{5}
\end{equation}%
The following different equivalent forms of the harmonic gauge condition $%
\left( \ref{5}\right) $ will be used judiciously:%
\begin{equation}
\Gamma ^{k}=0\Leftrightarrow g^{lm}g_{nm,l}=\frac{1}{2}g^{lm}g_{lm,n}%
\Leftrightarrow g_{,l}^{lk}=-\frac{1}{2}g^{kn}g^{lm}g_{lm,n}=\frac{1}{2}%
g^{kn}g_{lm}g_{,n}^{lm}.  \tag{2.19}  \label{6}
\end{equation}%
We now define the Lorentz gauge condition.

\begin{definition}
Relative to the coordinates $\left( x^{i}\right) _{i=1,...,4},$
the
Yang-Mills potential $A$ satisfies the Lorentz gauge condition if%
\begin{equation}
\Delta \equiv \nabla _{i}A^{i}=0.  \tag{2.20}
\end{equation}
\end{definition}

Let us now consider the complete Einstein-Yang-Mills-Higgs equations $\left( %
\ref{215}\right) $ in arbitrary local coordinates $\left(
x^{i}\right) _{i=1,...,4}$. The following proposition provides the
reduction of the EYMH equations, with unknowns $\left(
g_{ij},A_{p},\Phi \right) ,$ modulo the harmonic gauge and the
Lorentz gauge conditions.

\begin{proposition}
Let $\left( g_{ij},A_{p},\Phi \right) $ be such that the complete
EYMH equations $\left( \ref{215}\right) $ are satisfied together
with the harmonic gauge condition $\left( \ref{5}\right) $ and the
Lorentz gauge condition $\left( 2.20\right) $. Then $\left(
g_{ij},A_{p},\Phi \right) $
solves the following system of reduced EYMH equations:%
\begin{equation}
\begin{array}{l}
\widetilde{R}_{ij}=\tau _{ij}\left( g,A,\Phi ,Dg,DA,D\Phi \right) , \\
LA_{p}=J_{p}\left( A,\Phi ,D\Phi \right) , \\
S\Phi =H\left( \Phi \right) ,%
\end{array}
\tag{2.21}
\end{equation}%
where%
\begin{equation}
\begin{array}{l}
\widetilde{R}_{ij}\equiv R_{ij}-\frac{1}{2}\left( g_{ki}\Gamma
_{,j}^{k}+g_{kj}\Gamma _{,i}^{k}\right) \\
\text{ \ \ \ \ \ \ }=-\frac{1}{2}g^{km}g_{ij,mk}+Q_{ij}(g,Dg), \\
\tau _{ij}=F_{ik}F_{j}^{\quad k}-\frac{1}{4}g_{ij}F_{kl}F^{kl}+\widehat{%
\nabla }_{i}\Phi .\widehat{\nabla }_{j}\Phi
+\frac{1}{2}g_{ij}V\left( \Phi
^{2}\right) , \\
LA_{p}\equiv g_{jp}\widehat{\nabla }_{i}F^{ij}+\left( \Delta
_{,p}+\Gamma
_{,p}^{l}A_{l}+\Gamma ^{l}A_{l,p}\right) \\
\text{ \ \ \ \ \ \
}=g^{ik}A_{p,ik}+g_{,p}^{ki}A_{k,i}+g^{ik}\left[
A_{k},A_{p}\right] _{,i} \\
\text{ \ \ \ \ \ \ }+g_{jp}\left( g^{ik}g^{jl}\right) _{,i}\left[
A_{l,k}-A_{k,l}+\left[ A_{k},A_{l}\right] \right] \\
\text{ \ \ \ \ \ \ }+g_{jp}\Gamma _{im}^{i}F^{mj}+g_{jp}\Gamma
_{im}^{j}F^{im}+g_{jp}\left[ A_{i},F^{ij}\right] , \\
S\Phi \equiv \widehat{\nabla }_{i}\widehat{\nabla }^{i}\Phi
+\Gamma ^{l}\Phi
_{,l}-\left[ \Delta ,\Phi \right] \\
\text{ \ \ \ \ \ \ }=g^{ij}\Phi _{,ij}+2\left[ A_{i},\nabla
^{i}\Phi \right]
+\left[ A_{i},\left[ A^{i},\Phi \right] \right] .%
\end{array}
\tag{2.22}
\end{equation}%
Here $Q_{ij}$\ is a rational function of its arguments depending
quadratically on $Dg$, given\ by%
\begin{equation}
\begin{array}{l}
Q_{ij}\left( g,Dg\right) =\frac{1}{2}\left(
g_{ki,j}+g_{kj,i}\right) \Gamma ^{k}+\frac{1}{2}g^{km}g^{nl}\left(
g_{nk,j}g_{im,l}+g_{nk,i}g_{jm,l}\right)
\\
\text{ \ \ \ \ \ \ \ \ \ \ \ \ \ \ \ }-\frac{1}{4}%
g^{km}g^{nl}g_{kn,i}g_{lm,j}-\frac{1}{2}g^{km}g^{nl}g_{mn,k}\left(
g_{lj,i}+g_{li,j}-g_{ij,l}\right) \\
\text{ \ \ \ \ \ \ \ \ \ \ \ \ \ \
}+\frac{1}{4}g^{km}g^{nl}g_{km,l}\left(
g_{in,j}+g_{jn,i}-g_{ij,n}\right)
-\frac{1}{2}g^{km}g^{nl}g_{ki,n}\left(
g_{lj,m}-g_{mj,l}\right) .%
\end{array}
\tag{2.23}
\end{equation}
\end{proposition}

\begin{proof}
See \cite{15}.
\end{proof}

\begin{remark}
$\left( i\right) $\textit{\ Due to }$\left( 2.22\right)
,$\textit{\ any solution }$\left( g_{ij},A_{p},\Phi \right)
$\textit{\ of the reduced EYMH
system }$\left( 2.21\right) $\textit{\ that satisfies the constraints }$%
\Gamma ^{k}\equiv g^{ij}\Gamma _{ij}^{k}=0$\textit{\ and }$\Delta
\equiv
\nabla _{i}A^{i}=0$\textit{\ is also a solution of the complete EYMH system }%
$\left( \ref{215}\right) $.

$\left( ii\right) $\textit{\ For the constraints }$\Gamma
^{k}=0$\textit{\ and }$\Delta =0$\textit{\ to be satisfied
everywhere, it is enough that they
are satisfied on }$G^{1}\cup G^{2}$\textit{\ }$\left( \text{see \cite{18}}%
\right) $\textit{: one uses the Bianchi identities\ to show that\
}$\left( \Gamma ^{k},\Delta \right) $\textit{\ solves a second
order homogeneous linear system.}

$\left( iii\right) $\textit{\ The reduced EYMH system }$\left( 2.21\right) $%
\textit{\ constitutes the evolution system associated to the EYMH system }$%
\left( \ref{215}\right) $\textit{.}

$\left( iv\right) $\textit{\ The resolution of the constraints
problem consists in constructing, from arbitrary choice of some
components of the
gravitational potentials and Yang-Mills potentials (called free data) on }$%
G^{1}\cup G^{2}$\textit{, of all initial data for the reduced EYMH
such that the constraints }$\Gamma ^{k}=0$\textit{\ and }$\Delta
=0$\textit{\ are satisfied on }$G^{1}\cup G^{2}$ \textit{for the
solution of the corresponding evolution problem}.
\end{remark}

\section{The constraints problem for the characteristic EYMH system}

The goal here is to construct $C^{\infty }$ initial data for the
reduced EYMH system such that the constraints $\Gamma ^{k}=0$\ and
$\Delta =0$\ are satisfied on $G^{1}\cup G^{2}$ for the solution
of the corresponding evolution problem. The problem is addressed
in three main steps through a judicious adaptation of the
hierarchical method set up by\ Rendall \cite{24} to construct, for
the Einstein equations in vacuum and with perfect fluid source,
$C^{\infty }$ data satisfying the harmonic gauge conditions
$\Gamma
^{k}=0$ on $G^{1}\cup G^{2}$. The construction of the data is done fully on $%
G^{1}$ and it will be clear that data on $G^{2}$ are constructed
in quite a similar way. The novelty here is that the data are
constructed for the EYMH model whereas those of \cite{24} were
constructed either for the vacuum Einstein or Einstein-perfect
fluid models. Moreover all calculations, though very tedious and
lengthy, are performed in details. This work constitute an
important step towards the global resolution, by energy methods,
of the Goursat problem associated to the EYMH equations in spaces
of functions of finite differentiability order. The construction
will be made in a standard harmonic coordinates system. The
existence of such standard harmonic coordinates system has been
established by A.\ D.\ Rendall \cite{24}.

\begin{definition}
\textit{Let\ }$M$\textit{\ be a }$4$\textit{-dimensional manifold
endowed
with a lorentzian metric }$\mu ,$\textit{\ }$N^{1}$\textit{\ and }$N^{2}$%
\textit{\ two intersecting null hypersurfaces, }$S=N^{1}\cap
N^{2}$\textit{. Consider a local coordinates system }$\left(
x^{i}\right) $\textit{\ in a neighborhood of }$N^{1}\cup
N^{2}$\textit{. }$\left( x^{i}\right) $\textit{\ is a standard
harmonic system w.r.t. }$\mu ,$\textit{\ }$N^{1}$\textit{\ and
}$N^{2}$\textit{\ if the following conditions are satisfied:}

$\left( i\right) $\textit{\ }$\left( x^{i}\right) $\textit{\ is a
harmonic system w.r.t. }$\mu $ i.e.$,$ $\mu ^{ij}\overline{\Gamma
}_{ij}^{k}=0$ \textit{for all }$k,$\textit{\ where
}$\overline{\Gamma }_{ij}^{k}$\textit{\ are the Christoffel
symbols relative to the metric }$\mu $ in local coordinates
$\left( x^{i}\right) .$

$\left( ii\right) $\textit{\ }$N^{1}$\textit{\ and
}$N^{2}$\textit{\ are
locally defined by }$x^{1}=0$\textit{\ and }$x^{2}=0$ \textit{respectively}$%
, $

$\left( iii\right) $\textit{\ }$x^{1}$\textit{\ is an affine
parameter along the null geodesics that generate }$N^{2},$

$\left( iv\right) $\textit{\ }$x^{2}$\textit{\ is an affine
parameter along the null geodesics that generate }$N^{1},$

$\left( v\right) $\textit{\ }$x^{3}$\textit{\ and
}$x^{4}$\textit{\ are
constant along the null geodesics that generate }$N^{1}$\textit{\ or }$%
N^{2}. $
\end{definition}

\textbf{Consequence of the above definition (see \cite{24, 25})}

If $\left( x^{i}\right) $ is a standard harmonic system w.r.t. $\gamma ,$ $%
N^{1}$ and $N^{2},$ then the following relations hold:

On $N^{1},$%
\begin{equation}
\gamma _{2i}=0\text{ for }i\neq 1,\quad \gamma _{22,1}=2\gamma
_{12,2}. \tag{3.1}
\end{equation}

On $N^{2}$,
\begin{equation}
\gamma _{1i}=0\text{ for }i\neq 2,\quad \gamma _{11,2}=2\gamma
_{12,1}. \tag{3.2}
\end{equation}

For the sake of completeness of the paper we recall the
implementation of the method of A.\ D.\ Rendall \cite{24} to
construct $C^{\infty }$ initial data on $G_{T}^{1}$ for the
characteristic EYMH system as in \cite{15}. In the course of doing
this, we provide proofs of some important statements that were
missing in \cite{15}. Those proofs constitute the main
contribution of the present work, in comparison with paragraph 7.4 of \cite%
{15}. We assume the following conditions for the free data%
\begin{equation}
\begin{array}{l}
g_{22}=g_{23}=g_{24}=0,\quad A_{2}=0\text{ on }G_{T}^{1}, \\
\Phi ,\text{ }A_{3}\text{ and }A_{4}\text{ are given }C^{\infty
}\text{
functions on }G_{T}^{1}.%
\end{array}
\tag{3.3}  \label{4.14}
\end{equation}

\begin{remark}
\textit{\ The conditions }$g_{22}=g_{23}=g_{24}=0$\textit{\ on }$G_{T}^{1}$%
\textit{\ are in accordance with }$\left( 3.1\right) $\textit{\
since the
general idea is to produce a space-time for which the given coordinates in }$%
\mathbb{R}^{4}$\textit{\ are standard coordinates.}
\end{remark}

\subsection{Construction of\ $\left( g_{\protect\alpha \protect\beta %
}\right) _{\protect\alpha ,\protect\beta \in \left\{ 3,4\right\} }$\ and $%
g_{12}$, arrangement of relations $\Gamma ^{1}=0$\ and
$g_{22,1}=2g_{12,2}$ on $G_{T}^{1}$}

Let $T\in (0,T_{0}]$, $\left( h_{\alpha \beta }\right) $ a matrix
function with determinant $1$ at each point of $G_{T}^{1}$. Set
$g_{\alpha \beta }=\Omega h_{\alpha \beta }$, where $\Omega >0$ is
an unknown function called
the conformity factor. From the free data given above in $\left( \ref{4.14}%
\right) $ one easily sees that the following algebraic relations hold on $%
G_{T}^{1}$

\begin{equation}
\begin{array}{l}
g_{12}g^{12}=1,\quad g^{11}=g^{1\alpha }=0, \\
g^{2\beta }g_{\alpha \beta }=-g^{12}g_{1\alpha },\quad g_{\lambda
\beta
}g^{\alpha \beta }=\delta _{\lambda }^{\alpha }.%
\end{array}
\tag{3.4}  \label{4.15}
\end{equation}%
At this level, the expression of $R_{22}$ and $\tau _{22}$ are
needed. A straightforward calculation shows that on $G_{T}^{1}$
the following equalities hold (see appendix B)

\begin{equation}
\begin{array}{l}
R_{22}=\frac{1}{4}g^{12}g^{\alpha \beta }g_{\alpha \beta ,2}\left(
2g_{12,2}-g_{22,1}\right) +\frac{1}{4}g_{,2}^{\beta \lambda
}g_{\lambda \beta ,2}-\frac{1}{2}\left( g^{\alpha \beta }g_{\alpha
\beta ,2}\right)
_{,2}, \\
\tau _{22}=\Omega ^{-1}h^{\alpha \beta }A_{\alpha ,2}.A_{\beta
,2}+\left(
\Phi _{,2}\right) ^{2}.%
\end{array}
\tag{3.5}  \label{4.16}
\end{equation}%
If\textit{\ }in addition we assume $g_{22,1}=2g_{12,2}$\ on
$G_{T}^{1}$, then $\Gamma ^{1}=0$\ is equivalent to

\begin{equation}
g_{12,2}=\frac{1}{2}g_{12}\frac{\Omega _{,2}}{\Omega }.  \tag{3.6}
\label{4.17}
\end{equation}%
The equation

\begin{equation}
\frac{1}{4}g_{,2}^{\alpha \beta }g_{\alpha \beta
,2}-\frac{1}{2}\left( g^{\alpha \beta }g_{\alpha \beta ,2}\right)
_{,2}=\tau _{22},  \tag{3.7} \label{4.18}
\end{equation}%
provides the following non linear second order ODE with the
conformity factor $\Omega $\ as unknown

\begin{equation}
-\left( \frac{\Omega _{,2}}{\Omega }\right)
^{2}+\frac{1}{2}h_{\alpha \beta ,2}h_{,2}^{\alpha \beta }-2\left(
\frac{\Omega _{,2}}{\Omega }\right) _{,2}=\Omega ^{-1}h^{\alpha
\beta }A_{\alpha ,2}.A_{\beta ,2}.  \tag{3.8} \label{4.19}
\end{equation}%
If we set $\Omega =e^{v}$, then $\left( \ref{4.19}\right) $ reads

\begin{equation}
2v_{,22}=f\left( x,v,v_{,2}\right) ,  \tag{3.9}  \label{4.20}
\end{equation}%
\medskip where
\begin{equation}
f\left( x,v,v_{,2}\right) =-\left( v_{,2}\right)
^{2}-2e^{-v}h^{\alpha \beta }A_{\alpha ,2}.A_{\beta
,2}+\frac{1}{2}h_{\alpha \beta ,2}h_{,2}^{\alpha \beta }-2\left(
\Phi _{,2}\right) ^{2}.  \tag{3.10}  \label{4.20a}
\end{equation}%
The following proposition provides the construction of the
conformity factor. Its proof follows directly from known local
existence and uniqueness results concerning non linear ODEs
(depending on parameters with $C^{\infty } $ coefficients and
initial data).

\begin{proposition}
\textit{Let }$T\in (0,T_{0}]$\textit{\ and assume the following
smoothness
condition for the free data }%
\begin{equation}
h_{33},h_{34},h_{44},A_{3},A_{4},\Phi \in C^{\infty }\left(
G_{T}^{1}\right) .  \tag{3.11}
\end{equation}%
\textit{\ Take }$v_{0}$\textit{, }$v_{1}\in C^{\infty }\left(
\Gamma \right) ,$\textit{\ where }$\Gamma \equiv G_{T}^{1}\cap
G_{T}^{2}.$\textit{\ Then there exists }$T_{1}\in (0,T]$\textit{\
such that }$\left( \ref{4.20}\right) $\textit{\ has a unique
solution }$v\in C^{\infty }\left(
G_{T_{1}}^{1}\right) $\textit{\ satisfying }$v=v_{0}$\textit{\ and }$%
v_{,2}=v_{1}$\textit{\ on }$\Gamma .$
\end{proposition}

As the conformity factor is already known, we now consider the
following first order linear ODE with unknown $g_{12}$

\begin{equation}
g_{12,2}=\frac{1}{2}g_{12}v_{,2}.  \tag{3.12}  \label{4.22}
\end{equation}%
\medskip The following proposition provides the construction of $g_{12}$.
Its proof follows straightforwardly from known global existence
and
uniqueness results concerning linear ODEs (depending on parameters with $%
C^{\infty }$ coefficients and initial data).

\begin{proposition}
\textit{Let }$w_{0}\in C^{\infty }\left( \Gamma \right) .$ \textit{Then }$%
\left( \ref{4.22}\right) $\textit{\ has a unique solution
}$g_{12}\in
C^{\infty }\left( G_{T_{1}}^{1}\right) $\textit{\ satisfying }$g_{12}=w_{0}$%
\textit{\ on }$\Gamma .$
\end{proposition}

The condition $g_{22,1}-2g_{12,2}=0$ on $G_{T_{1}}^{1}$\ is now
arranged in the following proposition.\textbf{\ }

\begin{proposition}
On $G_{T_{1}}^{1},$\ the reduced equation $\widetilde{R}_{22}\
=\tau _{22}$\
is equivalent to the following homogenous ODE with unknown $%
g_{22,1}-2g_{12,2}$%
\begin{equation}
\left( g^{12}\right) ^{2}g_{12,2}\left( g_{22,1}-2g_{12,2}\right)
-g^{12}\left( g_{22,1}-2g_{12,2}\right) _{,2}=0.  \tag{3.13}
\label{4.23}
\end{equation}%
\ Assume $g_{22,1}=2g_{12,2}$\ on $\Gamma $.\ Then
$g_{22,1}-2g_{12,2}=0$\ on $G_{T_{1}}^{1}$ and so $\Gamma ^{1}=0$\
on $G_{T_{1}}^{1}$.
\end{proposition}

\begin{proof}
In view of $\left( \ref{4.22}\right) $ and $\left(
\ref{4.15}\right) ,$\ it
holds that $g^{\alpha \beta }g_{\alpha \beta ,2}=4g^{12}g_{12,2}$ on $%
G_{T_{1}}^{1}$. Thus on $G_{T_{1}}^{1}$ it holds that%
\begin{eqnarray*}
2\Gamma ^{1} &=&g^{ij}g^{1k}\left( 2g_{ki,j}-g_{ij,k}\right)
=g^{ij}g^{12}\left( 2g_{2i,j}-g_{ij,2}\right) \\
&=&g^{12}\left[ g^{12}\left( g_{12,2}\right) +g^{21}\left(
2g_{22,1}-g_{21,2}\right) +g^{\alpha \beta }\left( -g_{\alpha
\beta
,2}\right) \right] \\
&=&g^{12}\left( 2g^{12}g_{22,1}-g^{\alpha \beta }g_{\alpha \beta
,2}\right)
\\
&=&g^{12}\left( 2g^{12}g_{22,1}-4g^{12}g_{12,2}\right) \\
&=&2\left( g^{12}\right) ^{2}\left( g_{22,1}-2g_{12,2}\right) .
\end{eqnarray*}%
Thus%
\begin{eqnarray*}
\Gamma _{,2}^{1} &=&\left[ \left( g^{12}\right) ^{2}\left(
g_{22,1}-2g_{12,2}\right) \right] _{,2} \\
&=&2g^{12}g_{,2}^{12}\left( g_{22,1}-2g_{12,2}\right) +\left(
g^{12}\right) ^{2}\left( g_{22,1}-2g_{12,2}\right) _{,2}.
\end{eqnarray*}%
A simple calculation shows that $g_{,2}^{12}=-\left( g^{12}\right)
^{2}g_{12,2}$ on $G_{T_{1}}^{1},$ since $g^{12}g_{12}=1$ on
$G_{T_{1}}^{1}$. Hence
\begin{equation*}
\Gamma _{,2}^{1}=-2\left( g^{12}\right) ^{3}g_{12,2}\left(
g_{22,1}-2g_{12,2}\right) +\left( g^{12}\right) ^{2}\left(
g_{22,1}-2g_{12,2}\right) _{,2}.
\end{equation*}%
It is easy to see that $g_{k2}\Gamma _{,2}^{k}=$ $g_{12}\Gamma
_{,2}^{1}$ on
$G_{T_{1}}^{1},$ since $g_{2k}=0$ for $k\neq 1.$ In view of $\left( \ref%
{4.16}\right) $ we have
\begin{equation*}
R_{22}-g_{12}\Gamma _{,2}^{1}=\left( g^{12}\right)
^{2}g_{12,2}\left( g_{22,1}-2g_{12,2}\right) -g^{12}\left(
g_{22,1}-2g_{12,2}\right) _{,2}+\tau _{22}.
\end{equation*}%
Therefore, since $\widetilde{R}_{22}\equiv
R_{22}-\frac{1}{2}\left(
g_{k2}\Gamma _{,2}^{k}+g_{k2}\Gamma _{,2}^{k}\right) =-\frac{1}{2}%
g^{km}g_{22,mk}+Q_{22}$, the reduced equation%
\begin{equation*}
\widetilde{R}_{22}=\tau _{22}
\end{equation*}%
is equivalent to%
\begin{equation*}
\left( g^{12}\right) ^{2}g_{12,2}\left( g_{22,1}-2g_{12,2}\right)
-g^{12}\left( g_{22,1}-2g_{12,2}\right) _{,2}=0\text{ on
}G_{T_{1}}^{1}.
\end{equation*}%
The result now follows.
\end{proof}

\begin{remark}
\textit{In the same way, given} \textit{a matrix function }$\left(
h_{\alpha \beta }\right) =\left(
\begin{array}{ll}
h_{33} & h_{34} \\
h_{34} & h_{44}%
\end{array}%
\right) $\textit{\ with determinant }$1$\textit{\ at each point of }$\mathit{%
G}_{T}^{2}$\textit{, set }$g_{\alpha \beta }=\Omega h_{\alpha \beta },$%
\textit{\ where }$\Omega >0$\textit{\ is an unknown function
called the
conformity factor}$.$\textit{\ Assuming}%
\begin{equation*}
\begin{array}{l}
g_{11}=g_{13}=g_{14}=0\text{ \textit{on} }G_{T}^{2},\quad
A_{1}=0\text{
\textit{on} }G_{T}^{2}, \\
\Phi ,\text{ }A_{3}\text{ \textit{and} }A_{4}\text{ \textit{are given }}%
C^{\infty }\text{ \textit{functions\ on} }G_{T}^{2},%
\end{array}%
\end{equation*}%
$g_{12}$\textit{\ and }$g_{\alpha \beta }$\textit{\ are constructed on\ }$%
G_{T_{1}}^{2}$\textit{\ and the relations }$g_{11,2}-2g_{12,1}=0$,\textit{\ }%
$\Gamma ^{2}=0$\textit{\ are\ arranged\ on }$G_{T_{1}}^{2}.$
\end{remark}

We now proceed to the construction of the data $g_{13},$\ $g_{14}$\ and\ $%
A_{1}$\ in $C^{\infty }\left( G_{T_{1}}^{1}\right) $ as\ well\ as\
the\ arrangement\ of\ the\ relations\ $\Gamma ^{\alpha }=0$\ and\
$\Delta =0$\ on $G_{T_{1}}^{1}$, $\alpha =3,4$.

\subsection{Construction of\ $g_{1\protect\alpha }$\ and\ $A_{1}$,
arrangement of relations $\Gamma ^{\protect\alpha }=0$\ and\
$\Delta =0$\ on $G_{T_{1}}^{1},$ $\protect\alpha =3,4$}

We seek for a combination of $R_{2\alpha }$, $\Gamma ^{\alpha }$,
$\Gamma _{,2}^{\alpha },$ $LA_{2}$, $\Delta $ and $\Delta _{,2}$\
that will provide
a system of ODEs on $G_{T_{1}}^{1}$ with unknowns $g_{1\alpha }$ and $A_{1}$%
. It\ is\ at\ this\ moment\ that\ the\ assumption $A_{2}=0$ on
$G_{T}^{1}$, which\ permits to avoid to deal with $g_{11}$\ at
this level of the construction process, is$\ $needed.\ After
performing tedious and lengthy calculations, we have the following
result:

\begin{proposition}
$\left( i\right) $ On $G_{T_{1}}^{1},$\ the following combinations hold%
\begin{equation}
\begin{array}{l}
R_{2\alpha }+\frac{1}{2}g_{\alpha \beta }\Gamma _{,2}^{\beta
}+\left( g^{12}g_{12,2}g_{\alpha \beta }+\frac{1}{2}g_{\alpha
\beta ,2}\right) \Gamma
^{\beta } \\
=g^{12}g_{1\alpha ,22}+\left( g^{12}\right) ^{2}g_{12,2}g_{1\alpha
,2}-g_{\alpha \beta ,2}g^{\beta \lambda }g^{12}g_{1\lambda ,2} \\
+\left\{ \left( g^{12}\right) ^{2}g_{12,2}g_{\alpha \beta
}g_{,2}^{\beta \lambda }+\frac{1}{2}\left[ g_{\alpha \beta
}g^{12}g_{,22}^{\beta \lambda }-g^{12}g^{\beta \lambda }g_{\alpha
\beta ,22}\right] \right\} g_{1\lambda
}+c_{\alpha },%
\end{array}
\tag{3.14}  \label{4.24}
\end{equation}%
\begin{equation}
\begin{array}{l}
LA_{2}-2\Delta _{,2}-2g^{12}g_{12,2}\Delta +2\left(
g^{12}g_{12,2}A_{\nu
}+A_{\nu ,2}\right) \Gamma ^{\nu }+2A_{\nu }\Gamma _{,2}^{\nu } \\
=-2g^{12}A_{1,22}-2\left( g^{12}\right)
^{2}g_{12,2}A_{1,2}+2g^{12}g^{\alpha
\lambda }A_{\alpha ,2}g_{1\lambda ,2}+K^{\lambda }g_{1\lambda }+A_{g},%
\end{array}
\tag{3.15}  \label{4.25}
\end{equation}%
where
\begin{equation}
\begin{array}{l}
c_{\alpha }=\frac{1}{2}\left( g^{12}\right) g_{12,2}\left[
-2g^{12}g_{12,\alpha }+g^{\mu \theta }\left( 2g_{\alpha \mu
,\theta }-g_{\mu
\theta ,\alpha }\right) \right] \\
\text{ \ \ \ \ \ }+\frac{1}{4}g_{\alpha \beta ,2}\left[ -2g^{\beta
\lambda }g^{12}g_{12,\lambda }+g^{\beta \lambda }g^{\mu \theta
}\left( 2g_{\lambda
\mu ,\theta }-g_{\mu \theta ,\lambda }\right) \right] \\
\text{ \ \ \ \ \ }+\frac{1}{2}\left( g^{\lambda \beta }g_{\alpha
\beta
,2}\right) _{,\lambda }-3\left( g^{12}g_{12,2}\right) _{,\alpha }-\frac{1}{2}%
\left( g^{12}\right) ^{2}g_{12,2}g_{12,\alpha } \\
\text{ \ \ \ \ \ }+\frac{1}{2}g^{12}\left( g_{22,1\alpha
}+g_{12,2\alpha }\right) +\frac{1}{2}g_{,2}^{\beta \lambda }\left(
g_{\lambda \beta ,\alpha
}+g_{\lambda \alpha ,\beta }\right) \\
\text{ \ \ \ \ \ }+\frac{1}{2}g_{\alpha \beta }\left[ -2g^{\beta
\lambda }g^{12}g_{12,\lambda }+g^{\beta \lambda }g^{\mu \theta
}\left( 2g_{\lambda
\mu ,\theta }-g_{\mu \theta ,\lambda }\right) \right] _{,2},%
\end{array}
\tag{3.16}  \label{4.24a}
\end{equation}%
\begin{equation}
\begin{array}{l}
K^{\lambda }=4\left( g^{12}\right) ^{2}g_{12,2}g^{\alpha \lambda
}A_{\alpha ,2}+\left[ 2g^{12}g^{\alpha \lambda }g^{\beta \mu
}g_{\mu \beta ,2}A_{\alpha
}+2g^{12}g^{\alpha \lambda }A_{\alpha ,2}\right] _{,2} \\
\text{ \ \ \ \ \ }-2\left( g^{12}g^{\alpha \lambda }A_{\alpha
}\right)
_{,2}g^{\beta \mu }g_{\mu \beta ,2}-2g^{12}g^{\alpha \lambda }A_{\alpha }%
\left[ g^{\beta \mu }g_{\mu \beta ,2}\right] _{,2}, \\
A_{g}=-2\left( g^{12}g^{\alpha \lambda }A_{\alpha }\right)
_{,2}g_{12,\lambda }-2g^{12}g^{\alpha \lambda }A_{\alpha
}g_{12,2\lambda }
\\
\text{ \ \ \ \ \ }-g^{\beta \alpha }\left( \left[ A_{\beta },A_{\alpha ,2}%
\right] -g^{12}g_{12,\beta }A_{\alpha ,2}\right) \\
\text{ \ \ \ \ \ }+\left( g^{12}g^{\alpha \lambda }A_{\alpha
}\right) _{,2}g_{12}\left[ g^{\beta \mu }\left( g_{\mu \lambda
,\beta }+g_{\lambda
\beta ,\mu }-g_{\mu \beta ,\lambda }\right) \right] \\
\text{ \ \ \ \ \ }+g^{12}g^{\alpha \lambda }A_{\alpha }\left(
g_{12}\left[ g^{\beta \mu }\left( g_{\mu \lambda ,\beta
}+g_{\lambda \beta ,\mu }-g_{\mu
\beta ,\lambda }\right) \right] \right) _{,2} \\
\text{ \ \ \ \ \ }-\left\{ g^{\alpha \beta }\left[ g^{12}g_{12,\beta }+\frac{%
1}{2}g^{\lambda \mu }\left( g_{\mu \beta ,\lambda }+g_{\lambda \mu
,\beta }-g_{\beta \lambda ,\mu }\right) \right] +g_{,\beta
}^{\alpha \beta
}\right\} A_{\alpha ,2} \\
\text{ \ \ \ \ \ }+\left( 2g^{12}g^{\alpha \lambda }g_{12,\lambda
}A_{\alpha }-\left[ g^{\alpha \delta }g^{\beta \mu }\left( g_{\mu
\delta ,\beta }+g_{\delta \beta ,\mu }-g_{\mu \beta ,\delta
}\right) \right] A_{\alpha
}\right) _{,2},%
\end{array}
\tag{3.17}  \label{4.25a}
\end{equation}%
and all the coefficients are known on $G_{T_{1}}^{1}.$

$\left( ii\right) $ On $G_{T_{1}}^{1}$\ the system%
\begin{equation}
\begin{array}{l}
R_{23}+\frac{1}{2}g_{3\beta }\Gamma _{,2}^{\beta }+\left(
g^{12}g_{12,2}g_{3\beta }+\frac{1}{2}g_{3\beta ,2}-A_{\beta
}.A_{3,2}\right)
\Gamma ^{\beta }+A_{3,2}.\Delta =\tau _{23}, \\
R_{24}+\frac{1}{2}g_{4\beta }\Gamma _{,2}^{\beta }+\left(
g^{12}g_{12,2}g_{4\beta }+\frac{1}{2}g_{4\beta ,2}-A_{\beta
}.A_{4,2}\right)
\Gamma ^{\beta }+A_{4,2}.\Delta =\tau _{24}, \\
LA_{2}-2\Delta _{,2}-2g^{12}g_{12,2}\Delta +2\left(
g^{12}g_{12,2}A_{\nu
}+A_{\nu ,2}\right) \Gamma ^{\nu }+2A_{\nu }\Gamma _{,2}^{\nu }=J_{2},%
\end{array}
\tag{3.18}  \label{4.26}
\end{equation}%
is equivalent to the following second order system of ODEs with unknown $%
\left( A_{1},g_{13},g_{14}\right) $%
\begin{equation}
\begin{array}{l}
g^{12}g_{13,22}+\kappa _{3}^{\lambda }g_{1\lambda ,2}+\varkappa
_{3}.A_{1,2}+\chi _{3}^{\lambda }g_{1\lambda }+\digamma _{3}=0, \\
g^{12}g_{14,22}+\kappa _{4}^{\lambda }g_{1\lambda ,2}+\varkappa
_{4}.A_{1,2}+\chi _{4}^{\lambda }g_{1\lambda }+\digamma _{4}=0, \\
-2g^{12}A_{1,22}-2\left( g^{12}\right)
^{2}g_{12,2}A_{1,2}+a^{\lambda
}g_{1\lambda }+b=0,%
\end{array}
\tag{3.19}  \label{4.27}
\end{equation}%
where all the coefficients are known on $G_{T_{1}}^{1}$\ and given as follows%
\begin{equation}
\begin{array}{l}
\kappa _{3}^{3}=\left( g^{12}\right) ^{2}g_{12,2}-g_{3\beta ,2}g^{\beta 3},%
\text{\quad }\kappa _{3}^{4}=-g_{3\beta ,2}g^{\beta 4},\text{\quad
}\kappa
_{4}^{3}=-g_{4\beta ,2}g^{\beta 3}, \\
\kappa _{4}^{4}=\left( g^{12}\right) ^{2}g_{12,2}-g_{4\beta ,2}g^{\beta 4},%
\text{\quad }\varkappa _{3}=2g^{12}A_{3,2},\text{\quad }\varkappa
_{4}=2g^{12}A_{4,2}, \\
\chi _{\alpha }^{\lambda }=\left( g^{12}\right)
^{2}g_{12,2}g_{\alpha \beta }g_{,2}^{\beta \lambda
}+\frac{1}{2}\left[ g_{\alpha \beta
}g^{12}g_{,22}^{\beta \lambda }-g^{12}g^{\beta \lambda }g_{\alpha \beta ,22}%
\right] -2g^{\nu \lambda }g^{12}A_{\alpha ,2}.A_{\nu ,2}, \\
a^{\lambda }=4\left( g^{12}\right) ^{2}g_{12,2}g^{\alpha \lambda
}A_{\alpha ,2}+\left[ 2g^{12}g^{\alpha \lambda }g^{\beta \mu
}g_{\mu \beta ,2}A_{\alpha
}+2g^{12}g^{\alpha \lambda }A_{\alpha ,2}\right] _{,2} \\
\text{ \ \ \ \ \ }-2\left( g^{12}g^{\alpha \lambda }A_{\alpha
}\right)
_{,2}g^{\beta \mu }g_{\mu \beta ,2}-2g^{12}g^{\alpha \lambda }A_{\alpha }%
\left[ g^{\beta \mu }g_{\mu \beta ,2}\right] _{,2}, \\
\digamma _{\alpha }=c_{\alpha }+g^{\beta \lambda }\left( A_{\alpha
,\lambda }-A_{\lambda ,\alpha }+\left[ A_{\lambda },A_{\alpha
}\right] \right)
.A_{\beta ,2},\text{\quad }b=A_{g}-J_{2}.%
\end{array}
\tag{3.20}  \label{4.28}
\end{equation}
\end{proposition}

\begin{proof}
See appendix C.
\end{proof}

The proofs of the following statements that provide the construction of $%
\left( g_{13},g_{14},A_{1}\right) $ on $G_{T_{1}}^{1}$ with
$\Gamma ^{\beta }=0$\ and $\Delta =0$\ on $G_{T_{1}}^{1}$ are
direct consequences of Proposition 3.7.

\begin{proposition}
Let $a_{0},\ a_{1},\ b_{0},\ b_{1},\ c_{0},\ c_{1}\in C^{\infty
}\left( \Gamma \right) $. Then system $\left( \ref{4.27}\right) $\
has a unique solution $\left( g_{13},g_{14},A_{1}\right) $ in
$C^{\infty }\left( G_{T_{1}}^{1}\right) $\ satisfying
\begin{equation*}
\left( g_{13},\text{ }g_{14},\text{ }A_{1}\right) =\left( a_{0},\text{ }%
b_{0},\text{ }c_{0}\right) \text{ on }\Gamma ,
\end{equation*}%
\ and
\begin{equation*}
\left( g_{13,2},\text{ }g_{14,2},\text{ }A_{1,2}\right) =\left(
a_{1},\text{ }b_{1},\text{ }c_{1}\right) \text{ on }\Gamma .
\end{equation*}%
\
\end{proposition}

Now\ the relations $\Gamma ^{\beta }=0$\ and $\Delta =0$\ on $G_{T_{1}}^{1}$%
\ are arranged in the following proposition. The proof is similar
to the proof of Proposition 3.5.

\begin{proposition}
$\left( i\right) $ On $G_{T_{1}}^{1},$\ the reduced system
\begin{equation*}
\begin{array}{l}
\widetilde{R}_{2\alpha }=\tau _{2\alpha }, \\
LA_{2}=J_{2},%
\end{array}%
\end{equation*}%
\ is equivalent to%
\begin{equation}
\begin{array}{l}
g_{3\beta }\Gamma _{,2}^{\beta }+\left( g^{12}g_{12,2}g_{3\beta }+\frac{1}{2}%
g_{3\beta ,2}-A_{\beta }.A_{3,2}\right) \Gamma ^{\beta
}+A_{3,2}.\Delta =0,
\\
g_{4\beta }\Gamma _{,2}^{\beta }+\left( g^{12}g_{12,2}g_{4\beta }+\frac{1}{2}%
g_{4\beta ,2}-A_{\beta }.A_{4,2}\right) \Gamma ^{\beta
}+A_{4,2}.\Delta =0,
\\
2A_{\beta }\Gamma _{,2}^{\beta }-2\Delta
_{,2}-2g^{12}g_{12,2}\Delta
+2\left( g^{12}g_{12,2}A_{\beta }+A_{\beta ,2}\right) \Gamma ^{\beta }=0.%
\end{array}
\tag{3.21}  \label{4.29}
\end{equation}%
$\left( ii\right) $ if $\Gamma ^{\beta }=0$\ and $\Delta =0$\ on
$\Gamma $\ then $\Gamma ^{\beta }=0$\ and $\Delta =0$\ on
$G_{T_{1}}^{1}.$
\end{proposition}

\begin{proof}
\textbf{Proof of item }$\left( i\right) $.\textbf{\ }By definition of $%
\widetilde{R}_{ij}$ (see $\left( 2.22\right) $) and since $\Gamma
^{1}=0$ on $G^{1}$\ at this step of the construction process,\ the
reduced system
\begin{equation*}
\begin{array}{l}
\widetilde{R}_{23}=\tau _{23}, \\
\widetilde{R}_{24}=\tau _{24}, \\
LA_{2}=J_{2},%
\end{array}%
\end{equation*}%
is equivalent to
\begin{equation*}
\begin{array}{l}
R_{23}-\frac{1}{2}g_{3\beta }\Gamma _{,2}^{\beta }=\tau _{23}, \\
R_{24}-\frac{1}{2}g_{4\beta }\Gamma _{,2}^{\beta }=\tau _{24}, \\
LA_{2}=J_{2}.%
\end{array}%
\end{equation*}%
In view of $\left( \ref{4.26}\right) $, this is equivalent to $\left( \ref%
{4.29}\right) $.

\textbf{Proof of item }$\left( ii\right) $. $\left(
\ref{4.29}\right) $ is a linear homogeneous system of first order
ODEs on $G_{T_{1}}^{1},$\ with
unknown $\left( \Gamma ^{3},\Gamma ^{4},\Delta \right) ,$ with variable $%
x^{2},$ with $C^{\infty }$ coefficients depending smoothly on parameters $%
x^{3}$ and $x^{4}$. This yields $\Gamma ^{\alpha }=0$ and $\Delta =0$ in $%
G_{T_{1}}^{1}$, if $\Gamma ^{\alpha }=0$ and $\Delta =0$ on
$\Gamma $. (The conditions $\Gamma ^{\alpha }=0$ and $\Delta =0$
on $\Gamma $ can be
obtained upon a judicious choice of the data $a_{0},$\textit{\ }$a_{1},$%
\textit{\ }$b_{0},$\textit{\ }$b_{1},$\textit{\ }$c_{0},$\textit{\
}$c_{1}$ on $\Gamma $ (see Proposition\ 3.8).
\end{proof}

\begin{remark}
$\left( i\right) $\textit{\ In the same way, under assumption
}$A_{1}=0\ $on $G_{T}^{2},$ $g_{2\alpha }$\textit{\ and
}$A_{2}$\textit{\ are constructed on }$G_{T_{1}}^{2},$\textit{\
the relations }$\Gamma ^{\beta }=0$\textit{\ and }$\Delta
=0$\textit{\ are established on }$G_{T_{1}}^{2}$\textit{.}

$\left( ii\right) $\textit{\ The determination of }$g_{1\alpha
}$\textit{\ and }$A_{1}$\textit{\ on }$G_{T_{1}}^{1},$\textit{\
with }$\Gamma ^{\beta }=0 $\textit{\ and }$\Delta =0$\textit{\ on
}$G_{T_{1}}^{1},$\textit{\
provides the determination of }$g^{2\beta }$\textit{, }$g_{2\lambda ,1}$%
\textit{\ and }$A_{2,1}$\textit{\ on }$G_{T_{1}}^{1}$\textit{\
thanks to the
following relations that hold on }$G_{T_{1}}^{1}:$%
\begin{equation}
\begin{array}{l}
g^{2\beta }=-g^{12}g^{\beta \alpha }g_{1\alpha }, \\
\Delta =g^{12}\left( A_{2,1}+A_{1,2}\right) +g^{2\alpha }A_{\alpha
,2}+g^{\alpha \beta }A_{\alpha ,\beta }+\Gamma ^{\alpha }A_{\alpha }, \\
g^{12}g_{2\lambda ,1}=g_{\lambda \beta }\Gamma ^{\beta
}-g^{12}g_{1\lambda
,2}-g_{\lambda \beta }g_{,2}^{\beta \mu }g^{12}g_{1\mu } \\
\text{ \ \ \ \ \ \ \ \ \ \ \ \ \ \ }-\frac{1}{2}\left[
-2g^{12}g_{12,\lambda }+g^{\mu \theta }\left( 2g_{\lambda \mu
,\theta }-g_{\mu \theta ,\lambda
}\right) \right] .%
\end{array}
\tag{3.22}  \label{4.30}
\end{equation}%
$\left( iii\right) $\textit{\ Similarly, the determination of }$g_{2\alpha }$%
\textit{\ and }$A_{2}$\textit{\ on }$G_{T_{1}}^{2}$\textit{\ with
}$\Gamma ^{\beta }=0$\textit{\ and }$\Delta =0$\textit{\ on
}$G_{T_{1}}^{2}$\textit{\
provides the determination of }$g^{1\beta }$\textit{,\ }$g_{1\lambda ,2}$%
\textit{\ and }$A_{1,2}$\textit{\ on }$G_{T_{1}}^{2}$\textit{\
thanks to the
following relations that hold on }$G_{T_{1}}^{2}:$%
\begin{equation}
\begin{array}{l}
g^{1\beta }=-g^{12}g^{\beta \alpha }g_{2\alpha }, \\
\Delta =g^{12}\left( A_{2,1}+A_{1,2}\right) +g^{1\alpha }A_{\alpha
,1}+g^{\alpha \beta }A_{\alpha ,\beta }+\Gamma ^{\alpha }A_{\alpha }, \\
g^{12}g_{1\lambda ,2}=g_{\lambda \beta }\Gamma ^{\beta
}-g^{12}g_{2\lambda
,1}-g_{\lambda \beta }g_{,1}^{\beta \mu }g^{12}g_{2\mu } \\
\text{\quad \quad \quad \quad }-\frac{1}{2}\left[
-2g^{12}g_{12,\lambda }+g^{\mu \theta }\left( 2g_{\lambda \mu
,\theta }-g_{\mu \theta ,\lambda
}\right) \right] .%
\end{array}
\tag{3.23}  \label{4.31}
\end{equation}
\end{remark}

The\ last level of the hierarchy\ is now described.

\subsection{Construction of\ $g_{11}$\ on $G_{T_{1}}^{1}$\ and arrangement
of relation $\Gamma ^{2}=0$\ on $G_{T_{1}}^{1}$}

We now consider the reduced equations $\widetilde{R}_{\alpha \beta
}=\tau _{\alpha \beta }$ which are equivalent to $R_{\alpha \beta
}=\tau _{\alpha \beta }$ since
\begin{equation*}
\widetilde{R}_{\alpha \beta }=R_{\alpha \beta }-\frac{1}{2}\left(
g_{k\alpha
}\Gamma _{,\beta }^{k}+g_{k\beta }\Gamma _{,\alpha }^{k}\right) \text{ and }%
\Gamma ^{1}=\Gamma ^{3}=\Gamma ^{4}=0\text{ on }G_{T_{1}}^{1}.
\end{equation*}%
We seek for a combination of $g^{\alpha \beta }R_{\alpha \beta }$ and $%
\Gamma ^{2}$ that will provide an ODE with unknown $g_{11}$.
Analogously to Proposition 3.7, the following proposition holds
true:

\begin{proposition}
$\left( i\right) $ On $G_{T_{1}}^{1}$\ the following combinations hold%
\begin{equation}
\begin{array}{l}
g^{\alpha \beta }R_{\alpha \beta }-2\Gamma
_{,2}^{2}-2g^{12}g_{12,2}\Gamma
^{2} \\
=-2\left( g^{12}\right) ^{2}g_{11,22}+4\left( g^{12}\right)
^{3}g_{12,2}g_{11,2} \\
\text{ \ \ }+\left\{ 4\left( g^{12}\right) ^{4}\left( g_{12,2}\right) ^{2}+%
\frac{1}{2}\left( g^{12}\right) ^{2}\left( g^{\alpha \beta
}g_{\alpha \beta
,2}\right) _{,2}\right\} g_{11} \\
\text{ \ \ }+\frac{1}{4}g^{\alpha \beta }\left( N_{\alpha \beta
}+M_{\alpha
\beta }\right) -2W-2g^{12}g_{12,2}S, \\
g^{\alpha \beta }\tau _{\alpha \beta }=K,%
\end{array}
\tag{3.24}
\end{equation}

where, at this level of the construction process, $N_{\alpha \beta },$ $%
M_{\alpha \beta },$ $W,$ $S,$ $K$ are known on $G_{T_{1}}^{1}$ and given by%
\begin{equation}
\begin{array}{l}
N_{\alpha \beta }=-g_{\alpha \beta ,2}\left[ \left( g^{12}\right)
^{2}g_{22,1}g^{2\mu }g_{1\mu }-g^{12}g^{2\mu }g_{2\mu ,1}\right] \\
\text{ \ \ \ \ \ \ \ }-2\left( g^{12}\right) ^{2}g_{12,2}\left(
g_{1\beta ,\alpha }+g_{1\alpha ,\beta }\right) -g^{12}\left(
2g^{2\mu }g_{12,2}+g^{\mu \lambda }g_{2\lambda ,1}\right) \left(
g_{\beta \mu ,\alpha }+g_{\mu \alpha
,\beta }-g_{\alpha \beta ,\mu }\right) \\
\text{ \ \ \ \ \ \ \ }+g^{12}\left( g_{2\beta ,1\alpha
}+g_{2\alpha ,1\beta }\right) +g_{\alpha \beta ,2}\left(
g^{12}g^{2\mu }g_{1\mu }\right)
_{,2}+g^{12}g^{2\mu }g_{1\mu }g_{\alpha \beta ,22} \\
\text{ \ \ \ \ \ \ \ }+g_{,2}^{12}\left( g_{1\beta ,\alpha
}+g_{1\alpha ,\beta }\right) +g_{,2}^{2\mu }\left( g_{\mu \beta
,\alpha }+g_{\mu \alpha ,\beta }-g_{\alpha \beta ,\mu }\right)
+g^{12}\left( g_{1\beta ,2\alpha
}+g_{1\alpha ,2\beta }\right) \\
\text{ \ \ \ \ \ \ \ }+g^{2\mu }\left( g_{\beta \mu ,2\alpha
}+g_{\mu \alpha ,2\beta }-g_{\alpha \beta ,2\mu }\right)
-g_{,\lambda }^{2\lambda }g_{\alpha \beta ,2}+g_{,\lambda
}^{\lambda \mu }\left( g_{\mu \beta ,\alpha }+g_{\mu \alpha ,\beta
}-g_{\alpha \beta ,\mu }\right) -g^{2\lambda }g_{\alpha \beta
,2\lambda } \\
\text{ \ \ \ \ \ \ \ }+g^{\lambda \mu }\left( g_{\mu \beta
,\lambda \alpha
}+g_{\mu \alpha ,\lambda \beta }-g_{\alpha \beta ,\lambda \mu }\right) -%
\left[ 2g^{12}g_{12,\alpha }+g^{\lambda \mu }\left( g_{\mu \lambda
,\alpha }+g_{\mu \alpha ,\lambda }-g_{\alpha \lambda ,\mu }\right)
\right] _{,\beta
},%
\end{array}
\tag{3.25}
\end{equation}%
\begin{equation}
\begin{array}{l}
M_{\alpha \beta }=-g^{12}g_{\alpha \beta ,2}\left[
-g^{12}g_{22,1}g^{2\mu }g_{1\mu }+2g^{2\lambda }g_{2\lambda
,1}+g^{\lambda \mu }\left( g_{1\mu
,\lambda }-g_{1\lambda ,\mu }\right) \right] \\
\text{ \ \ \ \ \ \ }+\left[ g^{12}g^{2\mu }g_{1\mu }g_{\alpha
\beta ,2}+g^{12}\left( g_{1\beta ,\alpha }+g_{1\alpha ,\beta
}\right) +g^{2\mu }\left( g_{\mu \beta ,\alpha }+g_{\mu \alpha
,\beta }-g_{\alpha \beta ,\mu }\right) \right] \left(
3g^{12}g_{22,1}+g^{\lambda \mu }g_{\lambda \mu
,2}\right) \\
\text{ \ \ \ \ \ \ }+\left[ 2g^{12}g_{12,\lambda }+g^{\mu \theta
}\left( g_{\mu \theta ,\lambda }+g_{\theta \lambda ,\mu }-g_{\mu
\lambda ,\theta }\right) \right] \left[ -g^{2\lambda }g_{\alpha
\beta ,2}+g^{\mu \lambda }\left( g_{\mu \beta ,\alpha }+g_{\mu
\alpha ,\beta }-g_{\alpha \beta ,\mu
}\right) \right] \\
\text{ \ \ \ \ \ \ }-\left( g^{12}\right) ^{2}\left( g_{12,\beta
}+g_{2\beta ,1}-g_{1\beta ,2}\right) \left( g_{12,\alpha
}+g_{2\alpha ,1}-g_{1\alpha
,2}\right) \\
\text{ \ \ \ \ \ \ }+g^{12}g_{\mu \alpha ,2}\left[ g^{2\mu }\left(
g_{2\beta ,1}+g_{12,\beta }-g_{1\beta ,2}\right) +g^{\lambda \mu
}\left( g_{1\lambda
,\beta }-g_{1\beta ,\lambda }\right) \right] \\
\text{ \ \ \ \ \ \ }-2g^{12}g^{2\lambda }g_{1\lambda }g^{\theta
\mu }g_{\theta \beta ,2}g_{\alpha \mu ,2}-g^{12}g^{\lambda \mu
}g_{\lambda \beta ,2}\left( g_{1\mu ,\alpha }+g_{1\alpha ,\mu
}\right) -g^{\theta \mu }g_{\theta \beta ,2}g^{2\lambda }\left(
g_{\lambda \mu ,\alpha }+g_{\lambda
\alpha ,\mu }-g_{\alpha \mu ,\lambda }\right) \\
\text{ \ \ \ \ \ \ }-\left[ g^{12}\left( g_{12,\beta }+g_{1\beta
,2}-g_{2\beta ,1}\right) +g^{2\mu }g_{\mu \beta ,2}\right] \left[
g^{12}\left( g_{12,\alpha }+g_{1\alpha ,2}-g_{2\alpha ,1}\right)
+g^{2\lambda }g_{\lambda \alpha ,2}\right] \\
\text{ \ \ \ \ \ \ }-g^{12}g^{\lambda \theta }g_{\theta \alpha
,2}\left( g_{1\beta ,\lambda }+g_{1\lambda ,\beta }\right)
-g^{2\mu }g^{\lambda \theta }g_{\theta \alpha ,2}\left( g_{\mu
\beta ,\lambda }+g_{\mu \lambda ,\beta
}-g_{\lambda \beta ,\mu }\right) +m_{\alpha \beta }%
\end{array}
\tag{3.26}  \label{4.34}
\end{equation}%
with%
\begin{equation}
\begin{array}{l}
m_{\alpha \beta }=g^{12}g_{\lambda \beta ,2}\left[ g^{2\lambda
}\left( g_{2\alpha ,1}+g_{12,\alpha }-g_{1\alpha ,2}\right)
+g^{\lambda \mu }\left(
g_{1\mu ,\alpha }-g_{1\alpha ,\mu }\right) \right] \\
-\left[ -g^{2\mu }g_{\lambda \beta ,2}+g^{\theta \mu }\left(
g_{\theta \lambda ,\beta }+g_{\theta \beta ,\lambda }-g_{\lambda
\beta ,\theta }\right) \right] \left[ -g^{2\lambda }g_{\alpha \mu
,2}+g^{\delta \lambda }\left( g_{\delta \mu ,\alpha }+g_{\delta
\alpha ,\mu }-g_{\alpha \mu
,\delta }\right) \right] ,%
\end{array}
\tag{3.27}  \label{4.34a}
\end{equation}%
\begin{equation}
\begin{array}{l}
S=2g^{12}g^{2\lambda }g_{1\lambda ,2}+\frac{1}{2}g^{12}g^{\lambda
\mu
}\left( 2g_{1\lambda ,\mu }\right) \\
\text{ \ \ \ \ }+\frac{1}{2}g^{2\mu }\left[ g^{\lambda 2}\left(
2g_{\mu \lambda ,2}-g_{\lambda 2,\mu }\right) +g^{\lambda \theta
}\left( 2g_{\mu
\lambda ,\theta }-g_{\lambda \theta ,\mu }\right) \right] , \\
W=\left[ 2g^{12}g^{2\lambda }g_{1\lambda ,2}+g^{12}g^{\lambda \mu
}\left(
g_{1\lambda ,\mu }\right) \right] _{,2} \\
\text{ \ \ \ \ }+\frac{1}{2}\left\{ g^{2\mu }\left[ g^{\lambda
2}\left( 2g_{\mu \lambda ,2}-g_{\lambda 2,\mu }\right) +g^{\lambda
\theta }\left( 2g_{\mu \lambda ,\theta }-g_{\lambda \theta ,\mu
}\right) \right] \right\}
_{,2}, \\
K=2g^{\alpha \beta }g^{2\lambda }F_{2\alpha }.F_{\lambda \beta
}+g^{\alpha \beta }g^{\mu \lambda }F_{\mu \alpha }.F_{\lambda
\beta
}-F_{12}.F^{12}-F_{34}.F^{34} \\
\text{ \ \ \ \ }-F_{2\lambda }.\left[ g^{21}g^{\lambda
2}F_{12}+g^{23}g^{\lambda 2}F_{32}+g^{23}g^{\lambda
4}F_{34}+g^{24}g^{\lambda 2}F_{42}+g^{24}g^{\lambda 3}F_{43}\right] \\
\text{ \ \ \ \ }+g^{\alpha \beta }\left( \Phi _{,\alpha }+\left[
A_{\alpha },\Phi \right] \right) .\left( \Phi _{,\beta }+\left[
A_{\beta },\Phi \right]
\right) +V\left( \Phi ^{2}\right) .%
\end{array}
\tag{3.28}  \label{4.35}
\end{equation}%
\medskip $\left( ii\right) $ The equation%
\begin{equation}
g^{\alpha \beta }R_{\alpha \beta }-2\Gamma
_{,2}^{2}-2g^{12}g_{12,2}\Gamma ^{2}=g^{\alpha \beta }\tau
_{\alpha \beta },  \tag{3.29}  \label{4.36}
\end{equation}%
is equivalent to the following second order ODE on
$G_{T_{1}}^{1}$\ with
unknown $g_{11}$,%
\begin{equation}
-2\left( g^{12}\right) ^{2}g_{11,22}+4\left( g^{12}\right)
^{3}g_{12,2}g_{11,2}+\chi g_{11}+\psi =0,  \tag{3.30}
\label{4.37}
\end{equation}%
where
\begin{equation}
\begin{array}{l}
\chi =4\left( g^{12}\right) ^{4}\left( g_{12,2}\right) ^{2}+\frac{1}{2}%
\left( g^{12}\right) ^{2}\left( g^{\alpha \beta }g_{\alpha \beta
,2}\right)
_{,2}, \\
\psi =\frac{1}{4}g^{\alpha \beta }\left( N_{\alpha \beta
}+M_{\alpha \beta
}\right) -2W-2g^{12}g_{12,2}S-K.%
\end{array}
\tag{3.31}  \label{4.38}
\end{equation}%
\medskip \medskip
\end{proposition}

\begin{proof}
See appendix D.
\end{proof}

The proof of the following statement that provides the construction of $%
g_{11}$ on $G_{T_{1}}^{1}$ is straightforward.

\begin{proposition}
Let $d_{0},$\ $d_{1}\in C^{\infty }\left( \Gamma \right) .$ Then $\left( \ref%
{4.37}\right) $\ has a unique solution $g_{11}\in C^{\infty
}\left(
G_{T_{1}}^{1}\right) $\ satisfying $g_{11}=d_{0}$\ and $g_{11,2}=d_{1}$\ on $%
\Gamma .$
\end{proposition}

In the following proposition the relation $\Gamma ^{2}=0$\ on
$G_{T_{1}}^{1}$ is arranged. Its proof is similar to the proof of
Proposition 3.5.

\begin{proposition}
$\left( i\right) $ On $G_{T_{1}}^{1}$ the reduced system
\begin{equation}
\widetilde{R}_{\alpha \beta }=\tau _{\alpha \beta },  \tag{3.32}
\label{4.39}
\end{equation}%
implies the following homogenous ODE on $G_{T_{1}}^{1}$\ with unknown $%
\Gamma ^{2}$%
\begin{equation}
\Gamma _{,2}^{2}+g^{12}g_{12,2}\Gamma ^{2}=0.  \tag{3.33}
\label{4.40}
\end{equation}%
\medskip $\left( ii\right) $ Assume $\Gamma ^{2}=0$\ on $\Gamma $. Then $%
\Gamma ^{2}=0$\ on $G_{T_{1}}^{1}$.
\end{proposition}

\begin{proof}
Since $g_{23}=g_{24}=0,$\quad $\Gamma ^{1}=\Gamma ^{3}=\Gamma ^{4}=0$ on $%
G^{1}$ at this final step of the construction process, it follows
from the definition of $\widetilde{R}_{ij}$\ (see $\left(
2.22\right) $) that the reduced equation $\left( \ref{4.39}\right)
$ reads $R_{\alpha \beta }=\tau _{\alpha \beta }$.\ Thus, in view
of $\left( \ref{4.36}\right) $, equation
\begin{equation}
g^{\alpha \beta }R_{\alpha \beta }=g^{\alpha \beta }\tau _{\alpha
\beta }, \tag{4.1.30}
\end{equation}%
implies $\left( \ref{4.40}\right) $. $\left( \ref{4.40}\right) $
is a linear homogenous first order ODE on $G_{T_{1}}^{1}$, with
unknown function $\Gamma ^{2}$, of the real variable $x^{2}$, with
$C^{\infty }$ coefficients
depending smoothly on real parameters $x^{3}$ and $x^{4}$. Thus, assuming $%
\Gamma ^{2}=0$ on $\Gamma $ gives $\Gamma ^{2}=0$ on
$G_{T_{1}}^{1}$.
\end{proof}

\begin{remark}
\textit{In the same way }$g_{22}$\textit{\ is constructed on }$G_{T_{1}}^{2}$%
\textit{\ and the relation }$\Gamma ^{1}=0$\textit{\ is established on }$%
G_{T_{1}}^{2}$\textit{.}
\end{remark}

\section{Conclusion and compatibility conditions on\ $\Gamma \equiv
G_{T_{1}}^{1}\cap G_{T_{1}}^{2}$}

We have successfully adapted Rendall method through which, given a
positive
real number $0<T\leq T_{0}$, appropriate free data $\underset{\omega }{h}%
_{\alpha \beta },$ $\underset{\omega }{A}_{\alpha }$\ and $\underset{\omega }%
{\Phi }$\ in $C^{\infty }\left( G_{T}^{\omega }\right) $ and some
adequate conditions, initial data for the reduced
Einstein-Yang-Mills-Higgs system are constructed on $G_{T}^{\omega
},$ $\omega =1,2$. For those data we have established that the
solution of the evolution problem with those initial
data satisfies the relations $\Gamma ^{i}=0$ and $\Delta =0$ on $%
G_{T_{1}}^{\omega }$ for some $T_{1}\in (0,T]$. In fact, setting
$g_{\alpha
\beta }=\Omega h_{\alpha \beta }$ on $G_{T}^{1}\cup G_{T}^{2}$,\ where $%
h_{\alpha \beta }=\underset{\omega }{h}_{\alpha \beta }$ on $G_{T}^{\omega }$%
, $\omega =1,2$, $\left( \underset{\omega }{h}_{\alpha \beta
}\right) $ a symmetric positive definite matrix function with
determinant $1$ at each point of $G_{T}^{\omega },$ $\omega =1,2$,
and $\Omega $ an unknown positive function, we have constructed
$C^{\infty }$ initial data as follows:

$\left( i\right) $ Construction of $g_{\alpha \beta }$, $g_{12}$ on $%
G_{T_{1}}^{1}$ such that $\Gamma ^{1}=0$ and $g_{22,1}=2g_{12,2}$ on $%
G_{T_{1}}^{1}$ under the following conditions:

$%
\begin{array}{l}
g_{22}=g_{23}=g_{24}=0\text{ on }G_{T}^{1},\quad A_{2}=0\text{ on
}G_{T}^{1},
\\
g_{12},\text{ }\Omega \text{ and }\Omega _{,2}\text{ are given }C^{\infty }%
\text{ functions on }\Gamma ,\text{ such that }\Gamma ^{1}=0\text{ and }%
g_{22,1}=2g_{12,2}\text{ on }\Gamma .%
\end{array}%
\medskip $

$\left( ii\right) $ Construction of $g_{\alpha \beta }$, $g_{12}$ on $%
G_{T_{1}}^{2}$ such that $\Gamma ^{2}=0$ and $g_{11,2}=2g_{12,1}$ on $%
G_{T_{1}}^{2}$ under the following supplementary conditions (in
addition to conditions in $\left( i\right) $)

$%
\begin{array}{l}
g_{11}=g_{13}=g_{14}=0\text{ on }G_{T}^{2},\quad A_{1}=0\text{ on
}G_{T}^{2},
\\
\Omega _{,1}\text{ is a given }C^{\infty }\text{ function on
}\Gamma ,\text{
such that }\Gamma ^{2}=0\text{ and }g_{11,2}=2g_{12,1}\text{ on }\Gamma .%
\end{array}%
\medskip $

$\left( iii\right) $ Construction of $g_{1\alpha }$, $A_{1}$ on $%
G_{T_{1}}^{1}$ such that $\Gamma ^{\alpha }=0$ and $\Delta =0$ on $%
G_{T_{1}}^{1}$ under the following supplementary condition (in
addition to
conditions in $\left( i\right) $): $g_{1\alpha ,2}$ and $A_{1,2}$ are given $%
C^{\infty }$ functions on $\Gamma $, such that $\Gamma ^{\alpha }=0$ and $%
\Delta =0$ on $\Gamma $.$\medskip $

$\left( iv\right) $ Construction of $g_{2\alpha }$, $A_{2}$ on $%
G_{T_{1}}^{2} $ such that $\Gamma ^{\alpha }=0$ and $\Delta =0$ on $%
G_{T_{1}}^{2}$ under the following supplementary condition (in
addition to conditions in $\left( ii\right) $): $g_{2\alpha ,1}$
and $A_{2,1}$ are given
$C^{\infty }$ functions on $\Gamma $, such that $\Gamma ^{\alpha }=0$ and $%
\Delta =0$ on $\Gamma $.$\medskip $

$\left( 5i\right) $ Construction of $g_{11}$ on $G_{T_{1}}^{1}$ such that $%
\Gamma ^{2}=0$ on $G_{T_{1}}^{1}$ under the following
supplementary condition (in addition to conditions in $\left(
i\right) $ and $\left( iii\right) $): $g_{11,2}$ is a given
$C^{\infty }$ function on $\Gamma $, such that $\Gamma ^{2}=0$ on
$\Gamma $.$\medskip $

$\left( 6i\right) $ Construction of $g_{22}$ on $G_{T_{1}}^{2}$ such that $%
\Gamma ^{1}=0$ on $G_{T_{1}}^{2}$ under the following
supplementary condition (in addition to conditions in $\left(
ii\right) $ and $\left( iv\right) $): $g_{22,1}$ is a given
$C^{\infty }$ function on $\Gamma $, such that $\Gamma ^{1}=0$ on
$\Gamma $.$\medskip $

We now show how the above adequate conditions in\textbf{\ }$\left(
i\right) , $ $\left( ii\right) ,$ $\left( iii\right) ,$ $\left(
iv\right) ,$ $\left( 5i\right) $ and $\left( 6i\right) $\textbf{\
}are arranged.

Begin by taking $g_{12}=-1$ on $\Gamma $ (This\ is\ a\
non-restrictive property\ that can naturally be imposed to\ any\
metric\ in\ standard\
coordinates, see \cite{24}\ p.\ $232$). Then\ choose $\left( \underset{%
\omega }{h}_{\alpha \beta }\right) $,\ a $C^{\infty }$ symmetric
positive definite matrix function on\ $G_{T}^{\omega }$ with
determinant $1$ at each
point and\ set $g_{\alpha \beta }=\Omega h_{\alpha \beta },$ where $%
h_{\alpha \beta }=\underset{\omega }{h}_{\alpha \beta }$ on $G_{T}^{\omega }$%
, $\omega =1,2$. Let$\ \Omega =v_{0}$ on $\Gamma ,$ where $v_{0}$
is a given $C^{\infty }$ function on $\Gamma $. Take also
\begin{equation*}
\begin{array}{c}
g_{22}=g_{23}=g_{24}=0,\quad A_{2}=0\text{ on }G_{T}^{1}, \\
g_{11}=g_{13}=g_{14}=0,\quad A_{1}=0\text{ on }G_{T}^{2}.%
\end{array}%
\end{equation*}%
Then all the components $g_{ij}$ of the metric are determined on
$\Gamma $ since $g_{\alpha \beta }=\Omega h_{\alpha \beta }$,
$g_{11}=g_{1\alpha }=0$, $g_{22}=g_{2\alpha }=0$ on $\Gamma $,
$\Omega $ and $h_{\alpha \beta }$ are known on $\Gamma $.

Next choose $\underset{\omega }{\Phi },$ $\underset{\omega }{A}_{3}$ and $%
\underset{\omega }{A}_{4},$ which are\ given $C^{\infty }$\ functions on $%
G_{T}^{\omega }$ such that
\begin{equation*}
\underset{1}{\Phi }=\underset{2}{\Phi }\text{ on }\Gamma ,\quad \underset{1}{%
A}_{3}=\underset{2}{A}_{3}\text{ on }\Gamma ,\quad \underset{1}{A}_{4}=%
\underset{2}{A}_{4}\text{ on }\Gamma .
\end{equation*}%
Let $\Omega _{,1}=v_{1}$ and $\Omega _{,2}=v_{2}$ on $\Gamma ,$
where $v_{1}$
and $v_{2}$\ are two given $C^{\infty }$ functions on $\Gamma .$ Then Eqs. $%
\left( 3.6\right) $ and $\left( 3.7\right) $ on $G_{T}^{1}$\ as
well as
their following counterparts on $G_{T}^{2}$%
\begin{equation}
g_{12,1}=\frac{1}{2}g_{12}\frac{\Omega _{,1}}{\Omega }\text{ on\
}G_{T}^{2}, \tag{3.6$\prime $}
\end{equation}%
and

\begin{equation}
\frac{1}{4}g_{,1}^{\alpha \beta }g_{\alpha \beta
,1}-\frac{1}{2}\left( g^{\alpha \beta }g_{\alpha \beta ,1}\right)
_{,1}=\tau _{11}\text{ on}\ G_{T}^{2},  \tag{3.7$\prime $}
\end{equation}%
are\ satisfied and it holds that: (see \cite{24},\ p.\ $233$)
\begin{equation*}
g_{11,2}=g_{22,1}=\frac{1}{4}g^{\alpha \beta }g_{\alpha \beta ,1}\text{ on }%
\Gamma .
\end{equation*}%
This insures
\begin{equation*}
\Gamma ^{1}=\Gamma ^{2}=0\text{ on }\Gamma \equiv G_{T}^{1}\cap
G_{T}^{2}.
\end{equation*}%
Finally, Let $g_{13,2}=\widetilde{\underset{2}{b}}_{3},\quad g_{14,2}=%
\widetilde{\underset{2}{b}}_{4},\quad A_{1,2}=\widetilde{\underset{2}{A}}%
_{1} $ on $\Gamma ,$ where $\widetilde{\underset{2}{b}}_{3},\ \widetilde{%
\underset{2}{b}}_{4},\mathit{\ }\widetilde{\underset{2}{A}}_{1}$ are given $%
C^{\infty }$ functions on $\Gamma $. Then there is only one way to choose $%
g_{2\beta ,1}$ on $\Gamma $ such that $\Gamma ^{3}=\Gamma ^{4}=0$
on $\Gamma
$. In fact, by the definition of $\Gamma ^{\beta }$ (see $\left( \ref{4}%
\right) $), on $G_{T}^{1}$ it holds that

\begin{equation}
g^{12}g_{2\alpha ,1}=g_{\alpha \beta }\Gamma ^{\beta
}-g^{12}g_{1\alpha
,2}-g_{\alpha \beta }g_{,2}^{\beta \lambda }g^{12}g_{1\lambda }-\frac{1}{2}%
\left[ -2g^{12}g_{12,\alpha }+g^{\mu \theta }\left( 2g_{\alpha \mu
,\theta }-g_{\mu \theta ,\alpha }\right) \right] ,  \tag{4.1}
\label{4.41}
\end{equation}%
and on $G_{T}^{2}$ it holds that

\begin{equation}
g^{12}g_{1\alpha ,2}=g_{\alpha \beta }\Gamma ^{\beta
}-g^{12}g_{2\alpha
,1}-g_{\alpha \beta }g_{,1}^{\beta \lambda }g^{12}g_{2\lambda }-\frac{1}{2}%
\left[ -2g^{12}g_{12,\alpha }+g^{\mu \theta }\left( 2g_{\alpha \mu
,\theta }-g_{\mu \theta ,\alpha }\right) \right] .  \tag{4.2}
\label{4.42}
\end{equation}%
Since $g_{1\lambda }=0=g_{2\lambda }$ on $\Gamma $, it follows from $\left( %
\ref{4.41}\right) $ and $\left( \ref{4.42}\right) $ that

\begin{equation}
g^{12}g_{2\alpha ,1}=g_{\alpha \beta }\Gamma ^{\beta }-g^{12}g_{1\alpha ,2}-%
\frac{1}{2}\left[ -2g^{12}g_{12,\alpha }+g^{\mu \theta }\left(
2g_{\alpha \mu ,\theta }-g_{\mu \theta ,\alpha }\right) \right]
\text{ on }\Gamma . \tag{4.3}  \label{4.43}
\end{equation}%
So on $\Gamma ,$\ $\Gamma ^{\beta }=0$\ is equivalent to%
\begin{equation*}
g^{12}g_{2\alpha ,1}=-g^{12}g_{1\alpha ,2}-\frac{1}{2}\left[
-2g^{12}g_{12,\alpha }+g^{\mu \theta }\left( 2g_{\alpha \mu
,\theta }-g_{\mu \theta ,\alpha }\right) \right] .
\end{equation*}%
We now proceed to arrange the condition $\Delta =0$ on $\Gamma $. As $%
A_{2}=0 $ on $G_{T}^{1}$, $A_{1}=0$ on $G_{T}^{2}$, $\underset{\omega }{A}%
_{\alpha }$ are given as $C^{\infty }$ functions on $G_{T}^{\omega
}$, there
is only one way to choose $A_{2,1}$ on $\Gamma $ such that $\Delta =0$ on $%
\Gamma $. In fact, from the definitions of $\Delta $\ and $\Gamma
^{k}$\ (see $\left( 2.20\right) $ and $\left( \ref{4}\right) $),
on $G_{T}^{1}$ it holds that

\begin{equation}
\begin{array}{l}
g^{12}A_{2,1}=\Delta -g^{12}A_{1,2}+\left( 2g^{12}\Gamma
_{12}^{1}+g^{\alpha \beta }\Gamma _{\alpha \beta }^{1}\right)
A_{1}+g^{12}g^{\alpha \lambda
}A_{\alpha }\left( g_{1\lambda ,2}+g_{2\lambda ,1}\right) \\
\text{ \ \ \ \ \ \ \ \ \ \ \ }-\left[ 2g^{12}\Gamma _{2\alpha
}^{\beta }A_{\beta }g^{\alpha \lambda }+2\left( g^{12}\right)
^{2}g_{12,2}g^{\alpha \lambda }A_{\alpha }-g^{12}g^{\alpha \lambda
}A_{\alpha ,2}\right]
g_{1\lambda } \\
\text{ \ \ \ \ \ \ \ \ \ \ \ }-g^{12}g^{\alpha \lambda }A_{\alpha
}g_{12,\lambda }-g^{\alpha \beta }\left( A_{\beta ,\alpha }-\Gamma
_{\alpha
\beta }^{\lambda }A_{\lambda }\right) .%
\end{array}
\tag{4.4}  \label{4.44}
\end{equation}%
Since $g_{1\lambda }=0$\ and $A_{1}=0$\ on $\Gamma $, $\left( \ref{4.44}%
\right) $ implies that on $\Gamma $ the following equality holds

\begin{equation}
g^{12}A_{2,1}=\Delta -g^{12}A_{1,2}+g^{12}g^{\alpha \lambda
}A_{\alpha }\left( g_{1\lambda ,2}+g_{2\lambda ,1}\right)
-g^{12}g^{\alpha \lambda }A_{\alpha }g_{12,\lambda }-g^{\alpha
\beta }\left( A_{\beta ,\alpha }-\Gamma _{\alpha \beta }^{\lambda
}A_{\lambda }\right) .  \tag{4.5} \label{4.45}
\end{equation}%
Hence, on $\Gamma ,$\ $\Delta =0$ is equivalent to%
\begin{equation*}
g^{12}A_{2,1}=-g^{12}A_{1,2}+g^{12}g^{\alpha \lambda }A_{\alpha
}\left( g_{1\lambda ,2}+g_{2\lambda ,1}\right) -g^{12}g^{\alpha
\lambda }A_{\alpha }g_{12,\lambda }-g^{\alpha \beta }\left(
A_{\beta ,\alpha }-\Gamma _{\alpha \beta }^{\lambda }A_{\lambda
}\right) .
\end{equation*}%
It follows from the above discussion that all necessary data are given on $%
\Gamma $ and all necessary assumptions fulfilled.

We can now sum up the $C^{\infty }$\ resolution of the Goursat
problem for the EYMH system in the following theorem, where the
resolution of the evolution problem is a direct consequence of
Theorem 1 of \cite{24} and the constraints problem is solved by
the method just described above in section 3.

\begin{theorem}
Let $T\in (0,T_{0}]$ be a real number and $\omega \in \left\{
1,2\right\} .$
Let $\underset{\omega }{h}_{33},$ $\underset{\omega }{h}_{34},$\ $\underset{%
\omega }{h}_{44}$\ be $C^{\infty }$\ scalar functions on
$G_{T}^{\omega }$\ such that $\left( \underset{\omega }{h}_{\alpha
\beta }\right) $\ is a
symmetric positive definite matrix with determinant $1$\ at each point of $%
G_{T}^{\omega }$ and $\left( \underset{1}{h}_{33},\underset{1}{h}_{34},%
\underset{1}{h}_{44}\right) =\left( \underset{2}{h}_{33},\underset{2}{h}%
_{34},\underset{2}{h}_{44}\right) $\ on $\Gamma $. Let $\widetilde{\underset{%
\omega }{\Phi }},$\ $\underset{\omega }{\widetilde{A}}_{3}$, $\underset{%
\omega }{\widetilde{A}}_{4}$ be $C^{\infty }$\ functions on
$G_{T}^{\omega }$
such that $\left( \widetilde{\underset{1}{\Phi }},\underset{1}{\widetilde{A}}%
_{3},\underset{1}{\widetilde{A}}_{4}\right) =\left( \widetilde{\underset{2}{%
\Phi }},\underset{2}{\widetilde{A}}_{3},\underset{2}{\widetilde{A}}%
_{4}\right) $\ on $\Gamma $.\ Let $C^{\infty }$ functions
$\widetilde{\Omega
},\ \widetilde{\Omega }_{1},\ \widetilde{\Omega }_{2},\ \widetilde{\underset{%
2}{b}}_{3},\ \widetilde{\underset{2}{b}}_{4},\ \widetilde{\underset{2}{A}}%
_{1}$ be given on $\Gamma .$ Then there exists $T_{1}\in (0,T],$ a unique $%
C^{\infty }$\ scalar function $\Omega $\ on $G_{T_{1}}^{1}\cup
G_{T_{1}}^{2}, $ a unique $C^{\infty }$\ Lorentz metric $g_{ij}$\ on $%
L_{T_{1}},$\ a unique $C^{\infty }$\ Yang-Mills potential $A_{k}$\ on $%
L_{T_{1}}$ and a unique $C^{\infty }$\ Higgs function $\Phi $\ on
$L_{T_{1}}$ such that:

$\left( 1\right) $\ $g_{\alpha \beta }=\Omega h_{\alpha \beta }$\ on $%
G_{T_{1}}^{1}\cup G_{T_{1}}^{2},$ where $h_{\alpha \beta }=\underset{\omega }%
{h}_{\alpha \beta }$\textit{on} $G_{T}^{\omega }.$

$\left( 2\right) $\ $u=\left( g_{ij},A_{k},\Phi \right) $\
satisfies the Einstein-Yang-Mills-Higgs equations on $L_{T_{1}},$

$\left( 3\right) $ the given coordinates on $\mathbb{R}^{4}$ are
standard
coordinates for $g_{ij}$ and the Lorentz gauge condition $\nabla _{k}A^{k}=0$%
\ is satisfied on $L_{T_{1}},$ with $A_{2}=0$ on $G_{T_{1}}^{1}$ and $%
A_{1}=0 $ on $G_{T_{1}}^{2},$

$\left( 4\right) $\ $u=\left( g_{ij},A_{k},\Phi \right) $ induce
the given data on $G_{T_{1}}^{1}\cup G_{T_{1}}^{2},$

$\left( 5\right) $ on $\Gamma $ it holds that$:$ $\Omega =\widetilde{\Omega }%
;$ $\Omega _{,1}=\widetilde{\Omega }_{1};$ $\Omega _{,2}=\widetilde{\Omega }%
_{2};$ $g_{13,2}=\widetilde{\underset{2}{b}}_{3};$ $g_{14,2}=\widetilde{%
\underset{2}{b}}_{4}$ and
$A_{1,2}=\widetilde{\underset{2}{A}}_{1}.$
\end{theorem}

\section*{Appendix A: The conservation laws for the stress-energy tensor and
the current}

We first show that if $\left( g_{ij},A_{k},\Phi \right) $ is such
that the YMH system is satisfied then
\begin{equation*}
\nabla ^{j}\rho _{ij}=0,
\end{equation*}
where
\begin{equation}
\rho _{ij}=F_{ik}.F_{j}^{\quad k}-\frac{1}{4}g_{ij}F_{kl}.F^{kl}+\widehat{%
\nabla }_{i}\Phi .\widehat{\nabla }_{j}\Phi
-\frac{1}{2}g_{ij}\left( \widehat{\nabla }_{k}\Phi
.\widehat{\nabla }^{k}\Phi +V\left( \Phi ^{2}\right) \right) .
\tag{A.1}
\end{equation}%
It holds that%
\begin{equation}
\nabla ^{j}\rho _{ij}=\nabla _{p}\left( g^{jp}\rho _{ij}\right) .
\tag{A.2}
\end{equation}%
A direct calculation shows that
\begin{equation}
\begin{array}{l}
\nabla _{p}\left( g^{jp}\rho _{ij}\right) =\left( \nabla _{p}F_{ik}+\frac{1}{%
2}\nabla _{i}F_{kp}\right) .F^{pk}+F_{ik}.\nabla _{p}F^{pk} \\
\text{ \ \ \ \ \ \ \ \ \ \ \ \ \ \ \ \ \ }+\left( \nabla
_{p}\widehat{\nabla
}_{i}\Phi -\nabla _{i}\widehat{\nabla }_{p}\Phi \right) .\widehat{\nabla }%
^{p}\Phi +\widehat{\nabla }_{i}\Phi .\nabla _{p}\widehat{\nabla }^{p}\Phi -%
\frac{1}{2}\nabla _{i}\left[ V\left( \Phi ^{2}\right) \right] .%
\end{array}
\tag{A.3}
\end{equation}%
We have to handle the terms that appear in the right hand side (r.h.s) of $%
\left( A.3\right) $. Using the Bianchi identity
\begin{equation}
\widehat{\nabla }_{i}F_{jk}+\widehat{\nabla }_{j}F_{ki}+\widehat{\nabla }%
_{k}F_{ij}=0,  \tag{A.4}
\end{equation}%
we get, after a simple calculation,
\begin{equation}
\left( \nabla _{p}F_{ik}+\frac{1}{2}\nabla _{i}F_{kp}\right)
.F^{pk}=\left( \frac{1}{2}\left[ A_{i},F_{kp}\right] +\left[
A_{p},F_{ik}\right] \right) .F^{kp}  \tag{A.5}
\end{equation}%
The property of the scalar product $\left( \ref{66}\right) $ on
the Lie
algebra $\mathcal{G}$ and the antisymmetry of the Yang-Mills field $F$ yield%
\begin{equation}
\left[ A_{i},F_{kp}\right] .F^{pk}=A_{i}.\left[
F_{kp},F^{kp}\right] =0. \tag{A.6}
\end{equation}%
From $\left( A.5\right) $ and $\left( A.6\right) $ it follows that
\begin{equation}
\left( \nabla _{p}F_{ik}+\frac{1}{2}\nabla _{i}F_{kp}\right)
.F^{pk}=\left[ A_{p},F_{ik}\right] .F^{kp}.  \tag{A.7}
\end{equation}%
For the term $F_{ik}.\nabla _{p}F^{pk}$ of the r.h.s of $\left( A.3\right) $%
,\ we use the Yang-Mills equations $\widehat{\nabla
}_{p}F^{pk}=J^{k}$ to
have%
\begin{equation}
F_{ik}.\nabla _{p}F^{pk}=F_{ik}.\left( \widehat{\nabla
}_{p}F^{pk}-\left[ A_{p},F^{pk}\right] \right) =F_{ik}.\left(
J^{k}-\left[ A_{p},F^{pk}\right] \right) .  \tag{A.8}
\end{equation}%
From the expression $\left( \ref{216d}\right) $ of the current $J^{k},$ $%
\left( A.8\right) $ yields%
\begin{equation}
F_{ik}.\nabla _{p}F^{pk}=F_{ip}.\left[ \Phi ,\widehat{\nabla }^{p}\Phi %
\right] -F_{ik}.\left[ A_{p},F^{pk}\right] .  \tag{A.9}
\end{equation}%
Adding $\left( A.9\right) $ to $\left( A.7\right) $\ and using
once more the
property of the scalar product $\left( \ref{66}\right) $ on the Lie algebra $%
\mathcal{G}$, we gain
\begin{equation}
\left( \nabla _{p}F_{ik}+\frac{1}{2}\nabla _{i}F_{kp}\right)
.F^{pk}+F_{ik}.\nabla _{p}F^{pk}=F_{ip}.\left[ \Phi ,\widehat{\nabla }%
^{p}\Phi \right] .  \tag{A.10}
\end{equation}%
We now handle the terms $\left( \nabla _{p}\widehat{\nabla
}_{i}\Phi -\nabla
_{i}\widehat{\nabla }_{p}\Phi \right) .\widehat{\nabla }^{p}\Phi $ and $%
\widehat{\nabla }_{i}\Phi .\nabla _{p}\widehat{\nabla }^{p}\Phi -\frac{1}{2}%
\nabla _{i}\left[ V\left( \Phi ^{2}\right) \right] $ of the r.h.s
of $\left( A.3\right) $. Thanks to Jacobi identity, a
straightforward calculation shows that
\begin{equation}
\left( \nabla _{p}\widehat{\nabla }_{i}\Phi -\nabla _{i}\widehat{\nabla }%
_{p}\Phi \right) .\widehat{\nabla }^{p}\Phi =\left( \left[ F_{pi},\Phi %
\right] +\left[ A_{i},\widehat{\nabla }_{p}\Phi \right] -\left[ A_{p},%
\widehat{\nabla }_{i}\Phi \right] \right) .\widehat{\nabla
}^{p}\Phi . \tag{A.11}
\end{equation}%
For $\widehat{\nabla }_{i}\Phi .\nabla _{p}\widehat{\nabla
}^{p}\Phi $, we use the Higgs equations $\widehat{\nabla
}_{p}\widehat{\nabla }^{p}\Phi =H$
to have%
\begin{equation}
\widehat{\nabla }_{i}\Phi .\nabla _{p}\widehat{\nabla }^{p}\Phi =\widehat{%
\nabla }_{i}\Phi .\left( \widehat{\nabla }_{p}\widehat{\nabla }^{p}\Phi -%
\left[ A_{p},\widehat{\nabla }^{p}\Phi \right] \right) =\widehat{\nabla }%
_{i}\Phi .\left( H\left( \Phi \right) -\left[ A_{p},\widehat{\nabla }%
^{p}\Phi \right] \right) .  \tag{A.12}
\end{equation}%
From the expression $\left( \ref{220}\right) $\ of $H,$\ $\left(
A.12\right) $ implies
\begin{equation}
\widehat{\nabla }_{i}\Phi .\nabla _{p}\widehat{\nabla }^{p}\Phi =\widehat{%
\nabla }_{i}\Phi .\left( V^{\prime }\left( \Phi ^{2}\right) \Phi
-\left[ A_{p},\widehat{\nabla }^{p}\Phi \right] \right) .
\tag{A.13}
\end{equation}%
For the term $\nabla _{i}\left[ V\left( \Phi ^{2}\right) \right] $
we easily
get%
\begin{equation}
\nabla _{i}\left[ V\left( \Phi ^{2}\right) \right] =2V^{\prime
}\left( \Phi ^{2}\right) \Phi .\Phi _{,i}=2V^{\prime }\left( \Phi
^{2}\right) \Phi .\left( \widehat{\nabla }_{i}\Phi -\left[
A_{i},\Phi \right] \right) . \tag{A.14}
\end{equation}%
From $\left( A.11\right) $, $\left( A.13\right) $ and $\left(
A.14\right) $,
we obtain%
\begin{equation}
\left( \nabla _{p}\widehat{\nabla }_{i}\Phi -\nabla _{i}\widehat{\nabla }%
_{p}\Phi \right) .\widehat{\nabla }^{p}\Phi +\widehat{\nabla
}_{i}\Phi .\nabla _{p}\widehat{\nabla }^{p}\Phi -\frac{1}{2}\nabla
_{i}\left[ V\left(
\Phi ^{2}\right) \right] =\left[ F_{pi},\Phi \right] .\widehat{\nabla }%
^{p}\Phi .  \tag{A.15}
\end{equation}%
In view of $\left( A.10\right) $ and $\left( A.15\right) $, we have%
\begin{equation}
\begin{array}{l}
\left( \nabla _{p}F_{ik}+\frac{1}{2}\nabla _{i}F_{kp}\right)
.F^{pk}+F_{ik}.\nabla _{p}F^{pk} \\
+\left( \nabla _{p}\widehat{\nabla }_{i}\Phi -\nabla _{i}\widehat{\nabla }%
_{p}\Phi \right) .\widehat{\nabla }^{p}\Phi +\widehat{\nabla
}_{i}\Phi .\nabla _{p}\widehat{\nabla }^{p}\Phi -\frac{1}{2}\nabla
_{i}\left[ V\left(
\Phi ^{2}\right) \right] =0.%
\end{array}
\tag{A.16}
\end{equation}%
The conservation law $\nabla ^{j}\rho _{ij}=0$ now follows from
$\left( A.2\right) ,$\ $\left( A.3\right) $\ and $\left(
A.16\right) $.

The proof of the conservation law $\widehat{\nabla }_{k}J^{k}=0$
for the current $J^{k}=\left[ \Phi ,\widehat{\nabla }^{k}\Phi
\right] $ is obvious. From $\widehat{\nabla }_{k}\widehat{\nabla
}^{k}\Phi =H=V^{\prime }\left(
\Phi ^{2}\right) \Phi $ and $\left[ \widehat{\nabla }_{k}\Phi ,\widehat{%
\nabla }^{k}\Phi \right] =0$,\ a direct calculation yields
\begin{equation*}
\widehat{\nabla }_{k}J^{k}=\left[ \widehat{\nabla }_{k}\Phi
,\widehat{\nabla
}^{k}\Phi \right] +\left[ \Phi ,\widehat{\nabla }_{k}\widehat{\nabla }%
^{k}\Phi \right] =\left[ \Phi ,H\right] =\left[ \Phi ,V^{\prime
}\left( \Phi ^{2}\right) \Phi \right] =V^{\prime }\left( \Phi
^{2}\right) \left[ \Phi ,\Phi \right] =0.
\end{equation*}

\section*{Appendix B: Proof of relation $\left( \protect\ref{4.16}\right) $}

By virtue of $\left( \ref{217}\right) $, it holds that%
\begin{equation}
R_{22}=\Gamma _{22,k}^{k}-\Gamma _{2k,2}^{k}+\Gamma
_{kl}^{k}\Gamma _{22}^{l}-\Gamma _{2l}^{k}\Gamma _{2k}^{l}.
\tag{B.1}
\end{equation}%
We compute each term of the r.h.s of $\left( B.1\right) $\ on
$G^{1}$ by
using the conditions $\left( \ref{4.14}\right) $ and $\left( \ref{4.15}%
\right) $ to gain%
\begin{eqnarray*}
2\Gamma _{22,1}^{1} &=&g_{,1}^{1m}\left( 2g_{2m,2}-g_{22,m}\right)
+g^{1m}\left( 2g_{2m,21}-g_{22,m1}\right) \\
&=&g_{,1}^{11}\left( 2g_{21,2}-g_{22,1}\right) +g^{12}\left(
g_{22,21}\right) ,
\end{eqnarray*}%
\begin{eqnarray*}
2\Gamma _{22,2}^{2} &=&g_{,2}^{2m}\left( 2g_{2m,2}-g_{22,m}\right)
+g^{2m}\left( 2g_{2m,22}-g_{22,m2}\right) \\
&=&g_{,2}^{21}\left( 2g_{21,2}-g_{22,1}\right) +g^{21}\left(
2g_{21,22}-g_{22,12}\right) ,
\end{eqnarray*}%
\begin{equation*}
2\Gamma _{22,\alpha }^{\alpha }=g_{,\alpha }^{\alpha m}\left(
2g_{2m,2}-g_{22,m}\right) +g^{\alpha m}\left( 2g_{2m,2\alpha
}-g_{22,m\alpha }\right) =0,\text{\quad }\alpha =3,4.
\end{equation*}%
Thus%
\begin{equation}
2\Gamma _{22,k}^{k}=\left( g_{,1}^{11}+g_{,2}^{12}\right) \left(
2g_{12,2}-g_{22,1}\right) +2g^{12}g_{12,22}.  \tag{B.2}
\end{equation}%
By expanding the equalities $\left( g^{1i}g_{2i}\right) _{,1}=0$
and $\left(
g^{2i}g_{2i}\right) _{,2}=0$ on $G^{1}$,\ we get the respective equalities%
\begin{equation}
g_{,1}^{11}=-\left( g^{12}\right) ^{2}g_{22,1},  \tag{B.3}
\end{equation}%
and%
\begin{equation}
g_{,2}^{12}=-\left( g^{12}\right) ^{2}g_{12,2}.  \tag{B.4}
\end{equation}%
Then, considering $\left( B.2\right) $, $\left( B.3\right) $ and
$\left(
B.4\right) $, we obtain%
\begin{equation}
2\Gamma _{22,k}^{k}=-\left( g^{12}\right) ^{2}\left(
g_{12,2}+g_{22,1}\right) \left( 2g_{12,2}-g_{22,1}\right)
+2g^{12}g_{12,22}. \tag{B.5}
\end{equation}%
Similarly we obtain
\begin{equation}
2\Gamma _{2k,2}^{k}=-2\left( g^{12}g_{12,2}\right)
^{2}+2g^{12}g_{12,22}+\left( g^{\alpha \beta }g_{\alpha \beta
,2}\right) _{,2}.  \tag{B.6}
\end{equation}%
$\left( B.5\right) $ and $\left( B.6\right) $ yield%
\begin{equation}
\Gamma _{22,k}^{k}-\Gamma _{2k,2}^{k}=\frac{1}{2}\left(
g^{12}\right) ^{2}g_{22,1}\left( g_{22,1}-g_{12,2}\right)
-\frac{1}{2}\left( g^{\alpha \beta }g_{\alpha \beta ,2}\right)
_{,2}.  \tag{B.7}
\end{equation}%
It also holds that
\begin{equation*}
4\Gamma _{kl}^{k}\Gamma _{22}^{l}=4\left( \Gamma _{k1}^{k}\Gamma
_{22}^{1}+\Gamma _{k2}^{k}\Gamma _{22}^{2}+\Gamma _{k\alpha
}^{k}\Gamma _{22}^{\alpha }\right) .
\end{equation*}%
Straightforward computations on $G^{1}$ give%
\begin{equation*}
\begin{array}{l}
2\Gamma _{22}^{1}=0,\quad 2\Gamma _{22}^{2}=g^{12}\left(
2g_{12,2}-g_{22,1}\right) ,\quad 2\Gamma _{22}^{\alpha }=0, \\
2\Gamma _{12}^{1}=g^{12}g_{22,1},\quad 2\Gamma _{2\alpha }^{\alpha
}=g^{\alpha \beta }g_{\alpha \beta ,2}.%
\end{array}%
\end{equation*}%
This implies
\begin{equation*}
2\Gamma _{k2}^{k}=2g^{12}g_{12,2}+g^{\alpha \beta }g_{\alpha \beta
,2}.
\end{equation*}%
Thus%
\begin{equation}
4\Gamma _{kl}^{k}\Gamma _{22}^{l}=g^{12}\left(
2g_{12,2}-g_{22,1}\right) \left( 2g^{12}g_{12,2}+g^{\alpha \beta
}g_{\alpha \beta ,2}\right) . \tag{B.8}
\end{equation}%
In the same way we have
\begin{equation*}
4\Gamma _{2l}^{k}\Gamma _{2k}^{l}=4\left( \Gamma _{12}^{k}\Gamma
_{2k}^{1}+\Gamma _{22}^{k}\Gamma _{2k}^{2}+\Gamma _{2\alpha
}^{k}\Gamma _{2k}^{\alpha }\right) ,
\end{equation*}%
with%
\begin{eqnarray*}
4\Gamma _{12}^{k}\Gamma _{2k}^{1} &=&4\left( \Gamma
_{12}^{1}\Gamma _{12}^{1}+\Gamma _{12}^{2}\Gamma _{22}^{1}+\Gamma
_{12}^{\alpha }\Gamma
_{2\alpha }^{1}\right) =\left( g^{12}g_{22,1}\right) ^{2}, \\
4\Gamma _{22}^{k}\Gamma _{2k}^{2} &=&4\left( \Gamma
_{22}^{1}\Gamma _{12}^{2}+\Gamma _{22}^{2}\Gamma _{22}^{2}+\Gamma
_{22}^{\alpha }\Gamma _{2\alpha }^{2}\right) =\left[ g^{12}\left(
2g_{12,2}-g_{22,1}\right) \right]
^{2}, \\
4\Gamma _{2\alpha }^{k}\Gamma _{2k}^{\alpha } &=&4\left( \Gamma
_{2\alpha }^{1}\Gamma _{12}^{\alpha }+\Gamma _{2\alpha }^{2}\Gamma
_{22}^{\alpha }+\Gamma _{2\alpha }^{\beta }\Gamma _{2\beta
}^{\alpha }\right) =\left( g^{\beta \lambda }g_{\lambda \alpha
,2}\right) \left( g^{\alpha \mu }g_{\mu \beta ,2}\right) ,
\end{eqnarray*}%
as simple calculation on $G^{1}$ shows that%
\begin{eqnarray*}
\Gamma _{22}^{1} &=&\Gamma _{2\alpha }^{1}=\Gamma _{22}^{\alpha
}=0,\quad 2\Gamma _{12}^{1}=g^{12}g_{22,1},\quad 2\Gamma
_{22}^{2}=g^{12}\left(
2g_{12,2}-g_{22,1}\right) , \\
2\Gamma _{2\alpha }^{\beta } &=&g^{\beta m}\left( g_{m\alpha
,2}+g_{2m,\alpha }-g_{2\alpha ,m}\right) =g^{\beta \lambda
}g_{\lambda \alpha ,2}.
\end{eqnarray*}%
Now the following relations hold
\begin{equation*}
\left( g^{\beta \lambda }g_{\lambda \alpha }\right)
_{,2}=0\Leftrightarrow g^{\beta \lambda }g_{\lambda \alpha
,2}+g_{,2}^{\beta \lambda }g_{\lambda \alpha }=0\Leftrightarrow
g^{\beta \lambda }g_{\lambda \alpha ,2}=-g_{,2}^{\beta \lambda
}g_{\lambda \alpha }.
\end{equation*}%
Therefore
\begin{equation*}
4\Gamma _{2\alpha }^{k}\Gamma _{2k}^{\alpha }=-g_{,2}^{\beta
\lambda }g_{\lambda \alpha }g^{\alpha \mu }g_{\mu \beta
,2}=-g_{,2}^{\beta \lambda }g_{\lambda \beta ,2}.
\end{equation*}%
Thus%
\begin{equation}
4\Gamma _{2l}^{k}\Gamma _{2k}^{l}=\left( g^{12}g_{22,1}\right)
^{2}+\left[ g^{12}\left( 2g_{12,2}-g_{22,1}\right) \right]
^{2}-g_{,2}^{\beta \lambda }g_{\lambda \beta ,2}.  \tag{B.9}
\end{equation}%
$\left( B.8\right) $ and $\left( B.9\right) $ give%
\begin{align}
\Gamma _{kl}^{k}\Gamma _{22}^{l}-\Gamma _{2l}^{k}\Gamma _{2k}^{l}& =\frac{1}{%
4}g^{12}g^{\alpha \beta }g_{\alpha \beta ,2}\left(
2g_{12,2}-g_{22,1}\right)
\notag \\
& +\frac{1}{2}\left( g^{12}\right) ^{2}g_{22,1}\left(
g_{12,2}-g_{22,1}\right) +\frac{1}{4}g_{,2}^{\beta \lambda
}g_{\lambda \beta ,2}.  \tag{B.10}
\end{align}%
In view of $\left( B.1\right) $, $\left( B.7\right) $ and $\left(
B.10\right) $, we finally gain the\ first\ equality\ of\ $\left( \ref{4.16}%
\right) $.

On\ the\ other\ hand, in view of\ $\left( 2.22\right) $, we\ have
\begin{equation}
\tau _{22}=F_{2k}.F_{2}^{\text{\quad }k}-\frac{1}{4}g_{22}F_{kl}.F^{kl}+%
\widehat{\nabla }_{2}\Phi .\widehat{\nabla }_{2}\Phi +\frac{1}{2}%
g_{22}V\left( \Phi ^{2}\right) .  \tag{B.11}
\end{equation}%
On $G^{1}$, given that (see $\left( \ref{4.14}\right) $ and $\left( \ref%
{4.15}\right) $) $g^{11}=g^{1\alpha }=0,$ $g_{22}=g_{2\alpha }=0$, and $%
A_{2}=0$, $\left( B.11\right) $ yields the\ second\ equality\ of $\left( \ref%
{4.16}\right) $ by\ a\ direct\ and\ simple\ calculation.

\section*{Appendix C: Proof of Proposition 3.7}

\textbf{Proof of item }$\left( i\right) $

\textbf{Proof of\ }$\left( \ref{4.24}\right) $\textbf{. }By definition%
\begin{equation}
R_{2\alpha }=\Gamma _{2\alpha ,k}^{k}-\Gamma _{2k,\alpha
}^{k}+\Gamma _{lk}^{k}\Gamma _{2\alpha }^{l}-\Gamma _{l\alpha
}^{k}\Gamma _{2k}^{l}. \tag{C.1}
\end{equation}%
We compute each term of the r.h.s of $\left( C.1\right) $\ on
$G^{1}$ by
using the conditions $\left( \ref{4.14}\right) $ and $\left( \ref{4.15}%
\right) $ to gain
\begin{align}
2\Gamma _{2\alpha ,k}^{k}& =\left( g_{,1}^{11}+g_{,2}^{12}\right)
\left( g_{12,\alpha }+g_{1\alpha ,2}-g_{2\alpha ,1}\right)
+g^{12}\left( g_{22,1\alpha }+g_{12,2\alpha }+g_{1\alpha
,22}-g_{2\alpha ,12}\right)
\notag \\
& +\left( g_{,1}^{1\beta }+g_{,2}^{2\beta }\right) g_{\alpha \beta
,2}+g^{2\beta }g_{\alpha \beta ,22}+\left( g^{\lambda \beta
}g_{\alpha \beta ,2}\right) _{,\lambda }.  \tag{C.1a}
\end{align}%
By expanding the equalities $\left( g^{i\beta }g_{2i}\right) _{,1}=0$ and $%
\left( g^{i\beta }g_{1i}\right) _{,2}=0$ on $G^{1}$,\ we get
\begin{equation}
g_{,1}^{1\beta }=-g^{12}\left( 2g^{2\beta }g_{12,2}+g^{\lambda
\beta }g_{2\lambda ,1}\right) ,\quad g_{,2}^{2\beta
}=-g^{12}\left( g^{2\beta }g_{12,2}+g^{\lambda \beta }g_{1\lambda
,2}+g_{,2}^{\lambda \beta }g_{1\lambda }\right) .  \tag{C.1b}
\end{equation}%
$\left( B.3\right) ,$ $\left( B.4\right) ,$ $\left( C.1a\right) $\ and $%
\left( C.1b\right) $\ yield%
\begin{align}
2\Gamma _{2\alpha ,k}^{k}& =-3\left( g^{12}\right)
^{2}g_{12,2}\left(
g_{12,\alpha }+g_{1\alpha ,2}-g_{2\alpha ,1}\right)   \notag \\
& +g^{12}\left( g_{22,1\alpha }+g_{12,2\alpha }+g_{1\alpha
,22}-g_{2\alpha
,12}\right)   \notag \\
& -g^{12}\left[ 3g^{2\beta }g_{12,2}+g^{\lambda \beta }\left(
g_{1\lambda ,2}+g_{2\lambda ,1}\right) +g_{,2}^{\lambda \beta
}g_{1\lambda }\right]
g_{\alpha \beta ,2}  \notag \\
& +g^{2\beta }g_{\alpha \beta ,22}+\left( g^{\lambda \beta
}g_{\alpha \beta ,2}\right) _{,\lambda }.  \tag{C.2a}
\end{align}%
We now compute\ $\Gamma _{2k}^{k}$ and $\Gamma _{2k,\alpha }^{k}$ on\ $G^{1}$%
\ by using equalities $g^{\beta \lambda }g_{\lambda \beta
,2}=4g^{12}g_{12,2}
$ and $2g_{12,2}-g_{22,1}=0$ (see proof of Proposition 3.5),\ to have%
\begin{equation}
2\Gamma _{2k}^{k}=6g^{12}g_{12,2},\quad \Gamma _{2k,\alpha
}^{k}=3\left( g^{12}g_{12,2}\right) _{,\alpha }.  \tag{C.2b}
\end{equation}%
$\left( C.2a\right) $ and $\left( C.2b\right) $ give%
\begin{align}
\Gamma _{2\alpha ,k}^{k}-\Gamma _{2k,\alpha }^{k}&
=-\frac{3}{2}\left(
g^{12}\right) ^{2}g_{12,2}\left( g_{1\alpha ,2}-g_{2\alpha ,1}\right) +\frac{%
1}{2}g^{12}\left( g_{1\alpha ,22}-g_{2\alpha ,12}\right)   \notag \\
& -\frac{1}{2}g^{12}\left[ 3g^{2\beta }g_{12,2}+g^{\lambda \beta
}\left(
g_{1\lambda ,2}+g_{2\lambda ,1}\right) +g_{,2}^{\lambda \beta }g_{1\lambda }%
\right] g_{\alpha \beta ,2}  \notag \\
& +\frac{1}{2}g^{2\beta }g_{\alpha \beta ,22}+\frac{1}{2}\left(
g^{\lambda \beta }g_{\alpha \beta ,2}\right) _{,\lambda }-3\left(
g^{12}g_{12,2}\right)
_{,\alpha }  \notag \\
& -\frac{3}{2}\left( g^{12}\right) ^{2}g_{12,2}g_{12,\alpha }+\frac{1}{2}%
g^{12}\left( g_{22,1\alpha }+g_{12,2\alpha }\right) .  \tag{C.2c}
\end{align}%
In addition, direct calculations on $G^{1}$ give%
\begin{align}
4\Gamma _{lk}^{k}\Gamma _{2\alpha }^{l}& =6g^{12}g_{12,2}\left[
g^{21}\left( g_{21,\alpha }+g_{1\alpha ,2}-g_{2\alpha ,1}\right)
+g^{2\beta }g_{\beta
\alpha ,2}\right]   \notag \\
& +g^{\beta \lambda }g_{\lambda \alpha ,2}\left[
2g^{12}g_{12,\beta }+g^{\lambda \mu }\left( g_{\mu \beta ,\lambda
}+g_{\lambda \mu ,\beta }-g_{\beta \lambda ,\mu }\right) \right] ,
\tag{C.3a}
\end{align}%
and
\begin{align}
4\Gamma _{l\alpha }^{k}\Gamma _{2k}^{l}& =2\left( g^{12}\right)
^{2}g_{12,2}\left( g_{12,\alpha }+g_{2\alpha ,1}-g_{1\alpha
,2}\right)
\notag \\
& -g^{12}g_{\lambda \alpha ,2}\left[ 2g^{2\lambda
}g_{12,2}+g^{\lambda \beta
}\left( g_{1\beta ,2}+g_{2\beta ,1}-g_{12,\beta }\right) \right]   \notag \\
& +g^{\beta \theta }g_{\theta \alpha ,2}\left[ g^{12}\left(
g_{1\beta ,2}+g_{12,\beta }-g_{2\beta ,1}\right) +g^{2\mu }g_{\mu
\beta ,2}\right]
\notag \\
& +g^{\lambda \theta }g_{\theta \beta ,2}\left[ -g^{\beta
2}g_{\lambda \alpha ,2}+g^{\beta \mu }\left( g_{\lambda \mu
,\alpha }+g_{\mu \alpha ,\lambda }-g_{\lambda \alpha ,\mu }\right)
\right] .  \tag{C.3b}
\end{align}%
From the relation $\left( g^{\beta \lambda }g_{\alpha \beta }\right) _{,2}=0$%
, $\left( C.3a\right) $ and $\left( C.3b\right) $ yield
\begin{align}
\Gamma _{lk}^{k}\Gamma _{2\alpha }^{l}-\Gamma _{l\alpha
}^{k}\Gamma _{2k}^{l}& =2\left( g^{12}\right) ^{2}g_{12,2}\left(
g_{1\alpha
,2}-g_{2\alpha ,1}\right)   \notag \\
& +2g^{12}g_{12,2}g^{2\beta }g_{\beta \alpha
,2}+\frac{1}{2}g^{\beta \lambda
}g_{\alpha \beta ,2}g_{2\lambda ,1}  \notag \\
& +\frac{1}{2}g_{,2}^{\beta \lambda }\left( g_{\lambda \beta
,\alpha }+g_{\lambda \alpha ,\beta }\right) +\left( g^{12}\right)
^{2}g_{12,2}g_{12,\alpha }.  \tag{C.3c}
\end{align}%
In view of $\left( C.1\right) $,$\ \left( C.2c\right) $ and
$\left(
C.3c\right) $ give%
\begin{align}
R_{2\alpha }& =\frac{1}{2}\left( g^{12}\right) ^{2}g_{12,2}\left(
g_{1\alpha ,2}-g_{2\alpha ,1}\right) +\frac{1}{2}g^{12}\left(
g_{1\alpha
,22}-g_{2\alpha ,12}\right)   \notag \\
& +\frac{1}{2}g^{12}g_{12,2}g^{2\beta }g_{\beta \alpha ,2}-\frac{1}{2}g^{12}%
\left[ g^{\lambda \beta }g_{1\lambda ,2}+g_{,2}^{\lambda \beta }g_{1\lambda }%
\right] g_{\alpha \beta ,2}  \notag \\
& +\frac{1}{2}g^{2\beta }g_{\alpha \beta ,22}+\frac{1}{2}\left(
g^{\lambda \beta }g_{\alpha \beta ,2}\right) _{,\lambda }-3\left(
g^{12}g_{12,2}\right)
_{,\alpha }  \notag \\
& -\frac{1}{2}\left( g^{12}\right) ^{2}g_{12,2}g_{12,\alpha }+\frac{1}{2}%
g^{12}\left( g_{22,1\alpha }+g_{12,2\alpha }\right) +\frac{1}{2}%
g_{,2}^{\beta \lambda }\left( g_{\lambda \beta ,\alpha
}+g_{\lambda \alpha ,\beta }\right) .  \tag{C.4}
\end{align}

\textbf{Calculation of }$\tau _{2\alpha }$

In view of\ $\left( 2.22\right) \ $it holds that%
\begin{equation}
\tau _{ij}=F_{ik}.F_{j}^{\text{\quad }k}-\frac{1}{4}g_{ij}F_{kl}.F^{kl}+%
\widehat{\nabla }_{i}\Phi .\widehat{\nabla }_{j}\Phi +\frac{1}{2}%
g_{ij}V\left( \Phi ^{2}\right) .  \tag{C.5a}
\end{equation}%
This gives, since $g_{2\alpha }=0$ on $G^{1}$,
\begin{equation}
\tau _{2\alpha }=F_{2k}F_{\alpha }^{\text{\quad }k}+\widehat{\nabla }%
_{2}\Phi .\widehat{\nabla }_{\alpha }\Phi .  \tag{C.5b}
\end{equation}%
$\left( C.5b\right) $ reads%
\begin{equation}
\tau _{2\alpha }=F_{21}.F_{\alpha }^{\text{\quad }1}+F_{2\beta
}.F_{\alpha }^{\text{\quad }\beta }+\widehat{\nabla }_{2}\Phi
.\widehat{\nabla }_{\alpha }\Phi ,  \tag{C.5c}
\end{equation}%
with%
\begin{align}
F_{21}& =A_{1,2}-A_{2,1}+\left[ A_{2},A_{1}\right] ,\quad F_{\alpha }^{\text{%
\quad }1}=g^{1i}F_{\alpha i}=g^{12}F_{\alpha 2},  \notag \\
F_{\alpha }^{\text{\quad }\beta }& =g^{\beta i}F_{\alpha
i}=g^{2\beta }F_{\alpha 2}+g^{\beta \lambda }F_{\alpha \lambda
}=-g^{12}g^{\beta \lambda }g_{1\lambda }F_{\alpha 2}+g^{\beta
\lambda }F_{\alpha \lambda }.  \tag{C.6}
\end{align}%
$\left( C.5c\right) $ and $\left( C.6\right) $ give%
\begin{align}
\tau _{2\alpha }& =-g^{12}F_{2\alpha }.\left(
A_{1,2}-A_{2,1}+\left[ A_{2},A_{1}\right] \right) +g^{12}g^{\beta
\lambda }F_{2\alpha }.F_{2\beta
}g_{1\lambda }-g^{\beta \lambda }F_{\lambda \alpha }.F_{2\beta }  \notag \\
& +\left( \Phi _{,2}+\left[ A_{2},\Phi \right] \right) .\left(
\Phi _{,\alpha }+\left[ A_{\alpha },\Phi \right] \right) .
\tag{C.7}
\end{align}%
Now
\begin{equation}
F_{2\alpha }=A_{\alpha ,2}-A_{2,\alpha }+\left[ A_{2},A_{\alpha }\right] ,%
\text{ }F_{\lambda \alpha }=A_{\alpha ,\lambda }-A_{\lambda
,\alpha }+\left[ A_{\lambda },A_{\alpha }\right] .  \tag{C.8}
\end{equation}%
In\ view\ of\ $\left( C.8\right) $,\ the assumption $A_{2}=0$ on
$G^{1}$ implies that $F_{2\alpha }=A_{\alpha ,2}$ and $\left[
A_{2},A_{1}\right] =$\
$\left[ A_{2},\Phi \right] =0$ on $G^{1}$. It therefore follows from\ $%
\left( C.7\right) $\ and $\left( C.8\right) $\ that%
\begin{align}
\tau _{2\alpha }& =-g^{12}A_{\alpha ,2}\left(
A_{1,2}-A_{2,1}\right)
+g^{12}g^{\beta \lambda }A_{\alpha ,2}.A_{\beta ,2}g_{1\lambda }  \notag \\
& -g^{\beta \lambda }A_{\beta ,2}\left( A_{\alpha ,\lambda
}-A_{\lambda ,\alpha }+\left[ A_{\lambda },A_{\alpha }\right]
\right) +\left( \Phi _{,2}\right) .\left( \Phi _{,\alpha }+\left[
A_{\alpha },\Phi \right] \right) .  \tag{C.9}
\end{align}%
Similarly, the following equality holds on\textit{\ }$G^{1}$:%
\begin{equation}
\Gamma ^{\beta }=g^{\beta \lambda }g^{12}\left( g_{\lambda
1,2}+g_{\lambda 2,1}\right) +g_{,2}^{\beta \lambda
}g^{12}g_{1\lambda }+\frac{1}{2}\left[ -2g^{\beta \lambda
}g^{12}g_{12,\lambda }+g^{\beta \lambda }g^{\mu \theta }\left(
2g_{\lambda \mu ,\theta }-g_{\mu \theta ,\lambda }\right) \right]
. \tag{C.10}
\end{equation}%
Thus%
\begin{align}
\Gamma _{,2}^{\beta }& =g_{,2}^{\beta \lambda }g^{12}\left(
g_{\lambda 1,2}+g_{\lambda 2,1}\right) +g^{\beta \lambda }\left[
g_{,2}^{12}\left( g_{\lambda 1,2}+g_{\lambda 2,1}\right)
+g^{12}\left( g_{\lambda
1,22}+g_{\lambda 2,12}\right) \right]  \notag \\
& +g_{,22}^{\beta \lambda }g^{12}g_{1\lambda }+g_{,2}^{\beta
\lambda }\left[ g_{,2}^{12}g_{1\lambda }+g^{12}g_{1\lambda
,2}\right] +\frac{1}{2}\left[ -2g^{\beta \lambda
}g^{12}g_{12,\lambda }+g^{\beta \lambda }g^{\mu \theta }\left(
2g_{\lambda \mu ,\theta }-g_{\mu \theta ,\lambda }\right) \right]
_{,2}.  \tag{C.11}
\end{align}%
Replacing $g_{,2}^{12}$ by its expression given by $\left( B.4\right) $ gives%
\begin{align}
\Gamma _{,2}^{\beta }& =g_{,2}^{\beta \lambda }g^{12}\left(
g_{\lambda 1,2}+g_{\lambda 2,1}\right) +g^{\beta \lambda }\left[
-\left( g^{12}\right) ^{2}g_{12,2}\left( g_{\lambda
1,2}+g_{\lambda 2,1}\right) +g^{12}\left(
g_{\lambda 1,22}+g_{\lambda 2,12}\right) \right]  \notag \\
& +g_{,22}^{\beta \lambda }g^{12}g_{1\lambda }+g_{,2}^{\beta
\lambda }\left[ -\left( g^{12}\right) ^{2}g_{12,2}g_{1\lambda
}+g^{12}g_{1\lambda ,2}\right]
\notag \\
& +\frac{1}{2}\left[ -2g^{\beta \lambda }g^{12}g_{12,\lambda
}+g^{\beta \lambda }g^{\mu \theta }\left( 2g_{\lambda \mu ,\theta
}-g_{\mu \theta ,\lambda }\right) \right] _{,2}.  \tag{C.12}
\end{align}%
Then
\begin{align}
g_{\alpha \beta }\Gamma _{,2}^{\beta }& =g_{\alpha \beta
}g_{,2}^{\beta \lambda }g^{12}\left( g_{\lambda 1,2}+g_{\lambda
2,1}\right) +g_{\alpha \beta }g^{\beta \lambda }\left[ -\left(
g^{12}\right) ^{2}g_{12,2}\left( g_{\lambda 1,2}+g_{\lambda
2,1}\right) +g^{12}\left( g_{\lambda
1,22}+g_{\lambda 2,12}\right) \right]  \notag \\
& +g_{\alpha \beta }g_{,22}^{\beta \lambda }g^{12}g_{1\lambda
}+g_{\alpha \beta }g_{,2}^{\beta \lambda }\left[ -\left(
g^{12}\right)
^{2}g_{12,2}g_{1\lambda }+g^{12}g_{1\lambda ,2}\right]  \notag \\
& +\frac{1}{2}g_{\alpha \beta }\left[ -2g^{\beta \lambda
}g^{12}g_{12,\lambda }+g^{\beta \lambda }g^{\mu \theta }\left(
2g_{\lambda \mu ,\theta }-g_{\mu \theta ,\lambda }\right) \right]
_{,2}.  \tag{C.13}
\end{align}%
Using $\left( C.4\right) $ and $\left( C.13\right) $, we finally gain%
\begin{align}
R_{2\alpha }+\frac{1}{2}g_{\alpha \beta }\Gamma _{,2}^{\beta }&
=g^{12}g_{1\alpha ,22}-\left( g^{12}\right) ^{2}g_{12,2}g_{2\alpha ,1}-\frac{%
1}{2}g_{\alpha \beta ,2}g^{\beta \lambda }g^{12}\left(
3g_{1\lambda
,2}+g_{2\lambda ,1}\right)  \notag \\
& +\frac{1}{2}g_{\alpha \beta }\left( g^{12}g_{,2}^{\beta \lambda
}\right) _{,2}g_{1\lambda }-\frac{1}{2}g^{12}g_{,2}^{\lambda \beta
}g_{\alpha \beta ,2}g_{1\lambda }-\frac{1}{2}\left( g^{12}\right)
^{2}g_{12,2}g_{\alpha \beta
,2}g^{\beta \lambda }g_{1\lambda }  \notag \\
& -\frac{1}{2}g^{12}g^{\beta \lambda }g_{\alpha \beta ,22}g_{1\lambda }+%
\frac{1}{2}\left( g^{\lambda \beta }g_{\alpha \beta ,2}\right)
_{,\lambda }-3\left( g^{12}g_{12,2}\right) _{,\alpha
}-\frac{1}{2}\left( g^{12}\right)
^{2}g_{12,2}g_{12,\alpha }  \notag \\
& +\frac{1}{2}g^{12}\left( g_{22,1\alpha }+g_{12,2\alpha }\right) +\frac{1}{2%
}g_{,2}^{\beta \lambda }\left( g_{\lambda \beta ,\alpha
}+g_{\lambda \alpha
,\beta }\right)  \notag \\
& +\frac{1}{2}g_{\alpha \beta }\left[ -2g^{\beta \lambda
}g^{12}g_{12,\lambda }+g^{\beta \lambda }g^{\mu \theta }\left(
2g_{\lambda \mu ,\theta }-g_{\mu \theta ,\lambda }\right) \right]
_{,2}.  \tag{C.14}
\end{align}%
Now $\left( C.10\right) $\ implies that%
\begin{align}
g^{\beta \lambda }g^{12}g_{2\lambda ,1}& =\Gamma ^{\beta
}-g^{\beta \lambda
}g^{12}g_{1\lambda ,2}-g_{,2}^{\beta \lambda }g^{12}g_{1\lambda }  \notag \\
& -\frac{1}{2}\left[ -2g^{\beta \lambda }g^{12}g_{12,\lambda
}+g^{\beta \lambda }g^{\mu \theta }\left( 2g_{\lambda \mu ,\theta
}-g_{\mu \theta ,\lambda }\right) \right] ,  \tag{C.15}
\end{align}%
and%
\begin{align}
g^{12}g_{2\alpha ,1}& =g_{\alpha \beta }\Gamma ^{\beta
}-g^{12}g_{1\alpha
,2}-g_{\alpha \beta }g_{,2}^{\beta \lambda }g^{12}g_{1\lambda }  \notag \\
& -\frac{1}{2}\left[ -2g^{12}g_{12,\alpha }+g^{\mu \theta }\left(
2g_{\alpha \mu ,\theta }-g_{\mu \theta ,\alpha }\right) \right] .
\tag{C.16}
\end{align}%
The insertion of $\left( C.15\right) $ and $\left( C.16\right) $
in $\left( C.14\right) $ yields the expected relation $\left(
\ref{4.24}\right) $.

\textbf{Proof of\ }$\left( \ref{4.25}\right) $\textbf{. }From
$\left(
2.22\right) $ we have%
\begin{equation}
\begin{array}{l}
LA_{2}=g^{ik}A_{2,ik}+g_{,2}^{ki}A_{k,i}+g^{ik}\left[
A_{k},A_{2}\right]
_{,i} \\
+g_{j2}\left[ \left( g^{ik}g^{jl}\right) _{,i}\left[
A_{l,k}-A_{k,l}+\left[
A_{k},A_{l}\right] \right] +\Gamma _{im}^{i}F^{mj}+\Gamma _{im}^{j}F^{im}+%
\left[ A_{i},F^{ij}\right] \right] .%
\end{array}
\tag{C.17}
\end{equation}

\textbf{Calculation of }$g^{ik}A_{2,ik},$\textbf{\ }$g_{,2}^{ki}A_{k,i}$%
\textbf{\ and }$g^{ik}\left[ A_{k},A_{2}\right] _{,i}$. It\ holds\ that%
\begin{equation}
g^{ik}A_{2,ik}=2g^{12}A_{2,12}+g^{22}A_{2,22}+2g^{2\alpha
}A_{2,2\alpha }+g^{\alpha \beta }A_{2,\alpha \beta },\quad \alpha
,\beta \in \left\{ 3,4\right\} .  \tag{C.18}
\end{equation}%
The expression of $g^{22}$ in terms of $g_{1\lambda }$ and
$g_{11}$ via the equalities
\begin{equation*}
g^{2i}g_{1i}=0\text{ and }g^{2\alpha }=-g^{12}g^{\alpha \lambda
}g_{1\lambda }
\end{equation*}%
gives
\begin{equation}
g^{22}=-\left( g^{12}\right) ^{2}g_{11}+\left( g^{12}\right)
^{2}g^{\alpha \lambda }g_{1\alpha }g_{1\lambda }.  \tag{C.19}
\end{equation}%
Using the assumption $A_{2}=0$ on $G^{1}$,\ $\left( C.18\right) $\ yields%
\begin{equation}
g^{ik}A_{2,ik}=2g^{12}A_{2,12}.  \tag{C.20}
\end{equation}%
The calculation of $g_{,2}^{ki}A_{k,i}$ gives
\begin{equation}
g_{,2}^{ki}A_{k,i}=g_{,2}^{12}\left( A_{1,2}+A_{2,1}\right)
+g_{,2}^{2\alpha }\left( A_{\alpha ,2}+A_{2,\alpha }\right) .
\tag{C.21}
\end{equation}%
We calculate $g_{,2}^{2\alpha }$ by using the relation $g^{2\alpha
}=-g^{12}g^{\alpha \lambda }g_{1\lambda }$ satisfied on $G^{1}$ to
have
\begin{equation}
g_{,2}^{2\alpha }=\left( \left( g^{12}\right)
^{2}g_{12,2}g^{\alpha \lambda }-g^{12}g_{,2}^{\alpha \lambda
}\right) g_{1\lambda }-g^{12}g^{\alpha \lambda }g_{1\lambda ,2}.
\tag{C.22}
\end{equation}%
Since $g_{,2}^{12}=-\left( g^{12}\right) ^{2}g_{12,2}$ (see
$\left(
B.4\right) $), we deduce from $\left( C.21\right) $ and $\left( C.22\right) $%
\ that

\begin{align}
g_{,2}^{ki}A_{k,i}& =-\left( g^{12}\right) ^{2}g_{12,2}\left(
A_{1,2}+A_{2,1}\right)  \notag \\
& +A_{\alpha ,2}\left[ \left( \left( g^{12}\right)
^{2}g_{12,2}g^{\alpha \lambda }-g^{12}g_{,2}^{\alpha \lambda
}\right) g_{1\lambda }-g^{12}g^{\alpha \lambda }g_{1\lambda
,2}\right] .  \tag{C.23}
\end{align}%
As $g^{11}=g^{1\alpha }=0$ and $A_{2}=0$\ on $G^{1},$ the calculation of $%
g^{ik}\left[ A_{k},A_{2}\right] _{,i}$ on $G^{1}$\ easily gives
\begin{equation}
g^{ik}\left[ A_{k},A_{2}\right] _{,i}=0.  \tag{C.24}
\end{equation}%
$\left( C.20-C.24\right) $ imply%
\begin{align}
g^{ik}A_{2,ik}+g_{,2}^{ki}A_{k,i}+g^{ik}\left[ A_{k},A_{2}\right]
_{,i}& =2g^{12}A_{2,12}-\left( g^{12}\right) ^{2}g_{12,2}\left(
A_{1,2}+A_{2,1}\right)  \notag \\
& +A_{\alpha ,2}\left[ \left( \left( g^{12}\right)
^{2}g_{12,2}g^{\alpha \lambda }-g^{12}g_{,2}^{\alpha \lambda
}\right) g_{1\lambda }-g^{12}g^{\alpha \lambda }g_{1\lambda
,2}\right] .  \tag{C.25}
\end{align}

\textbf{Calculation of }$g_{j2}\left[ \left( g^{ik}g^{jl}\right)
_{,i}\left[ A_{l,k}-A_{k,l}+\left[ A_{k},A_{l}\right] \right]
+\Gamma _{im}^{i}F^{mj}+\Gamma _{im}^{j}F^{im}+\left[
A_{i},F^{ij}\right] \right] .$

Using the assumption $g_{22}=g_{23}=g_{24}=0$ on $G^{1}$, we get

\begin{align}
& g_{j2}\left[ \left( g^{ik}g^{jl}\right) _{,i}\left[
A_{l,k}-A_{k,l}+\left[
A_{k},A_{l}\right] \right] +\Gamma _{im}^{i}F^{mj}+\Gamma _{im}^{j}F^{im}+%
\left[ A_{i},F^{ij}\right] \right]  \notag \\
& =g_{12}\left[ \left( g^{ik}g^{1l}\right) _{,i}\left(
A_{l,k}-A_{k,l}+\left[
A_{k},A_{l}\right] \right) +\Gamma _{im}^{i}F^{m1}+\Gamma _{im}^{1}F^{im}+%
\left[ A_{i},F^{i1}\right] \right] .  \tag{C.26}
\end{align}

From the assumption $g^{11}=g^{13}=g^{14}=0,$ $A_{2}=0$ on
$G^{1}$, a simple
calculation shows that%
\begin{align}
g_{,i}^{ik}g^{1l}\left( A_{l,k}-A_{k,l}+\left[ A_{k},A_{l}\right]
\right) & =g^{12}\left( g_{,1}^{1k}+g_{,2}^{2k}+g_{,\beta }^{\beta
k}\right) \left(
A_{2,k}-A_{k,2}+\left[ A_{k},A_{2}\right] \right)  \notag \\
& =g^{12}\left\{ \left( g_{,1}^{11}+g_{,2}^{12}\right) \left(
A_{2,1}-A_{1,2}\right) -\left( g_{,1}^{1\alpha }+g_{,2}^{2\alpha
}+g_{,\beta }^{\alpha \beta }\right) A_{\alpha ,2}\right\} ,
\tag{C.27}
\end{align}%
with%
\begin{equation}
g^{1k}g_{,1}^{1l}\left( A_{l,k}-A_{k,l}+\left[ A_{k},A_{l}\right]
\right) =g^{12}\left\{ g_{,1}^{11}\left( A_{1,2}-A_{2,1}\right)
+g_{,1}^{1\alpha }A_{\alpha ,2}\right\} ,  \tag{C.28}
\end{equation}%
\begin{equation}
g^{2k}g_{,2}^{1l}\left( A_{l,k}-A_{k,l}+\left[ A_{k},A_{l}\right]
\right) =g_{,2}^{12}\left\{ g^{12}\left( A_{2,1}-A_{1,2}\right)
-g^{2\alpha }A_{\alpha ,2}\right\} ,  \tag{C.29}
\end{equation}%
\begin{equation}
g^{\beta k}g_{,\beta }^{1l}\left( A_{l,k}-A_{k,l}+\left[
A_{k},A_{l}\right] \right) =-g^{\alpha \beta }g_{,\beta
}^{12}A_{\alpha ,2}.  \tag{C.30}
\end{equation}%
$\left( C.27-C.30\right) $ yield%
\begin{align}
& \left( g^{ik}g^{1l}\right) _{,i}\left( A_{l,k}-A_{k,l}+\left[ A_{k},A_{l}%
\right] \right)  \notag \\
& =g^{12}\left\{ 2g_{,2}^{12}\left( A_{2,1}-A_{1,2}\right) -\left(
g_{,2}^{2\alpha }+g_{,\beta }^{\alpha \beta }\right) A_{\alpha
,2}\right\} -g_{,2}^{12}g^{2\alpha }A_{\alpha ,2}-g^{\alpha \beta
}g_{,\beta }^{12}A_{\alpha ,2}.  \tag{C.31}
\end{align}%
Inserting the relations%
\begin{align}
g^{2\alpha }& =-g^{12}g^{\alpha \lambda }g_{1\lambda },\quad
g_{,2}^{12}=-\left( g^{12}\right) ^{2}g_{12,2},\quad g_{,\beta
}^{12}=-\left( g^{12}\right) ^{2}g_{12,\beta },  \notag \\
g_{,2}^{2\alpha }& =\left( \left( g^{12}\right)
^{2}g_{12,2}g^{\alpha \lambda }-g^{12}g_{,2}^{\alpha \lambda
}\right) g_{1\lambda }-g^{12}g^{\alpha \lambda }g_{1\lambda ,2},
\tag{C.32}
\end{align}%
in\ $\left( C.31\right) $,\ we get%
\begin{align}
g_{12}\left( g^{ik}g^{1l}\right) _{,i}\left(
A_{l,k}-A_{k,l}+\left[ A_{k},A_{l}\right] \right) & =-2\left(
g^{12}\right) ^{2}g_{12,2}\left(
A_{2,1}-A_{1,2}\right)  \notag \\
& -\left( 2\left( g^{12}\right) ^{2}g_{12,2}g^{\alpha \lambda
}-g^{12}g_{,2}^{\alpha \lambda }\right) A_{\alpha ,2}g_{1\lambda }  \notag \\
& +g^{12}g^{\alpha \lambda }A_{\alpha ,2}g_{1\lambda ,2}-g_{,\beta
}^{\alpha \beta }A_{\alpha ,2}+g^{\alpha \beta }g^{12}g_{12,\beta
}A_{\alpha ,2}. \tag{C.33}
\end{align}

We now handle\textbf{\ }$g_{j2}\left[ \Gamma
_{im}^{i}F^{mj}+\Gamma
_{im}^{j}F^{im}+\left[ A_{i},F^{ij}\right] \right] $. Using the assumption $%
g_{22}=g_{23}=g_{24}=0$ on $G^{1}$, we get%
\begin{equation}
g_{j2}\Gamma _{im}^{i}F^{mj}=g_{12}\left( \Gamma
_{i2}^{i}F^{21}+\Gamma _{i\beta }^{i}F^{\beta 1}\right) .
\tag{C.34}
\end{equation}%
Simple calculations on $G^{1}$\ give
\begin{equation}
\Gamma _{i2}^{i}=3g^{12}g_{12,2},\quad 2\Gamma _{i\beta
}^{i}=2g^{12}g_{12,\beta }+g^{\lambda \mu }\left( g_{\mu \beta
,\lambda }+g_{\lambda \mu ,\beta }-g_{\beta \lambda ,\mu }\right)
,  \tag{C.35}
\end{equation}%
and%
\begin{equation}
F^{21}=\left( g^{12}\right) ^{2}\left( A_{2,1}-A_{1,2}\right)
+\left( g^{12}\right) ^{2}g^{\alpha \lambda }A_{\alpha
,2}g_{1\lambda },,\quad F^{\beta 1}=-g^{12}g^{\beta \alpha
}A_{\alpha ,2}.  \tag{C.36}
\end{equation}%
$\left( C.35\right) $ and $\left( C.36\right) $\ give%
\begin{align}
\Gamma _{i2}^{i}F^{21}& =3\left( g^{12}\right) ^{3}g_{12,2}\left[
\left( A_{2,1}-A_{1,2}\right) +g^{\alpha \lambda }A_{\alpha
,2}g_{1\lambda }\right]
,  \notag \\
\Gamma _{i\beta }^{i}F^{\beta 1}& =-g^{12}g^{\beta \alpha
}A_{\alpha ,2} \left[ g^{12}g_{12,\beta }+\frac{1}{2}g^{\lambda
\mu }\left( g_{\mu \beta ,\lambda }+g_{\lambda \mu ,\beta
}-g_{\beta \lambda ,\mu }\right) \right] , \tag{C.37}
\end{align}%
and then
\begin{align}
g_{12}\Gamma _{i2}^{i}F^{21}& =3\left( g^{12}\right)
^{2}g_{12,2}\left[
\left( A_{2,1}-A_{1,2}\right) +g^{\alpha \lambda }A_{\alpha ,2}g_{1\lambda }%
\right] ,  \notag \\
g_{12}\Gamma _{i\beta }^{i}F^{\beta 1}& =-g^{\beta \alpha
}A_{\alpha ,2} \left[ g^{12}g_{12,\beta }+\frac{1}{2}g^{\lambda
\mu }\left( g_{\mu \beta ,\lambda }+g_{\lambda \mu ,\beta
}-g_{\beta \lambda ,\mu }\right) \right] . \tag{C.38a}
\end{align}%
For the term $g_{j2}\Gamma _{im}^{j}F^{im}$, since $\Gamma
_{im}^{j}=\Gamma _{mi}^{j}$ and $F^{im}=-F^{mi}$, a simple
computation gives\
\begin{equation}
g_{j2}\Gamma _{im}^{j}F^{im}=0.  \tag{C.38b}
\end{equation}%
The calculation of $g_{j2}\left[ A_{i},F^{ij}\right] $ gives%
\begin{equation}
g_{j2}\left[ A_{i},F^{ij}\right] =-g^{\beta \alpha }\left[
A_{\beta },A_{\alpha ,2}\right] .  \tag{C.39}
\end{equation}%
Finally, from $\left( C.17\right) $, $\left( C.25\right) $,\
$\left(
C.33\right) $, $\left( C.38a-b\right) $\ and $\left( C.39\right) $,\ we gain%
\begin{align}
LA_{2}& =2g^{12}A_{2,12}-g^{\beta \alpha }\left[ A_{\beta },A_{\alpha ,2}%
\right] +g^{12}g_{12,\beta }g^{\alpha \beta }  \notag \\
& -2\left( g^{12}\right) ^{2}g_{12,2}A_{1,2}+2\left( g^{12}\right)
^{2}g_{12,2}g^{\alpha \lambda }A_{\alpha ,2}g_{1\lambda }  \notag \\
& +\left[ g^{\alpha \beta }\left( -\left[ g^{12}g_{12,\beta }+\frac{1}{2}%
g^{\lambda \mu }\left( g_{\mu \beta ,\lambda }+g_{\lambda \mu
,\beta
}-g_{\beta \lambda ,\mu }\right) \right] \right) -g_{,\beta }^{\alpha \beta }%
\right] A_{\alpha ,2}.  \tag{C.40}
\end{align}

\textbf{Calculation of\ }$\nabla _{2}\nabla ^{k}A_{k}$

By definition it holds that
\begin{equation}
\nabla ^{k}A_{k}=g^{ik}\nabla _{i}A_{k}=g^{ik}\left(
A_{k,i}-\Gamma _{ik}^{l}A_{l}\right) .  \tag{C.41}
\end{equation}%
Using the assumption $g^{11}=g^{1\alpha }=0$ and $A_{2}=0$\ on
$G^{1}$\ and the equality $g^{2\alpha }=-g^{12}g^{\alpha \lambda
}g_{1\lambda }\ $on$\
G^{1}$, we have%
\begin{equation}
g^{ki}A_{k,i}=g^{12}\left( A_{1,2}+A_{2,1}\right) -g^{12}g^{\alpha
\lambda }\left( A_{\alpha ,2}+A_{2,\alpha }\right) g_{1\lambda }.
\tag{C.42}
\end{equation}%
Using the notation $\Gamma ^{l}=g^{ik}\Gamma _{ik}^{l}$, the calculation of $%
g^{ik}\Gamma _{ik}^{l}A_{l}$ gives $g^{ik}\Gamma
_{ik}^{l}A_{l}=\Gamma ^{l}A_{l}$. Since $\Gamma ^{1}=0\ $on
$G^{1}$ at this step of the construction process, we deduce that
\begin{equation}
\nabla ^{k}A_{k}=g^{12}\left( A_{1,2}+A_{2,1}\right)
-g^{12}g^{\alpha
\lambda }A_{\alpha ,2}g_{1\lambda }+\Gamma ^{\alpha }A_{\alpha }\text{ on }%
G^{1}.  \tag{C.43}
\end{equation}%
From $\left( C.10\right) $ we\ deduce\ that on\ $G^{1}$ the\
following\
equality\ holds:%
\begin{align}
\Gamma _{,2}^{\alpha }& =\left[ g^{12}g_{,2}^{\beta \lambda
}\right] _{,2}g_{1\lambda }+g^{12}g_{,2}^{\beta \lambda
}g_{1\lambda ,2}+g^{12}g^{\alpha \lambda }\left( g_{2\lambda
,12}+g_{1\lambda ,22}\right) +\left( g^{12}g^{\alpha \lambda
}\right) _{,2}\left( g_{2\lambda
,1}+g_{1\lambda ,2}\right)  \notag \\
& +\left( -g^{12}g^{\alpha \lambda }g_{12,\lambda
}+\frac{1}{2}\left[ g^{\alpha \delta }g^{\beta \mu }\left( g_{\mu
\delta ,\beta }+g_{\delta \beta ,\mu }-g_{\mu \beta ,\delta
}\right) \right] \right) _{,2}.  \tag{C.44}
\end{align}%
$\left( C.43\right) $ and $\left( C.10\right) $ yield
\begin{align}
\nabla ^{k}A_{k}& =g^{12}\left( A_{1,2}+A_{2,1}\right) +\left[
g^{12}g_{,2}^{\beta \lambda }A_{\alpha }-g^{12}g^{\alpha \lambda
}A_{\alpha ,2}\right] g_{1\lambda }+g^{12}g^{\alpha \lambda
}A_{\alpha }\left(
g_{2\lambda ,1}+g_{1\lambda ,2}\right)  \notag \\
& -g^{12}g^{\alpha \lambda }g_{12,\lambda }A_{\alpha
}+\frac{1}{2}\left[ g^{\alpha \delta }g^{\beta \mu }\left( g_{\mu
\delta ,\beta }+g_{\delta \beta ,\mu }-g_{\mu \beta ,\delta
}\right) \right] A_{\alpha }.  \tag{C.45}
\end{align}%
Differentiation of $\left( C.45\right) $ w.r.t. $x^{2}$ gives%
\begin{align}
\nabla _{2}\left( \nabla ^{k}A_{k}\right) & =g^{12}\left(
A_{1,22}+A_{2,12}\right) +g_{,2}^{12}\left( A_{1,2}+A_{2,1}\right)  \notag \\
& +\left[ g^{12}g_{,2}^{\beta \lambda }A_{\alpha }-g^{12}g^{\alpha
\lambda }A_{\alpha ,2}\right] _{,2}g_{1\lambda }+\left[
g^{12}g_{,2}^{\beta \lambda }A_{\alpha }-g^{12}g^{\alpha \lambda
}A_{\alpha ,2}\right] g_{1\lambda ,2}
\notag \\
& +\left( g^{12}g^{\alpha \lambda }A_{\alpha }\right) _{,2}\left(
g_{2\lambda ,1}+g_{1\lambda ,2}\right) +g^{12}g^{\alpha \lambda
}A_{\alpha
}\left( g_{2\lambda ,12}+g_{1\lambda ,22}\right)  \notag \\
& +\left( -g^{12}g^{\alpha \lambda }g_{12,\lambda }A_{\alpha }+\frac{1}{2}%
\left[ g^{\alpha \delta }g^{\beta \mu }\left( g_{\mu \delta ,\beta
}+g_{\delta \beta ,\mu }-g_{\mu \beta ,\delta }\right) \right]
A_{\alpha }\right) _{,2}.  \tag{C.46}
\end{align}%
As $g_{,2}^{12}=-\left( g^{12}\right) ^{2}g_{12,2}$ (see\ $\left(
B.4\right) $)\ and $\nabla ^{k}A_{k}=\Delta ,$ we gain
\begin{align}
\Delta _{,2}& =g^{12}\left( A_{1,22}+A_{2,12}\right) -\left(
g^{12}\right)
^{2}g_{12,2}\left( A_{1,2}+A_{2,1}\right)  \notag \\
& +\left[ g^{12}g_{,2}^{\beta \lambda }A_{\alpha }-g^{12}g^{\alpha
\lambda }A_{\alpha ,2}\right] _{,2}g_{1\lambda }+\left[
g^{12}g_{,2}^{\beta \lambda }A_{\alpha }-g^{12}g^{\alpha \lambda
}A_{\alpha ,2}\right] g_{1\lambda ,2}
\notag \\
& +\left( g^{12}g^{\alpha \lambda }A_{\alpha }\right) _{,2}\left(
g_{2\lambda ,1}+g_{1\lambda ,2}\right) +g^{12}g^{\alpha \lambda
}A_{\alpha
}\left( g_{2\lambda ,12}+g_{1\lambda ,22}\right)  \notag \\
& +\left( -g^{12}g^{\alpha \lambda }g_{12,\lambda }A_{\alpha }+\frac{1}{2}%
\left[ g^{\alpha \delta }g^{\beta \mu }\left( g_{\mu \delta ,\beta
}+g_{\delta \beta ,\mu }-g_{\mu \beta ,\delta }\right) \right]
A_{\alpha }\right) _{,2}.  \tag{C.47}
\end{align}%
$\left( C.40\right) $ and $\left( C.47\right) $ yield%
\begin{align}
LA_{2}-2\Delta _{,2}& =-2g^{12}A_{1,22}+2\left( g^{12}\right)
^{2}g_{12,2}A_{2,1}+2\left( g^{12}\right) ^{2}g_{12,2}g^{\alpha
\lambda
}A_{\alpha ,2}g_{1\lambda }  \notag \\
& +\left[ 4\left( g^{12}\right) ^{2}g_{12,2}g^{\alpha \lambda
}A_{\alpha }+g^{12}g^{\alpha \lambda }g^{\beta \mu }g_{\mu \beta
,2}A_{\alpha }+2g^{12}g^{\alpha \lambda }A_{\alpha ,2}\right]
_{,2}g_{1\lambda }  \notag
\\
& +\left[ 4\left( g^{12}\right) ^{2}g_{12,2}g^{\alpha \lambda
}A_{\alpha }+g^{12}g^{\alpha \lambda }g^{\beta \mu }g_{\mu \beta
,2}A_{\alpha
}+2g^{12}g^{\alpha \lambda }A_{\alpha ,2}\right] g_{1\lambda ,2}  \notag \\
& -2\left( g^{12}g^{\alpha \lambda }A_{\alpha }\right) _{,2}\left(
g_{2\lambda ,1}+g_{1\lambda ,2}\right) -2g^{12}g^{\alpha \lambda
}A_{\alpha
}\left( g_{2\lambda ,12}+g_{1\lambda ,22}\right)  \notag \\
& +\left[ g^{\alpha \beta }\left( g^{12}g_{12,\beta }-\left[
g^{12}g_{12,\beta }+\frac{1}{2}g^{\lambda \mu }\left( g_{\mu \beta
,\lambda }+g_{\lambda \mu ,\beta }-g_{\beta \lambda ,\mu }\right)
\right] \right)
-g_{,\beta }^{\alpha \beta }\right] A_{\alpha ,2}  \notag \\
& -g^{\beta \alpha }\left[ A_{\beta },A_{\alpha ,2}\right] +\left(
2g^{12}g^{\alpha \lambda }g_{12,\lambda }A_{\alpha }-\left[
g^{\alpha \delta }g^{\beta \mu }\left( g_{\mu \delta ,\beta
}+g_{\delta \beta ,\mu }-g_{\mu \beta ,\delta }\right) \right]
A_{\alpha }\right) _{,2}.  \tag{C.48}
\end{align}%
From $\left( C.16\right) $\ we have%
\begin{align}
\left( g_{2\lambda ,1}+g_{1\lambda ,2}\right) & =g_{12}g_{\nu
\lambda }\Gamma ^{\nu }+g^{\beta \mu }g_{\mu \beta ,2}g_{1\lambda
}+g_{12,\lambda }
\notag \\
& -\frac{1}{2}g_{12}\left[ g^{\beta \mu }\left( g_{\mu \lambda
,\beta }+g_{\lambda \beta ,\mu }-g_{\mu \beta ,\lambda }\right)
\right] . \tag{C.49}
\end{align}%
Thus%
\begin{align}
\left( g_{2\lambda ,12}+g_{1\lambda ,22}\right) & =\left(
g_{12}g_{\nu \lambda }\right) _{,2}\Gamma ^{\nu }+g_{12}g_{\nu
\lambda }\Gamma _{,2}^{\nu }+\left[ \frac{1}{2}g^{\beta \mu
}g_{\mu \beta ,2}\right] _{,2}g_{1\lambda }
\notag \\
& +g^{\beta \mu }g_{\mu \beta ,2}g_{1\lambda ,2}+g_{12,2\lambda }-\frac{1}{2}%
\left( g_{12}\left[ g^{\beta \mu }\left( g_{\mu \lambda ,\beta
}+g_{\lambda \beta ,\mu }-g_{\mu \beta ,\lambda }\right) \right]
\right) _{,2}. \tag{C.50}
\end{align}%
One proceeds in the same way to calculate $A_{2,1}$ from$\ \left(
C.45\right) \ $and $\left( C.49\right) $ to have%
\begin{align}
A_{2,1}& =g_{12}\Delta -A_{\alpha }g_{12}\Gamma ^{\alpha
}-A_{1,2}+g^{\alpha \lambda }A_{\alpha ,2}g_{1\lambda }-A_{\alpha
}g^{\alpha \lambda
}g_{12,\lambda }  \notag \\
& +\frac{1}{2}g_{12}g^{\alpha \lambda }A_{\alpha }\left[ g^{\beta
\mu }\left( g_{\mu \lambda ,\beta }+g_{\lambda \beta ,\mu }-g_{\mu
\beta ,\lambda }\right) \right] +g^{\alpha \lambda }g_{12,\lambda
}A_{\alpha }
\notag \\
& -\frac{1}{2}g_{12}\left[ g^{\alpha \delta }g^{\beta \mu }\left(
g_{\mu \delta ,\beta }+g_{\delta \beta ,\mu }-g_{\mu \beta ,\delta
}\right) \right] A_{\alpha }.  \tag{C.51}
\end{align}%
Inserting $\left( C.49\right) $, $\left( C.50\right) $\ and
$\left( C.51\right) $ in $\left( C.48\right) $ and using the
equality $g^{\beta \mu }g_{\mu \beta ,2}=4g^{12}g_{12,2}$\ on
$G^{1}$,\ we gain the desired relation $\left( \ref{4.25}\right)
$.

\textbf{Proof of item }$\left( ii\right) $. One uses the expression of $%
A_{2,1}$\ given above in $\left( C.51\right) $\ to obtain, from
$\left(
C.9\right) $, the following expression of $\tau _{2\alpha }$%
\begin{align}
\tau _{2\alpha }& =-2g^{12}A_{\alpha ,2}.A_{1,2}+A_{\alpha
,2}.\Delta -A_{\nu }.A_{\alpha ,2}\Gamma ^{\nu }+2g^{\nu \lambda
}g^{12}A_{\alpha
,2}.A_{\nu ,2}g_{1\lambda }  \notag \\
& -g^{\beta \lambda }\left( A_{\alpha ,\lambda }-A_{\lambda
,\alpha }+\left[ A_{\lambda },A_{\alpha }\right] \right) .A_{\beta
,2}+\left( \Phi _{,2}\right) .\left( \Phi _{,\alpha }+\left[
A_{\alpha },\Phi \right] \right) .  \tag{C.52}
\end{align}%
The insertion of $\left( C.52\right) $\ in $\left(
\ref{4.26}\right) $ gives the desired system $\left(
\ref{4.27}\right) $ with the appropriate known coefficients given
by $\left( \ref{4.28}\right) $.

\section*{Appendix D: Proof of Proposition 3.11}

\textbf{Proof of item }$\left( i\right) $. By definition, we have%
\begin{equation}
R_{\alpha \beta }=\Gamma _{\alpha \beta ,k}^{k}-\Gamma _{\alpha
k,\beta }^{k}+\Gamma _{lk}^{k}\Gamma _{\alpha \beta }^{l}-\Gamma
_{l\beta }^{k}\Gamma _{\alpha k}^{l}.  \tag{D.1}
\end{equation}%
As in the previous steps, each term in the r.h.s of $\left(
D.1\right) $ is
calculated meticulously on $G^{1}$. Here one uses the equality%
\begin{equation*}
g_{,1}^{12}=\left( g^{12}\right) ^{3}g_{22,1}g_{11}+\left(
g^{12}\right) ^{2}g_{22,1}g^{2\mu }g_{1\mu }-g^{12}\left(
g^{12}g_{12,1}+g^{2\mu }g_{2\mu ,1}\right) ,
\end{equation*}%
which follows from the equalities $\left( g^{2i}g_{2i}\right)
_{,1}=0,\quad
g^{2i}g_{1i}=0$ on\ $G^{1}$, to gain%
\begin{align}
2\Gamma _{\alpha \beta ,k}^{k}& =-\left(
g_{,1}^{11}+g_{,2}^{12}\right) g_{\alpha \beta
,1}-2g^{12}g_{\alpha \beta ,12}+\left( g^{12}\right)
^{2}g_{\alpha \beta ,2}g_{12,1}  \notag \\
& +\left( g^{12}\right) ^{2}g_{\alpha \beta ,2}g_{11,2}+\left\{
\left[ \left( g^{12}\right) ^{2}g_{\alpha \beta ,2}\right]
_{,2}-\left(
g^{12}\right) ^{3}g_{22,1}g_{\alpha \beta ,2}\right\} g_{11}  \notag \\
& -g_{\alpha \beta ,2}\left[ \left( g^{12}\right)
^{2}g_{22,1}g^{2\mu
}g_{1\mu }-g^{12}g^{2\mu }g_{2\mu ,1}\right]   \notag \\
& +g_{,1}^{11}\left( g_{1\beta ,\alpha }+g_{1\alpha ,\beta
}\right) +g_{,1}^{1\mu }\left( g_{\beta \mu ,\alpha }+g_{\mu
\alpha ,\beta
}-g_{\alpha \beta ,\mu }\right)   \notag \\
& +g^{12}\left( g_{2\beta ,1\alpha }+g_{2\alpha ,1\beta }\right)
+g_{\alpha
\beta ,2}\left( g^{12}g^{2\mu }g_{1\mu }\right) _{,2}  \notag \\
& +g^{12}g^{2\mu }g_{1\mu }g_{\alpha \beta ,22}+g_{,2}^{12}\left(
g_{1\beta ,\alpha }+g_{1\alpha ,\beta }\right) +g_{,2}^{2\mu
}\left( g_{\mu \beta
,\alpha }+g_{\mu \alpha ,\beta }-g_{\alpha \beta ,\mu }\right)   \notag \\
& +g^{12}\left( g_{1\beta ,2\alpha }+g_{1\alpha ,2\beta }\right)
+g^{2\mu }\left( g_{\beta \mu ,2\alpha }+g_{\mu \alpha ,2\beta
}-g_{\alpha \beta
,2\mu }\right)   \notag \\
& -g_{,\lambda }^{2\lambda }g_{\alpha \beta ,2}+g_{,\lambda
}^{\lambda \mu }\left( g_{\mu \beta ,\alpha }+g_{\mu \alpha ,\beta
}-g_{\alpha \beta ,\mu
}\right)   \notag \\
& -g^{2\lambda }g_{\alpha \beta ,2\lambda }+g^{\lambda \mu }\left(
g_{\mu \beta ,\lambda \alpha }+g_{\mu \alpha ,\lambda \beta
}-g_{\alpha \beta ,\lambda \mu }\right) .  \tag{D.2a}
\end{align}%
Similarly, we have%
\begin{equation}
2\Gamma _{\alpha k,\beta }^{k}=\left[ 2g^{12}g_{12,\alpha
}+g^{\lambda \mu }\left( g_{\mu \lambda ,\alpha }+g_{\mu \alpha
,\lambda }-g_{\alpha \lambda ,\mu }\right) \right] _{,\beta }.
\tag{D.2b}
\end{equation}%
$\left( D.2a\right) $ and $\left( D.2b\right) $\ yield%
\begin{align}
2\left( \Gamma _{\alpha \beta ,k}^{k}-\Gamma _{\alpha k,\beta
}^{k}\right) & =-\left( g_{,1}^{11}+g_{,2}^{12}\right) g_{\alpha
\beta ,1}-2g^{12}g_{\alpha \beta ,12}+\left( g^{12}\right)
^{2}g_{\alpha \beta ,2}g_{12,1}+\left(
g^{12}\right) ^{2}g_{\alpha \beta ,2}g_{11,2}  \notag \\
& +\left\{ \left[ \left( g^{12}\right) ^{2}g_{\alpha \beta
,2}\right] _{,2}-\left( g^{12}\right) ^{3}g_{22,1}g_{\alpha \beta
,2}\right\} g_{11}+N_{\alpha \beta },  \tag{D.2c}
\end{align}%
where $N_{\alpha \beta }$ is known and given on $G^{1}$ by $\left(
3.25\right) $.

The calculation of\textbf{\ }$\Gamma _{lk}^{k}\Gamma _{\alpha
\beta }^{l}$
and $\Gamma _{l\beta }^{k}\Gamma _{\alpha k}^{l}\ $gives%
\begin{align}
4\Gamma _{lk}^{k}\Gamma _{\alpha \beta }^{l}& =-2\left(
g^{12}\right) ^{2}g_{\alpha \beta ,2}g_{12,1}-g^{12}\left(
3g^{12}g_{22,1}+g^{\lambda \mu
}g_{\lambda \mu ,2}\right) g_{\alpha \beta ,1}  \notag \\
& +\left( g^{12}\right) ^{2}g_{\alpha \beta ,2}\left(
4g^{12}g_{22,1}+g^{\lambda \mu }g_{\lambda \mu ,2}\right) g_{11}  \notag \\
& -g^{12}g_{\alpha \beta ,2}\left[ -g^{12}g_{22,1}g^{2\mu }g_{1\mu
}+2g^{2\lambda }g_{2\lambda ,1}+g^{\lambda \mu }\left( g_{1\mu
,\lambda
}+g_{\lambda \mu ,1}-g_{1\lambda ,\mu }\right) \right]   \notag \\
& +\left[ g^{12}g^{2\mu }g_{1\mu }g_{\alpha \beta ,2}+g^{12}\left(
g_{1\beta ,\alpha }+g_{1\alpha ,\beta }\right) +g^{2\mu }\left(
g_{\mu \beta ,\alpha }+g_{\mu \alpha ,\beta }-g_{\alpha \beta ,\mu
}\right) \right] \left(
3g^{12}g_{22,1}+g^{\lambda \mu }g_{\lambda \mu ,2}\right)   \notag \\
& +\left[ 2g^{12}g_{12,\lambda }+g^{\mu \theta }\left( g_{\mu
\theta
,\lambda }+g_{\theta \lambda ,\mu }-g_{\mu \lambda ,\theta }\right) \right] %
\left[ -g^{2\lambda }g_{\alpha \beta ,2}+g^{\mu \lambda }\left(
g_{\mu \beta ,\alpha }+g_{\mu \alpha ,\beta }-g_{\alpha \beta ,\mu
}\right) \right] , \tag{D.3a}
\end{align}%
and%
\begin{align}
4\Gamma _{l\beta }^{k}\Gamma _{\alpha k}^{l}& =2\left(
g^{12}\right) ^{2}g^{\lambda \mu }g_{\lambda \beta ,2}g_{\alpha
\mu ,2}g_{11}-2g^{12}g^{\lambda \mu }g_{\lambda \beta ,2}g_{\alpha
\mu ,1}
\notag \\
& +\left( g^{12}\right) ^{2}\left( g_{12,\beta }+g_{2\beta
,1}-g_{1\beta ,2}\right) \left( g_{12,\alpha }+g_{2\alpha
,1}-g_{1\alpha ,2}\right)
\notag \\
& -g^{12}g_{\mu \alpha ,2}\left[ g^{2\mu }\left( g_{2\beta
,1}+g_{12,\beta }-g_{1\beta ,2}\right) +g^{\lambda \mu }\left(
g_{\lambda \beta
,1}+g_{1\lambda ,\beta }-g_{1\beta ,\lambda }\right) \right]   \notag \\
& +2g^{12}g^{2\lambda }g_{1\lambda }g^{\theta \mu }g_{\theta \beta
,2}g_{\alpha \mu ,2}  \notag \\
& +g^{12}g^{\lambda \mu }g_{\lambda \beta ,2}\left( g_{1\mu
,\alpha }+g_{1\alpha ,\mu }\right) +g^{\theta \mu }g_{\theta \beta
,2}g^{2\lambda }\left( g_{\lambda \mu ,\alpha }+g_{\lambda \alpha
,\mu }-g_{\alpha \mu
,\lambda }\right)   \notag \\
& +\left[ g^{12}\left( g_{12,\beta }+g_{1\beta ,2}-g_{2\beta
,1}\right) +g^{2\mu }g_{\mu \beta ,2}\right] \left[ g^{12}\left(
g_{12,\alpha
}+g_{1\alpha ,2}-g_{2\alpha ,1}\right) +g^{2\lambda }g_{\lambda \alpha ,2}%
\right]   \notag \\
& +g^{12}g^{\lambda \theta }g_{\theta \alpha ,2}\left( g_{1\beta
,\lambda }+g_{1\lambda ,\beta }\right) +g^{2\mu }g^{\lambda \theta
}g_{\theta \alpha ,2}\left( g_{\mu \beta ,\lambda }+g_{\mu \lambda
,\beta }-g_{\lambda \beta
,\mu }\right)   \notag \\
& -g^{12}g_{\lambda \beta ,2}\left[ g^{2\lambda }\left( g_{2\alpha
,1}+g_{12,\alpha }-g_{1\alpha ,2}\right) +g^{\lambda \mu }\left(
g_{\mu
\alpha ,1}+g_{1\mu ,\alpha }-g_{1\alpha ,\mu }\right) \right]   \notag \\
& +\left[ -g^{2\mu }g_{\lambda \beta ,2}+g^{\theta \mu }\left(
g_{\theta \lambda ,\beta }+g_{\theta \beta ,\lambda }-g_{\lambda
\beta ,\theta }\right) \right] \left[ -g^{2\lambda }g_{\alpha \mu
,2}+g^{\delta \lambda }\left( g_{\delta \mu ,\alpha }+g_{\delta
\alpha ,\mu }-g_{\alpha \mu ,\delta }\right) \right] .  \tag{D.3b}
\end{align}%
From\ $\left( D.3a\right) $ and $\left( D.3b\right) $ it follows
that
\begin{align}
4\left( \Gamma _{lk}^{k}\Gamma _{\alpha \beta }^{l}-\Gamma
_{l\beta }^{k}\Gamma _{\alpha k}^{l}\right) & =-2\left(
g^{12}\right) ^{2}g_{\alpha \beta ,2}g_{12,1}-g^{12}\left(
3g^{12}g_{22,1}+g^{\lambda \mu }g_{\lambda
\mu ,2}\right) g_{\alpha \beta ,1}  \notag \\
& +\left( g^{12}\right) ^{2}g_{\alpha \beta ,2}\left(
4g^{12}g_{22,1}+g^{\lambda \mu }g_{\lambda \mu ,2}\right)
g_{11}-2\left( g^{12}\right) ^{2}g^{\lambda \mu }g_{\lambda \beta
,2}g_{\alpha \mu ,2}g_{11}
\notag \\
& +2g^{12}g^{\lambda \mu }g_{\lambda \beta ,2}g_{\alpha \mu
,1}+2g^{12}g_{\mu \alpha ,2}g^{\lambda \mu }g_{\lambda \beta
,1}+M_{\alpha \beta },  \tag{D.3c}
\end{align}%
where $M_{\alpha \beta }$ is known and given on $G^{1}$ by $\left( \ref{4.34}%
\right) $ and $\left( \ref{4.34a}\right) $. $\left( D.2c\right) $ and $%
\left( D.3c\right) $ yield%
\begin{align}
R_{\alpha \beta }& =-\frac{1}{2}\left(
g_{,1}^{11}+g_{,2}^{12}\right) g_{\alpha \beta ,1}-g^{12}g_{\alpha
\beta ,12}+\frac{1}{2}\left(
g^{12}\right) ^{2}g_{\alpha \beta ,2}g_{11,2}  \notag \\
& +\frac{1}{2}\left\{ \left[ \left( g^{12}\right) ^{2}g_{\alpha \beta ,2}%
\right] _{,2}-\left( g^{12}\right) ^{3}g_{22,1}g_{\alpha \beta
,2}\right\}
g_{11}  \notag \\
& -\frac{1}{4}g^{12}\left( 3g^{12}g_{22,1}+g^{\lambda \mu
}g_{\lambda \mu
,2}\right) g_{\alpha \beta ,1}  \notag \\
& +\frac{1}{4}\left( g^{12}\right) ^{2}g_{\alpha \beta ,2}\left(
4g^{12}g_{22,1}+g^{\lambda \mu }g_{\lambda \mu ,2}\right) g_{11}  \notag \\
& -\frac{1}{2}\left( g^{12}\right) ^{2}g^{\lambda \mu }g_{\lambda
\beta ,2}g_{\alpha \mu ,2}g_{11}+\frac{1}{2}g^{12}g^{\lambda \mu
}g_{\lambda \beta ,2}g_{\alpha \mu ,1}+\frac{1}{4}\left( N_{\alpha
\beta }+M_{\alpha \beta }\right) .  \tag{D.4}
\end{align}%
Now using the following relations (see $\left( C.3-C.4\right) $,
and proof of Proposition 3.5)
\begin{equation*}
g_{,1}^{11}=-\left( g^{12}\right) ^{2}g_{22,1}=-2\left(
g^{12}\right)
^{2}g_{12,2},\text{\quad }g_{,2}^{12}=-\left( g^{12}\right) ^{2}g_{12,2},%
\text{\quad }g^{\lambda \mu }g_{\lambda \mu ,2}=4g^{12}g_{12,2},
\end{equation*}%
together with $\left( D.4\right) $, we gain%
\begin{align}
R_{\alpha \beta }& =\frac{1}{2}g^{12}g^{\lambda \mu }g_{\lambda
\beta ,2}g_{\alpha \mu ,1}+\frac{1}{4}g^{12}g_{\alpha \beta
,2}g^{\lambda \mu
}g_{\lambda \mu ,1}  \notag \\
& +\frac{1}{2}g^{12}g_{\mu \alpha ,2}g^{\lambda \mu }g_{\lambda
\beta ,1}-\left( g^{12}\right) ^{2}g_{12,2}g_{\alpha \beta
,1}-g^{12}g_{\alpha \beta ,12}+\frac{1}{2}\left( g^{12}\right)
^{2}g_{\alpha \beta ,2}g_{11,2}
\notag \\
& +\frac{1}{2}\left\{ \left[ \left( g^{12}\right) ^{2}g_{\alpha \beta ,2}%
\right] _{,2}+4\left( g^{12}\right) ^{3}g_{12,2}g_{\alpha \beta
,2}-\left( g^{12}\right) ^{2}g^{\lambda \mu }g_{\lambda \beta
,2}g_{\alpha \mu
,2}\right\} g_{11}  \notag \\
& +\frac{1}{4}\left( N_{\alpha \beta }+M_{\alpha \beta }\right) .
\tag{D.5}
\end{align}%
From $\left( D.5\right) $, after some supplementary calculation, we obtain%
\begin{align}
g^{\alpha \beta }R_{\alpha \beta }& =-g^{12}g_{,2}^{\alpha \beta
}g_{\alpha \beta ,1}-g^{12}g^{\alpha \beta }g_{\alpha \beta
,12}+2\left( g^{12}\right)
^{3}g_{12,2}g_{11,2}  \notag \\
& +\left\{ 4\left( g^{12}\right) ^{4}\left( g_{12,2}\right) ^{2}+\frac{1}{2}%
\left( g^{12}\right) ^{2}\left( g^{\alpha \beta }g_{\alpha \beta
,2}\right) _{,2}\right\} g_{11}+\frac{1}{4}g^{\alpha \beta }\left(
N_{\alpha \beta }+M_{\alpha \beta }\right) .  \tag{D.6}
\end{align}%
We now expand the equality $2\Gamma ^{2}=g^{ij}g^{2m}\left(
2g_{mi,j}-g_{ij,m}\right) $ by using the equalities
\begin{equation*}
g^{11}=g^{13}=g^{14}=0,\quad g_{22}=g_{23}=g_{24}=0,
\end{equation*}%
assumed on $G^{1}$ to gain
\begin{equation}
\Gamma ^{2}=\left( g^{12}\right)
^{2}g_{11,2}-\frac{1}{2}g^{12}g^{\lambda \mu }g_{\lambda \mu
,1}+S,  \tag{D.7}
\end{equation}%
where $S$\ is known and given on $G^{1}$\ by $\left( \ref{4.35}\right) $. $%
\left( D.7\right) $\ gives%
\begin{align}
\Gamma _{,2}^{2}& =\left( g^{12}\right) ^{2}g_{11,22}-2\left(
g^{12}\right) ^{3}g_{12,2}g_{11,2}+\frac{1}{2}\left( g^{12}\right)
^{2}g_{12,2}g^{\lambda
\mu }g_{\lambda \mu ,1}  \notag \\
& -\frac{1}{2}g^{12}g_{,2}^{\lambda \mu }g_{\lambda \mu ,1}-\frac{1}{2}%
g^{12}g^{\lambda \mu }g_{\lambda \mu ,12}+W,  \tag{D.8}
\end{align}%
where $W$\ is known and given on $G^{1}$\ by $\left( \ref{4.35}\right) $%
\textit{.} $\left( D.6\right) $, $\left( D.7\right) $\textit{\
}and\textit{\ }$\left( D.8\right) $\textit{\ }yield the first
equality of $\left( 3.24\right) $.

We now prove the second equality of\textbf{\ }$\left( 3.24\right) $.\ From $%
\left( 2.22\right) $ we have%
\begin{equation}
g^{\alpha \beta }\tau _{\alpha \beta }=g^{\alpha \beta }F_{\alpha
k}.F_{\beta i}g^{ki}-\frac{1}{2}F_{kl}.F^{kl}+g^{\alpha \beta }\widehat{%
\nabla }_{\alpha }\Phi .\widehat{\nabla }_{\beta }\Phi +V\left(
\Phi ^{2}\right) .  \tag{D.9}
\end{equation}%
It is worth noting at this step of the construction process that,
apart from
$g_{11}$ and $F_{1\alpha }$, all the $g_{ij}$ and $F_{ij}$ are known on $%
G^{1}$. One deduces that, apart from $g^{22}$ and $F^{2\alpha }$,\ all the $%
g^{ik}$ and $F^{ik}$\ are known on $G^{1}$. More precisely, the
following
equalities hold on $G^{1}$.%
\begin{align}
g^{22}& =-\left( g^{12}\right) ^{2}g_{11}-g^{12}g^{2\lambda
}g_{1\lambda },
\notag \\
F^{12}& =\left( g^{12}\right) ^{2}F_{21}+g^{12}g^{2\lambda
}F_{2\lambda },
\notag \\
F^{1\alpha }& =g^{12}g^{\alpha \lambda }F_{2\lambda },  \notag \\
F^{2\alpha }& =g^{21}g^{\alpha 2}F_{12}+g^{21}g^{\alpha \lambda
}F_{1\lambda }+g^{22}g^{\alpha \lambda }F_{2\lambda
}+g^{23}g^{\alpha
2}F_{32}+g^{23}g^{\alpha 4}F_{34}  \notag \\
& +g^{24}g^{\alpha 2}F_{42}+g^{24}g^{\alpha 3}F_{43},  \notag \\
F^{34}& =g^{32}g^{4\lambda }F_{2\lambda
}+g^{33}g^{42}F_{32}+g^{33}g^{44}F_{34}+g^{34}g^{42}F_{42}+g^{34}g^{43}F_{43}.
\tag{D.10}
\end{align}%
We will then examine all the terms of the r.h.s of $\left(
D.9\right) $ in
order to highlight the unknown functions. The following equalities hold on $%
G^{1}$ via direct calculations%
\begin{equation}
g^{\alpha \beta }g^{ki}F_{\alpha k}.F_{\beta i}=2g^{12}g^{\alpha
\beta }F_{1\alpha }.F_{2\beta }+g^{22}g^{\alpha \beta }F_{2\alpha
}.F_{2\beta }+2g^{\alpha \beta }g^{2\lambda }F_{2\alpha
}.F_{\lambda \beta }+g^{\alpha \beta }g^{\mu \lambda }F_{\mu
\alpha }.F_{\lambda \beta }.  \tag{D.11}
\end{equation}%
As the tensor\ $\left( F^{ij}\right) $\ is antisymmetric, from the
above computations, we obtain
\begin{align}
\frac{1}{2}F_{kl}.F^{kl}& =2g^{12}g^{\alpha \lambda }F_{1\alpha
}.F_{2\lambda }+g^{22}g^{\lambda \beta }F_{2\beta }.F_{2\lambda
}+F_{12}.F^{12}+F_{34}.F^{34}  \notag \\
& +F_{2\lambda }.\left[ g^{21}g^{\lambda 2}F_{12}+g^{23}g^{\lambda
2}F_{32}+g^{23}g^{\lambda 4}F_{34}+g^{24}g^{\lambda
2}F_{42}+g^{24}g^{\lambda 3}F_{43}\right] .  \tag{D.12}
\end{align}%
$\left( D.11\right) $ and\ $\left( D.12\right) $\ yield%
\begin{align}
g^{\alpha \beta }g^{ki}F_{\alpha k}.F_{\beta
i}-\frac{1}{2}F_{kl}.F^{kl}& =2g^{\alpha \beta }g^{2\lambda
}F_{2\alpha }.F_{\lambda \beta }+g^{\alpha \beta }g^{\mu \lambda
}F_{\mu \alpha }.F_{\lambda \beta
}-F_{12}.F^{12}-F_{34}.F^{34}  \notag \\
& -F_{2\lambda }.\left[ g^{21}g^{\lambda 2}F_{12}+g^{23}g^{\lambda
2}F_{32}+g^{23}g^{\lambda 4}F_{34}+g^{24}g^{\lambda
2}F_{42}+g^{24}g^{\lambda 3}F_{43}\right] .  \tag{D.13}
\end{align}%
The second equality of $\left( 3.24\right) $\ follows
straightforwardly from $\left( D.9\right) $\ and $\left(
D.13\right) $.

\textbf{Proof of item }$\left( ii\right) $\textbf{.} In view of
$\left( 3.24\right) $, the equation
\begin{equation*}
g^{\alpha \beta }R_{\alpha \beta }-2\Gamma
_{,2}^{2}-2g^{12}g_{12,2}\Gamma ^{2}=g^{\alpha \beta }\tau
_{\alpha \beta },
\end{equation*}%
is equivalent to%
\begin{align}
K& =-2\left( g^{12}\right) ^{2}g_{11,22}+4\left( g^{12}\right)
^{3}g_{12,2}g_{11,2}  \notag \\
& +\left\{ 4\left( g^{12}\right) ^{4}\left( g_{12,2}\right) ^{2}+\frac{1}{2}%
\left( g^{12}\right) ^{2}\left( g^{\alpha \beta }g_{\alpha \beta
,2}\right)
_{,2}\right\} g_{11}  \notag \\
& +\frac{1}{4}g^{\alpha \beta }\left( N_{\alpha \beta }+M_{\alpha
\beta }\right) -2W-2g^{12}g_{12,2}S.  \tag{D.14}
\end{align}%
$\left( D.14\right) $ is arranged under the simplified form $\left( \ref%
{4.37}\right) $ with $\chi $ and $\psi $ given by $\left( \ref{4.38}\right) $%
.\medskip

\textbf{Acknowledgement.}\textit{\ I wish to acknowledge warm
welcome and hospitality from Professor Mamadou Sango at University
of Pretoria where this work was finalized at the beginning of my
postdoctoral fellowship. I would also like to express my
thankfulness to Professor Marcel Dossa for his constant support
and advice.}


\begin{thebibliography}{99}
\bibitem{1} A. Balakin, H. Dehnen, and A.\ E.\ Zayats, Effective metrics in
the non-minimal Einstein-Yang-Mills-Higgs theory,\ \textit{Ann.
Phys.} \textbf{323} (2008) 2183-2207.

\bibitem{2} A.\ Cabet, Local existence of a solution of a semilinear wave
equation with Gradient in a neighborhood of Initial Characteristic
Hypersurfaces of a Lorentzian Manifold,\ \textit{Commun. Part.
Diff. Eq.} \textbf{33}, (2008) 2105-2156.

\bibitem{3} G. Caciotta and F. Nicolo, Global characteristic problem for
Einstein vacuum equations with small initial data: (I) The initial
constraints,\ \textit{JHDE} \textbf{2} (1) (2005) 201-277.

\bibitem{4} F. Cagnac, Probl\`{e}me de Cauchy sur un cono\"{\i}de caract\'{e}%
ristique pour des \'{e}quations quasi-lin\'{e}aires,\ \textit{Ann.
Mat. Pura ed Applicata} \textbf{IV} (CXXIX), 1980\ 13-41.

\bibitem{5} F. Cagnac, Probl\`{e}me de Cauchy sur un cono\"{\i}de caract\'{e}%
ristique,\ \textit{Ann. Fac. Sci. Toulouse} \textbf{11}, (1980)
11-19.

\bibitem{6} F. Cagnac and M. Dossa, Probl\`{e}me de Cauchy sur un cono\"{\i}%
de caract\'{e}ristique. Applications \`{a} certains syst\`{e}mes non lin\'{e}%
aires d'origine physique,\ \textit{Physics on Manifolds},
Proceedings of the International Colloquium in honour of Yvonne
Choquet-Bruhat Paris, june 3-5, 1992 (35-47), edited by Flato,
Kerner, \textit{Lichnerovicz Mathematical Physics Studies}
\textbf{15} (1994), Kluwer Academic Publishers.

\bibitem{7} Y. Choquet-Bruhat, Th\'{e}or\`{e}me d'existence pour certains
syst\`{e}mes d'\'{e}quations aux d\'{e}riv\'{e}es partielles non lin\'{e}%
aires,\ \textit{Acta Math.} \textbf{88} (1952) 141--225.

\bibitem{8} Y. Choquet-Bruhat, Yang-Mills-Higgs fields in three space-time
dimensions,\ \textit{M\'{e}m. Soc. Math. France}, N. S.
\textbf{46} (1991) 73-97; \textit{Analyse globale et physique
math\'{e}matique}, Lyon, (1989).

\bibitem{9} Y. Choquet-Bruhat and D. Christodoulou, Existence of global
solutions of the Yang-Mills, Higgs and spinor field equations in
$3+1$ dimensions,\ \textit{Ann. Sci. Ecole Norm. Sup.} \textbf{14}
(1981) 481-506.

\bibitem{10} D. Christodoulou, The Formation of black holes in General
Relativity,\ \textit{Arxiv: gr-qc/0805.3880v1},\ 594 pp.,\ May\
2008, \textit{EMIS},\ 2009.

\bibitem{11} D. Christodoulou and H. M\"{u}ller zum Hagen, Probl\`{e}me de
valeur initiale caract\'{e}ristique pour des syst\`{e}mes quasi lin\'{e}%
aires du second ordre,\ \textit{C. R. Acad.Sci. Paris}, Ser.
I\textbf{\ 293} (1981) 39-42.

\bibitem{12} T. Damour and B.\ G.\ Schmidt, Reliability of perturbation
theory in General Relativity,\ \textit{J. Math. Phys.} \textbf{31}
(1990)2241-2453.

\bibitem{13} M. Dossa, Espaces de Sobolev non isotropes \`{a} poids et probl%
\`{e}mes de Cauchy quasi-lin\'{e}aires sur un cono\"{\i}de caract\'{e}%
ristique,\ \textit{Ann. Inst. Henri Poincar\'{e}, Phys. Th\'{e}o.} \textbf{66%
} (1) (1997) 37-107.

\bibitem{14} M. Dossa and C. Tadmon, The Goursat problem for the
Einstein-Yang-Mills-Higgs system in weighted Sobolev spaces,\
\textit{C. R. Acad. Sci. Paris}, S\'{e}rie I \textbf{348} (2010)\
35-39.

\bibitem{15} M. Dossa and C. Tadmon, The characteristic initial value
problem for the Einstein-Yang-Mills-Higgs system in weighted
Sobolev spaces,\ \textit{Applied Math. Research eXpress}
\textbf{2010}, (2), (2010) 154-231.

\bibitem{16} G.\ F.\ R. Ellis et al., Ideal Observational Cosmology,\
\textit{Physics Reports} \textbf{124}, (1985) 315-417.

\bibitem{17} F.\ G.\ Friedlander, On the radiation field of pulse solutions
of the wave equations,\ \textit{Proc. R. Soc.} \textbf{A 269}
(1962) 53-65. ibid., \textbf{279} (1964) 386-394. ibid.,
\textbf{299}, 264-278.

\bibitem{18} S.\ W.\ Hawking\ and G.\ F.\ R.\ Ellis, The large scale
structure of space-time,\ \textit{Cambridge University Press},
1973.

\bibitem{19} D.\ E.\ Houpa and M.\ Dossa, Probl\`{e}mes de Goursat pour les
syst\`{e}mes semi-lin\'{e}aires hyperboliques,\ \textit{C. R.\
Acad. Sci. Paris}, Ser. I \textbf{341} (2005) 15-20.

\bibitem{20} J. Kannar, On the existence of $C^{\infty }$ solution to the
asymptotic characteristic initial value problem in General
Relativity,\ \textit{Proc. R. Soc. Lond.}\ \textbf{A452} (1996)
945-952.

\bibitem{21} H.\ Lindblad and I. Rodnianski, Global existence for the
Einstein vacuum equations in wave coordinates,\ \textit{Commun.
Math. Phys.} 256, 43-110, 2005.

\bibitem{22} H. M\"{u}ller zum Hagen and F.\ H.\ J\"{u}rgen Seifert, On
characteristic Initial-Value and Mixed Problems,\ \textit{Gen.
Rel. Grav.} \textbf{8 }(1977) 259-301.

\bibitem{23} H. M\"{u}ller zum Hagen, Characteristic initial value problem
for hyperbolic systems of second order differential equations,\
\textit{Ann. Inst. Henri Poincar\'{e}, Phys. Th\'{e}o.}
\textbf{53} (1990) 159-216.

\bibitem{24} A. D. Rendall, Reduction of the characteristic initial value
problem to the Cauchy problem and its applications to the Einstein
equations,\ \textit{Proc. R. Soc. Lond.} \textbf{A 427} (1990)
221-239.

\bibitem{25} A.\ D.\ Rendall, The characteristic initial value problem for
the Einstein Equations, Non linear hyperbolic equations and field
theory\
\textit{(Lake Como 1991)} \textit{Pitman, Res. Notes, Maths-ser.} \textbf{253%
} \textit{Longman Sci. Tech. Harlow} (1992)154-163.

\bibitem{26} J.\ M.\ Stewart, Classical General Relativity,\ ed. W. B.\
Bonnor, \textit{J. N. Islam and M. A. H.\ Mac Callum}, Cambridge
University Press, 1984.

\bibitem{27} J.\ M.\ Stewart and H.\ Friedrich, Numerical Relativity I. The
Initial Value Problem,\ \textit{Proc. R. Soc. Lond.} \textbf{A
384}, (1982) 427-454.

\bibitem{28} C.\ Svedberg, Future Stability of the Einstein-Maxwell-Scalar
Field System, \textit{Ann. Henri Poincar\'{e}} \textbf{12} (2011)
849--917

\bibitem{29} M.\ S.\ Volkov and\ D.\ V.\ Gal'tsov, Gravitating non-abelian
solitons and black holes with Yang-Mills fields,\ \textit{Phys.
Rep.} \textbf{319} (1999) 1-83.
\end{thebibliography}
\end{document}